\documentclass[a4paper,UKenglish]{lmcs}
\pdfoutput=1

\usepackage{lastpage}
\lmcsdoi{15}{1}{3}
\lmcsheading{}{\pageref{LastPage}}{}{}%
{Jul.~19,~2016}{Jan.~09,~2018}{}

\usepackage{microtype}

\usepackage{amssymb,amsthm,amsmath}
\usepackage{latexsym}
\usepackage{bussproofs}
\usepackage{cmll}
\usepackage{stmaryrd}
\usepackage{bbold}
\usepackage{esvect}
\hypersetup{
  unicode,
  hyperindex
}
\usepackage{array}
\usepackage{alltt}

\usepackage{url}

\usepackage{lineno}
\usepackage{tikz}
\usetikzlibrary{matrix}

\theoremstyle{plain}
\newtheorem{theorem}{Theorem}
\newtheorem{lemma}[theorem]{Lemma}
\newtheorem{corollary}[theorem]{Corollary}
\theoremstyle{definition}
\newtheorem{definition}[theorem]{Definition}

\theoremstyle{remark}
\newtheorem*{remark}{Remark}

\newcommand{\proofitem}[1]{\paragraph*{\mdseries\textit{#1}}}

\newcommand{\Beginproof}{\proofitem{Proof.}}
\newcommand{\Endproof}{
  \ifmmode 
  \else \leavevmode\unskip\penalty9999 \hbox{}\nobreak\hfill
  \fi
  \quad\hbox{$\Box$}
  \par\medskip}

\newcommand\Eqref[1]{(\ref{#1})}

\newcommand\Eg{\textsl{e.g.}~}

\renewcommand{\phi}{\varphi}
\renewcommand\epsilon{\varepsilon}

\newcommand{\Implies}{\Rightarrow}

\newcommand{\St}{\mid}

\renewcommand{\Bot}{{\mathord{\perp}}}
\newcommand{\Top}{\top}

\newcommand\cF{\mathcal{F}}

\newcommand\cL{\mathcal{L}}

\newcommand\cP{\mathcal{P}}

\newcommand\cR{\mathcal{R}}
\newcommand\cS{\mathcal{S}}

\newcommand\cV{\mathcal{V}}
\newcommand\cW{\mathcal{W}}

\newcommand\Fini{{\mathrm{fin}}}

\newcommand\set[1]{\{#1\}}

\newcommand\Union{\bigcup}

\newcommand{\Linarrow}{\multimap}

\newcommand\Myleft{}
\newcommand\Myright{}

\newcommand\Web[1]{\Myleft|{#1}\Myright|}

\newcommand\Supp[1]{\operatorname{\mathsf{supp}}({#1})}

\newcommand\Mset[1]{[{#1}]}

\newcommand\ITens{\otimes}
\newcommand\Tens[2]{{#1}\ITens{#2}}
\newcommand\Tensp[2]{\left({#1}\ITens{#2}\right)}
\newcommand\IWith{\mathrel{\&}}
\newcommand\With[2]{{#1}\IWith{#2}}
\newcommand\IPlus{\oplus}
\newcommand\Plus[2]{{#1}\IPlus{#2}}
\newcommand\Orth[2][]{#2^{\Bot_{#1}}}
\newcommand\Orthp[2][]{(#2)^{\Bot_{#1}}}

\newcommand\Bwith{\mathop{\&}}
\newcommand\Bplus{\mathop\oplus}

\newcommand\Biorth[1]{#1^{\Bot\Bot}}

\newcommand\Triorth[1]{{#1}^{\Bot\Bot\Bot}}

\newcommand\One{1}

\newcommand\IExcl{{\mathord{!}}}

\newcommand\Card[1]{\#{#1}}

\newcommand\Locun[1]{1^J}

\newcommand\Isom\simeq

\newcommand\Comp{\mathrel\circ}

\newcommand\Limpl[2]{{#1}\Linarrow{#2}}

\newcommand\Injtr[1]{{\mathsf{in}}^{#1}}

\newcommand\Derel[1]{\mathsf{der}_{#1}}

\newcommand\Nat{{\mathbb{N}}}

\newcommand\Natnz{{\Nat^+}}

\newcommand\SNat{\mathsf N}
\newcommand\Snat{\mathsf N}
\newcommand\Snum[1]{\overline{#1}}

\newcommand\Sfix{{\mathrm{fix}}}

\newcommand\BOOL{o}

\newcommand\True{\mathbf t}
\newcommand\False{\mathbf f}

\newcommand\Zero{0}
\newcommand\App[2]{\left({#1}\right){#2}}

\newcommand\Abs[2]{\lambda{#1}\,{#2}}
\newcommand\Abst[3]{\lambda{#1}^{#2}\,{#3}}

\newcommand\List[3]{#1_{#2},\dots,#1_{#3}}

\newcommand\Listbis[3]{#1_{#2}\dots #1_{#3}}

\newcommand\Kronecker[2]{\delta_{{#1},{#2}}}

\newcommand\Subst[3]{{#1}\left[{#2}/{#3}\right]}

\newcommand\Substbis[2]{{#1}\left[{#2}\right]}

\newcommand\Factor[1]{{#1}!}

\newcommand\Multinom[2]{\mathsf{mn}(#2)}

\newcommand\Real{\mathbb{R}}
\newcommand\Realp{\mathbb{R}^+}
\newcommand\Realpto[1]{(\Realp)^{#1}}

\newcommand\Realpc{\overline{\mathbb{R}^+}}
\newcommand\Realpcto[1]{\Realpc^{#1}}
\newcommand\Izu{[0,1]}
\newcommand\Rational{\mathbb Q}

\newcommand\Bcanon[1]{e_{#1}}

\newcommand\Mfin[1]{\mathcal M_\Fini({#1})}
\newcommand\Mfinc[2]{\mathcal M_{#1}({#2})}
\newcommand\Evlin{\operatorname{\mathrm{ev}}}

\newcommand\Norm[1]{\|{#1}\|}

\newcommand\Rel[1]{\mathrel{#1}}

\newcommand\Wred{\mathord\to_{\mathsf w}}

\newcommand\Wredtr{\Wred^*}
\newcommand\Redst[1]{\mathop{\mathsf{Red}}}

\newcommand\Msetofsubst[1]{\bar F}

\newcommand\PcohEM[1]{\EM{\mathsf P}{#1}}

\newcommand\Pcoh[1]{\mathsf P{#1}}
\newcommand\Pcohp[1]{\mathsf P(#1)}

\newcommand\Base[1]{e_{#1}}
\newcommand\Matapp[2]{{#1}\Compl{#2}}

\newcommand\PCOH{\mathbf{Pcoh}}

\newcommand\Leftu{\lambda}
\newcommand\Rightu{\rho}
\newcommand\Assoc{\alpha}
\newcommand\Sym{\sigma}
\newcommand\Msetempty{\Mset{\,}}

\newcommand\Retri\zeta
\newcommand\Retrp\rho

\newcommand\Impl[2]{{#1}\Rightarrow{#2}}

\newcommand\Tsem[1]{[{#1}]}

\newcommand\Tsemca[1]{[{#1}]^\oc}
\newcommand\Psem[1]{\left[{#1}\right]}

\newcommand\Tnat\iota

\newcommand\Num[1]{\underline{#1}}
\newcommand\Loop\Omega
\newcommand\Loopt[1]{\Omega^{#1}}
\newcommand\Ran[1]{\mathsf{ran}(#1)}

\newcommand\Tseq[3]{{#1}\vdash{#2}:{#3}}

\newcommand\Timpl\Impl
\newcommand\Simpl\Impl

\newcommand\PCF{\mathsf{PCF}}

\newcommand\Redone[1]{\stackrel{#1}\rightarrow}
\newcommand\Redonetr[1]{\stackrel{#1}{\rightarrow^*}}

\newcommand\Weak[1]{\operatorname{\mathrm{w}}_{#1}}
\newcommand\Contr[1]{\operatorname{\mathrm{c}}_{#1}}

\newcommand\Der[1]{\operatorname{{\mathrm{der}}}_{#1}}

\newcommand\Digg[1]{\operatorname{{\mathrm{dig}}}_{#1}}

\newcommand\Fun[1]{\widehat{#1}}

\newcommand\Id{\operatorname{\mathrm{Id}}}

\newcommand\Proj[1]{{\mathrm{pr}}_{#1}}
\newcommand\Inj[1]{{\mathrm{in}}_{#1}}

\newcommand\Excl[1]{\oc{#1}}

\newcommand\Exclp[1]{\oc({#1})}
\newcommand\Int[1]{\wn{#1}}

\newcommand\Prom[1]{#1^!}
\newcommand\Promp[1]{\left(#1\right)^!}

\newcommand\Relincl\eta
\newcommand\Relrestr\rho

\newcommand\Seely{\mathrm m}
\newcommand\Seelyz{\Seely^0}
\newcommand\Seelyt{\Seely^2}

\newcommand\Compl{\,}

\newcommand\Curlin{\operatorname{\mathrm{cur}}}

\newcommand\Op[1]{{#1}^{\mathsf{op}}}
\newcommand\Kl[1]{{#1}_\oc}

\newcommand\NUM[1]{\underline{#1}}

\newcommand\FIXT[3]{\mathsf{fix}\,#1^{\oc#2}\,#3}

\newcommand\FIXTP[3]{\mathsf{fix}\,#1^{\oc(#2)}\,#3}

\newcommand\LAPP[2]{\langle#1\rangle#2}
\newcommand{\EAPP}[2]{\LAPP{#1}{\STOP{#2}}}
\newcommand\ABST[3]{\lambda #1^{#2}\, #3}

\newcommand\Abstpref[2]{\lambda #1^{#2}\,}

\newcommand\IMPL[2]{#1\Rightarrow #2}

\newcommand\Transcl[1]{{#1}^*}

\newcommand\Eval[2]{\langle#1,#2\rangle}

\newcommand\Redmats{\mathsf{Red}}

\newcommand\Mexpset[2]{\mathsf L(#1,#2)}

\newcommand\Expmonisoz{\Seely^0}
\newcommand\Expmonisob[2]{\Seely^2_{#1,#2}}

\newcommand\Injms[2]{#1\cdot#2}

\newcommand\Obseq{\sim}

\newcommand\Probe{\mathsf{eq}}
\newcommand\Pprod{\mathsf{prod}}
\newcommand\Pchoose{\mathsf{choose}}
\newcommand\Pext[2]{\mathsf{ext}\left(#1,\,#2\right)}
\newcommand\Pwin[2]{\mathsf{win}_{#1}({#2})}

\newcommand\Rseg[2]{[#1,#2]}

\newcommand\Bnfeq{\mathrel{\mathord:\mathord=}}
\newcommand\Bnfor{\,\,\mathord|\,\,}

\newcommand\NAT{\iota}

\newcommand\STOP[1]{#1^{\mathord\oc}}
\newcommand\GO[1]{\mathsf{der}\,#1}

\newcommand\LOOP[1]{\Omega^{#1}}

\newcommand\CALLCC{\mathsf{call/cc}}

\newcommand\Coalgc[1]{\underline{#1}}
\newcommand\Coalgm[1]{\mathrm h_{#1}}
\newcommand\EM[1]{#1^{\mathord\oc}}
\newcommand\EMR[1]{\EM{#1}_{\mathsf{den}}}

\newcommand\Obj[1]{\mathsf{Obj}(#1)}

\newcommand\PAIR[2]{\left(#1,#2\right)}
\newcommand\TUPLE[1]{(#1)}

\newcommand\PR[2]{\mathsf{pr}_{#1}#2}
\newcommand\IN[2]{\mathsf{in}_{#1}#2}
\newcommand\CASE[5]{\mathsf{case}(#1,#2\cdot #3,#4\cdot #5)}

\newcommand\CBPV{\HPLAM}

\newcommand\pCBPV{\HPLAM^{\mathsf{p}}}

\newcommand\ONE{\mathbf 1}
\newcommand\ZERO{\mathbf 0}
\newcommand\EXCL[1]{\oc#1}
\newcommand\TENS[2]{\Tens{#1}{#2}}
\newcommand\PLUS[2]{\Plus{#1}{#2}}
\newcommand\TREC[2]{\textsf{Rec}\,#1\cdot#2}

\newcommand\FORG[1]{#1}

\newcommand\LIMPL[2]{\Limpl{#1}{#2}}

\newcommand\TSUCC[1]{\mathsf{suc}(#1)}

\newcommand\IFB[3]{\mathsf{if}(#1,#2,#3)}
\newcommand\IFV[4]{\mathsf{ifz}(#1,#2,#3\cdot#4)}

\newcommand\TSEQ[3]{#1\vdash#2:#3}

\newcommand\ONELEM{()}

\newcommand\STREAM[1]{\mathsf S_{#1}}

\newcommand\Forgca{\mathsf U}

\newcommand\Embr[1]{#1_{\mathord\subseteq}}
\newcommand\Emb[1]{#1^+}
\newcommand\Ret[1]{#1^-}
\newcommand\Embf{\mathsf E}

\newcommand\Vect[1]{\vv{#1}}

\newcommand\DICE[3]{\mathsf{dice}_{#1}(#2,#3)}

\newcommand\Funofemb{\mathsf{F}}

\newcommand\HPLAM{\Lambda_{\mathsf{HP}}}

\newcommand\Testv[1]{#1^0}
\newcommand\Testt[1]{#1^-}
\newcommand\Testa[1]{#1^+}
\newcommand\Lenv[1]{|#1|^0}
\newcommand\Lent[1]{|#1|^-}
\newcommand\Lena[1]{|#1|^+}
\newcommand\mon[2]{\mathbb c^\mathbb{1}_{#1}\left(#2\right)}
\newcommand\coeff[2]{\mathbb m^{#1}({#2})}
\newcommand\coeffv[1]{\coeff 0{#1}}
\newcommand\coefft[1]{\coeff -{#1}}
\newcommand\coeffa[1]{\coeff +{#1}}

\newcommand\Appsep{\,}

\newcommand\AND{\wedge}

\newcommand\Projt[1]{\Proj{#1}^{\mathord\ITens}}

\newcommand\Rels[1]{\mathsf{Rel}(#1)}
\newcommand\Relsv[1]{\mathsf{Rel}^{\mathsf v}(#1)}
\newcommand\Nrel[1]{{#1}^{\mathord -}}
\newcommand\Prel[1]{{#1}^{\mathord +}}
\newcommand\Srel[2]{#1^{#2}}

\newcommand\POS{\mathord +}
\newcommand\NEG{\mathord -}

\newcommand\Subrel{\sqsubseteq}

\newcommand\Infrel{\bigsqcap}

\newcommand\Crel[1]{\overline{#1}}

\newcommand\Andc{\mathrel\wedge}

\newcommand\Trel[2]{\cR(#1)_{#2}}
\newcommand\Trelv[2]{\cV(#1)_{#2}}

\newcommand\FOLD[1]{\mathsf{fold}(#1)}
\newcommand\UNFOLD[1]{\mathsf{unfold}(#1)}

\newcommand\Vsymb{\mathsf v}
\newcommand\Gsymb{\mathsf g}

\newcommand\Restr[2]{{\mathsf p}^{\Gsymb}(#1,#2)}
\newcommand\Restrv[2]{{\mathsf p}^{\Vsymb}(#1,#2)}
\newcommand\Srestr[2]{{\mathsf I}^\Gsymb(#1,#2)}
\newcommand\Srestrv[2]{{\mathsf I}^\Vsymb(#1,#2)}

\newcommand\Srestrp[3]{{\mathsf I}^{#3}(#1,#2)}

\newcommand\Height[1]{\mathsf{h}(#1)}

\newcommand\Onelem{*}

\newcommand\Case{\mathrm{case}}

\newcommand\COIN[1]{\mathsf{coin}(#1)}

\newcommand\Cnat{\mathsf{nat}}
\newcommand\Cpair[2]{(#1,#2)}
\newcommand\Cproj[2]{\mathsf{pr}_{#1}\,#2}
\newcommand\Cifz[3]{\mathsf{ifz}(#1,#2,#3)}
\newcommand\Csuc[1]{\mathsf{suc}\,#1}
\newcommand\Cpred[1]{\mathsf{pred}\,#1}
\newcommand\Cfix[3]{\mathsf{fix}\,#1^{#2}\cdot#3}
\newcommand\Cprod[2]{#1\times#2}

\newcommand\Ctype[1]{{#1}^*} 
\newcommand\Cterm[1]{{#1}^*} 
\newcommand\Clen[1]{\mathsf{l}(#1)}
\newcommand\Ccat[2]{#1\cdot#2}

\begin{document}


\title[Probabilistic CBPV]{Probabilistic Call By Push Value}

\author[T. Ehrhard and C. Tasson]{Thomas Ehrhard} 
\address{{CNRS, IRIF, UMR 8243,
    Univ Paris Diderot\\
    Sorbonne Paris Cit\'e, F-75205 Paris, France}} 
\email{\{thomas.ehrhard,christine.tasson\}@irif.fr} 

\author[]{Christine Tasson} 


\keywords{MANDATORY list of keywords} \subjclass{MANDATORY list of acm
  classifications}


\begin{abstract}
  We introduce a probabilistic extension of Levy's
  Call-By-Push-Value. This extension consists simply in adding a
  ``flipping coin'' boolean closed atomic expression. This language
  can be understood as a major generalization of Scott's PCF
  encompassing both call-by-name and call-by-value and featuring
  recursive (possibly lazy) data types.  We interpret the language in
  the previously introduced denotational model of probabilistic
  coherence spaces, a categorical model of full classical Linear
  Logic, interpreting data types as coalgebras for the resource
  comonad. We prove adequacy and full abstraction, generalizing
  earlier results to a much more realistic and powerful programming
  language.
\end{abstract}

\maketitle

\section{Introduction}

\emph{Call-by-Push-Value}~\cite{LevyP04} is a class of functional languages
generalizing the lambda-calculus in several directions. From the point of view
of Linear Logic we understand it as a half-polarized system bearing some
similarities with \Eg{}classical Parigot's lambda-mu-calculus, this is why we
call it $\CBPV{}$. The main idea of Laurent and Regnier interpretation of
call-by-name lambda-mu in Linear Logic~\cite{LaurentRegnier03} (following
actually~\cite{Girard91a}) is that all types of the minimal fragment of the
propositional calculus (with $\Rightarrow$ as unique connective) are
interpreted as \emph{negative} types of Linear Logic which are therefore
naturally equipped with structural morphisms: technically speaking, these types
are algebras of the $\wn$-monad of Linear Logic. This additional structure of
negative types allows to perform logical structural rules on the \emph{right
  side} of typing judgments even if these formulas are not necessarily of shape
$\wn\sigma$, and this is the key towards giving a computational content to
classical logical rules, generalizing the fundamental discovery of Griffin on
typing $\CALLCC{}$ with Peirce Law~\cite{Griffin90}.

From our point of view, the basic idea of $\CBPV{}$ is quite similar, though
somehow dual and used in a less systematic way: data types are interpreted as
\emph{positive} types of Linear Logic equipped therefore with structural
morphisms (as linear duals of negative formulas, they are coalgebras of the
$\oc$-comonad) and admit therefore structural rules on the \emph{left side} of
typing judgment even if they are not of shape $\EXCL\sigma$. This means that a
function defined on a data type can have a \emph{linear function type} even if
it uses its argument in a non-linear way: this non-linearity is automatically
implemented by means of the structural morphisms the positive data type is
equipped with.

The basic positive type in Linear Logic is $\EXCL\sigma$ (where $\sigma$ is any
type): it is the very idea of Girard's call-by-name translation of the
lambda-calculus into Linear Logic to represent the implication type
$\sigma\Rightarrow\tau$ by means of the decomposition
$\LIMPL{\EXCL\sigma}\tau$. The new idea of $\CBPV{}$ is to generalize this use
of the linear implication to any type construction of shape
$\LIMPL{\phi}{\tau}$ where $\phi$ is a positive type, without imposing any
linearity restriction on the usage of the argument of type $\phi$ used by a
function of type $\LIMPL{\phi}{\tau}$. This non-symmetrical restriction in the
use of the linear implication motivates our description of $\CBPV$ as a
``half-polarized'' system: in a fully polarized system like Laurent's
\emph{Polarized Linear Logic} LLP~\cite{LaurentRegnier03}, one would also
require the type $\sigma$ to be negative in $\LIMPL\phi\sigma$ (the last system
presented in~\cite{Ehrhard16a} implements this idea)
and the resulting formalism would host classical computational primitives such
as $\CALLCC$ as well. The price to pay, as illustrated
in~\cite{AminiEhrhard15}, is a less direct access to data types: it is
impossible to give a function from integers to integers the expected type
$\LIMPL\Tnat\Tnat$ (where $\Tnat$ is the type of flat natural numbers
satisfying $\Tnat=\PLUS\ONE\Tnat$), the simplest type one can give to such a
term is $\LIMPL\Tnat{\wn\Tnat}$ which complicates its denotational
interpretation\footnote{One can also consider $\wn$ as the computational monad
  of linear continuations and use a translation from direct style into monadic
  style (which, for this monad, is just a version of the familiar CPS
  translation). This is just a matter of presentation and of syntactic sugar
  and does not change the denotational interpretation in the kind of concrete
  models of Linear Logic we have in mind such as the relational model, the
  coherence space model etc.}.

Not being polarized on the right side of implications, $\CBPV$ remains
``intuitionistic'' just as standard functional programming languages whose
paradigmatic example is $\PCF$. So what is the benefit of this
special status given to positive formulas considered as ``data types''? There
are several answers to this question.
\begin{itemize}
\item First, and most importantly, it gives a \emph{call-by-value
    access} to data types: when translating $\PCF$ into Linear Logic,
  the simplest type for a function from integers to integers is
  $\LIMPL{\EXCL\Tnat}{\Tnat}$. This means that arguments of type
  $\Tnat$ are used in a call-by-name way: such arguments are evaluated
  again each time they are used. This can of course be quite
  inefficient. It is also simply \emph{wrong} if we extend our
  language with a random integer generator since in that case each
  evaluation of such an argument can lead to a different value: in
  $\PCF$ there is no way to keep memory of the value obtained for one
  evaluation of such a parameter and probabilistic programming is
  therefore impossible. In $\CBPV$ data types such as $\Tnat$ can be
  accessed in call-by-value, meaning that they are evaluated once and
  that the resulting value is kept for further computation: this is
  typically the behavior of a function of type
  $\LIMPL\Tnat\Tnat$. This is not compulsory however and an explicit
  $\oc$ type constructor still allows to define functions of type
  $\LIMPL{\EXCL\Tnat}{\Tnat}$ in $\CBPV$, with the usual $\PCF$
  behavior.  
\item Positive types being closed under positive Linear Logic
  connectives (direct sums and tensor product) and under ``least
  fixpoint'' constructions, it is natural to allow corresponding
  constructions of positive types in $\CBPV$ as well, leading to a
  language with rich data type constructions (various kinds of trees,
  streams etc are freely available) and can be accessed in
  call-by-value as explained above for integers. From this data types
  point of view, the $\oc$ Linear Logic connective corresponds to the
  type of \emph{suspensions} or \emph{thunks} which are boxes (in the
  usual Linear Logic sense) containing unevaluated pieces of program.
\item   As already mentioned,
  since the Linear Logic type constructors $\LIMPL{}{}$ and $\EXCL{}$
  are available in $\CBPV$ (with the restriction explained above on
  the use of $\LIMPL{}{}$ that the left side type must be positive),
  one can represent in $\CBPV$ both Girard's translations from
  lambda-calculus into Linear Logic introduced in~\cite{Girard87}: the
  usual one which is call-by-name and the ``boring one'' which is
  call-by-value.
  So our language $\CBPV$ is not intrinsically call-by-value and hosts
  very simply Girard's representation of call-by-name in Linear Logic
  as further explained in Section~\ref{sec:PCF-products}.  
\end{itemize}

Concretely, in $\CBPV$, a term of positive type can be a \emph{value}, and then
it is discardable and duplicable and, accordingly, its denotational
interpretation is a morphism of coalgebras: values are particular terms whose
interpretation is easily checked to be such a morphism, which doesn't preclude
other terms of positive type to have the same property of course, in particular
terms reducing to values! Being a value is a property which can be decided in
time at most the size of the term and values are freely duplicable and
discardable. The ``$\beta$-rules'' of the calculus (the standard
$\beta$-reduction as well as the similar reduction rules associated with tensor
product and direct sum) are subject to restrictions on certain subterms of
redexes to be values (because they must be duplicable and discardable) and
these restrictions make sense thanks to this strong decidability property of
being a value.

\subsection*{Probabilities in $\CBPV$.} 
Because of the possibility offered by $\CBPV$ of handling values in a
call-by-value manner, this language is particularly suitable for
probabilistic functional programming. Contrarily to the common monadic
viewpoint on effects, we consider an extension of the language where
probabilistic choice is a primitive $\COIN p$ of type $\PLUS\ONE\ONE$
(the type of booleans)\footnote{And not of type $T(\PLUS\ONE\ONE)$
  where $T$ would be a computational monad of probabilistic
  computations.}  parameterized by $p\in[0,1]\cap\Rational$ which is
the probability of getting $\True$ (and $1-p$ is the probability of
getting $\False$). So our probabilistic extension $\pCBPV$ of $\CBPV$
is in direct style, but, more importantly, the denotational semantics
we consider is itself in ``direct style'' and does not rely on any
auxiliary computational monad of probability distributions~\cite{Saheb-Djahromi80,JonesPlotkin89} (see~\cite{JungTix98} for the
difficulties related with the monadic approach to probabilistic
computations), random variables~\cite{phdBarker,goubaultvarraca11,Mislove16}, or measures~\cite{StatonYWHK16,HeunenKSY17}.

On the contrary, we interpret our language in the model of \emph{probabilistic
  coherence spaces}~\cite{DanosEhrhard08} that we already used for providing a
fully abstract semantics for probabilistic
$\PCF$~\cite{EhrhardPaganiTasson14}. A probabilistic coherence space $X$ is
given by an at most countable set $\Web X$ (the web of $X$) and a set $\Pcoh X$
of $\Web X$-indexed families of non-negative real numbers, to be considered as
some kind of ``generalized probability distributions''. This set of families of
real numbers is subject to a closure property implementing a simple intuition
of probabilistic observations.  Probabilistic coherence spaces are a model of
classical Linear Logic and can be seen as $\omega$-continuous domains equipped
with an operation of convex linear combination, and the linear morphisms of
this model are exactly the Scott continuous functions commuting with these
convex linear combinations.

Besides, probabilistic coherence spaces can be seen as particular
\emph{d-cones}~\cite{TixKP09a} and even
\emph{Kegelspitzen}~\cite{KeimelP16}, that is, complete partial orders
equipped with a Scott continuous ``convex structure'' allowing to
compute probabilistic linear combinations of their elements.
Kegelspitzen have been used recently by Rennela to define a
denotational model of a probabilistic extension of
FPC~\cite{Rennela16}. The main difference with respect to our approach
seems to be the fact that non-linear morphisms (corresponding to
morphisms of type $\Limpl{\Excl\sigma}{\tau}$ in our setting) are
general Scott continuous functions in Rennela's model\footnote{More
  precisely, in his interpretation of FPC, Rennela uses strict Scott
  continuous functions, but, along the same lines, it seems clear that
  Kegelspitzen give rise to a model of PCF where morphisms are general
  Scott continuous functions.}, whereas they are analytic
functions\footnote{Meaning that it is definable by a power series.} in
ours, which can be seen as very specific Scott continuous
functions. See also~\cite{EhrhardPaganiTasson18} where these functions
are seen to be \emph{stable} in a generalized sense. 

As shown in~\cite{DanosEhrhard08} probabilistic coherence spaces have
all the required completeness properties for interpreting recursive
type definitions (that we already used in~\cite{EhrhardPaganiTasson11}
for interpreting the pure lambda-calculus) and so we are able to
associate a probabilistic coherence space with all types of $\pCBPV$.

In this model the type $\PLUS\ONE\ONE$ is interpreted as the set of
sub-probability distributions on $\{\True,\False\}$ so that we have a
straightforward interpretation of $\COIN p$. Similarly the type of
flat integers $\Tnat$ is interpreted as a probabilistic coherence space $\Snat$ 
such that $\Web\Snat=\Nat$ and $\Pcoh\Snat$ is the set of all probability
distributions on the natural numbers. Given probabilistic spaces $X$ and $Y$,
the space $\Limpl XY$ has $\Web X\times\Web Y$ as web and $\Pcoh{(\Limpl XY)}$
is the set of all $\Web X\times\Web Y$ matrices which, when applied to an
element of $\Pcoh X$ gives an element of $\Pcoh Y$. The web of the space $\Excl
X$ is the set of all finite multisets of elements of $\Web X$ so that an
element of $\Limpl{\Excl X}{Y}$ can be considered as a power series on as many
variables as there are elements in $\Web X$ (the composition law associated
with the Kleisli category of the $\oc$-comonad is compatible with this
interpretation of morphisms as power series).

From a syntactic point of view, the only values of $\PLUS\ONE\ONE$ are $\True$
and $\False$, so $\COIN p$ is not a value. Therefore we cannot reduce
$\LAPP{\ABST x{\PLUS\ONE\ONE}{M}}{\COIN p}$ to $\Subst M{\COIN p}{x}$ and this
is a good thing since then we would face the problem that the boolean values of
the various occurrences of $\COIN p$ might be different. We have first to
reduce $\COIN p$ to a value, and the reduction rules of our probabilistic
$\CBPV$ stipulate that $\COIN p$ reduces to $\True$ with probability $p$ and to
$\False$ with probability $1-p$ (in accordance with the kind of operational
semantics that we considered in our earlier work on this topic, starting
with~\cite{DanosEhrhard08}). So $\LAPP{\ABST x{\PLUS\ONE\ONE}{M}}{\COIN p}$
reduces to $\Subst M\True x$ with probability $p$ and to $\Subst M\False x$
with probability $1-p$, which is perfectly compatible with the intuition that
in $\CBPV$ application is a linear operation (and that implication is linear:
the type of $\ABST x{\PLUS\ONE\ONE}{M}$ is $\LIMPL{(\PLUS\ONE\ONE)}{\sigma}$
for some type $\sigma$): in this operational semantics as well as in the
denotational semantics outlined above, linearity corresponds to commutation
with (probabilistic) convex linear combinations.

\subsection*{Contents.}
The results presented in this paper illustrate the tight connection between the
syntactic and the denotational intuitions underpinning our understanding of
this calculus. 

We first introduce in Section~\ref{sec:syntax} the syntax and operational
semantics of $\pCBPV$, an abstract programming language very close to Paul
Levy's Call by Push Value (CBPV)~\cite{LevyP04}. It however differs from Levy's
language mainly by the fact that CBPV computation types products and recursive
type definitions have no counterparts in our language. This choice is mainly
motivated by the wish of keeping the presentation reasonably short. It is
argued in Sections~\ref{subsec:exsyn} and~\ref{sec:PCF-products} that $\pCBPV$
is expressive enough for containing well behaved lazy data types such as the
type of streams, and for encoding call-by-name languages with products.

In Section~\ref{sec:PCS}, we present the Linear Logic model of probabilistic
coherence spaces, introducing mainly the underlying linear category $\PCOH$,
where $\pCBPV$ general types are interpreted, and the Eilenberg-Moore category
$\EM\PCOH$, where the positive types are interpreted. In order to simplify the
Adequacy and Full Abstraction proofs, we restrict actually our attention to a
well-pointed subcategory of $\EM\PCOH$ whose objects we call ``dense
coalgebras'': this will allow to consider all morphisms as functions. As
suggested by one of the referees, there are probably smaller well-pointed
subcategories of $\EM\PCOH$ where we might interpret our positive types, and in
particular the category of families introduced in~\cite{Abramsky1998}
describing call-by-value games. This option will be explored in further
work. We prefer here to work with the most general setting as it is also
compatible with a probabilistic extension of the last system of~\cite{Ehrhard16a}, which features classical $\CALLCC{}$-like capabilities.

We prove then in Section~\ref{sec:adequacy} an Adequacy Theorem whose
statement is extremely simple: given a closed term $M$ of type $\ONE$
(which has exactly one value $\ONELEM$), the denotational semantics of
$M$, which is an element of $[0,1]$, coincides with its probability to
reduce to $\ONELEM$ (such a term can only diverge or reduce to
$\ONELEM$). In spite of its simple statement the proof of this result
requires some efforts mainly because of the presence of
unrestricted\footnote{Meaning that recursive type definitions are not
  restricted to covariant types.} recursive types in $\pCBPV$. The
method used in the proof relies on an idea of Pitts~\cite{Pitts93} and
is described in the introduction of Section~\ref{sec:adequacy}.

Last we prove Full Abstraction in Section~\ref{sec:FA} adapting the
technique used in~\cite{EhrhardPaganiTasson11} to the present $\pCBPV$
setting. The basic idea consists in associating, with any element $a$
of the web of the probabilistic coherence space $\Tsem\sigma$
interpreting the type $\sigma$, a term $\Testt a$ of type\footnote{For
  technical reasons and readability of the proof, the type we give to
  $\Testt a$ in Section~\ref{sec:FA} slightly differs from this one.}
$\LIMPL{\EXCL\sigma}{\LIMPL{\EXCL\Tnat}{\ONE}}$ such that, given two
elements $w$ and $w'$ of $\Pcoh{\Tsem\sigma}$ such that
$w_a\not=w'_a$, the elements $\Psem{\Testt a}{}\Prom w$ and
$\Psem{\Testt a}{}\Prom{(w')}$ of $\Pcoh{(\LIMPL{\EXCL\Tnat}{\ONE})}$
are different power series depending on a finite number $n$ of
parameters (this number $n$ depends actually only on $a$) so that we
can find a rational number valued sub-probability distribution for
these parameters where these power series take different values in
$[0,1]$. Applying this to the case where $w$ and $w'$ are the
interpretations of two closed terms $M$ and $M'$ of type $\sigma$, we
obtain, by combining $\Testt a$ with the rational sub-probability
distribution which can be represented in the syntax using $\COIN p$
for various values of $p$, a $\CBPV$ closed term $C$ of type
$\LIMPL{\EXCL\sigma}{\ONE}$ such that the probability of convergence
of $\LAPP C{\STOP M}$ and $\LAPP C{\STOP{(M')}}$ are different (by
adequacy). This proves that if two (closed) terms are operationally
equivalent then they have the same semantics in probabilistic
coherence spaces, that is, equational full abstraction.

\subsection*{Further developments.}
These results are similar to the ones reported in~\cite{EhrhardPaganiTasson15}
but are actually different, and there is no clear logical connection between
them, because the languages are quite different, and therefore, the observation
contexts also. And this even in spite of the fact that $\PCF$ can be faithfully
encoded in $\CBPV$. This seems to show that the semantical framework for
probabilistic functional programming offered by probabilistic coherence spaces
is very robust and deserves further investigations. One major outcome of the
present work is a natural extension of probabilistic computation to rich
data-types, including types of potentially infinite values (streams etc).

Our full abstraction result cannot be extended to inequational full abstraction
with respect to the natural order relation on the elements of probabilistic
coherence spaces: a natural research direction will be to investigate other
(pre)order relations and their possible interactive definitions. Also, it is
quite tempting to replace the equality of probabilities in the definition of
contextual equivalence by a distance; this clearly requires further
developments.

\section{Probabilistic Call By Push Value}\label{sec:syntax}
We introduce the syntax of $\pCBPV$ of CBPV (where HP stands for
``half polarized'').

Types are given by the following BNF syntax. We define by mutual
induction two kinds of types: \emph{positive types} and \emph{general
  types}, given type variables $\zeta$, $\xi$\dots:
\begin{align}
  \text{positive}\quad \phi,\psi,\dots{} &\Bnfeq \ONE \Bnfor \EXCL
  \sigma \Bnfor \TENS\phi\psi \Bnfor \PLUS\phi\psi \Bnfor \zeta \Bnfor
  \TREC \zeta\phi \label{eq:cbpv-pos-types}
  \\
  \text{general}\quad \sigma,\tau\dots{} &\Bnfeq \FORG\phi \Bnfor
  \LIMPL\phi\sigma
  \label{eq:cbpv-gen-types}
\end{align}
  The type $\ONE$ is the
neutral element of $\ITens$ and it might seem natural to have also a
type $\ZERO$ as the neutral element of $\IPlus$. We didn't do so
because there is no associated constructor in the syntax of terms, and
the only closed terms of type $\ZERO$ that one can write are
ever-looping terms.

Observe also that there are no restrictions on the variance of types in
the recursive type construction: for instance, in $\TREC\zeta\phi$ is
a well-formed positive type if $\phi=\EXCL{(\LIMPL\zeta\zeta)}$, where
$\zeta$ has a negative and a positive occurrence. Do well notice that
our ``positive types'' are positive in the sense of logical
polarities, and \emph{not} of the variance of type variables!

Terms are given by the following BNF syntax, given variables
$x,y,\dots$:
\begin{align*}
  M,N\dots{} \Bnfeq x & \Bnfor \ONELEM \Bnfor \STOP M
  \Bnfor \PAIR MN \Bnfor \IN \ell M \Bnfor \IN r M \\
  & \Bnfor \ABST x\phi M \Bnfor \LAPP MN \Bnfor \CASE M{x_\ell}{N_\ell}{x_r}{N_r}\\
  & \Bnfor \PR \ell M \Bnfor \PR r M \Bnfor \GO M \Bnfor \FIXT x\sigma M \\
  & \Bnfor \FOLD M \Bnfor \UNFOLD M \\
    & \Bnfor \COIN p, \; p\in [0,1]\cap\mathbb Q
\end{align*}
\begin{remark}
  This functional language $\pCBPV$, or rather the sublanguage $\CBPV$
  which is $\pCBPV$ stripped from the $\COIN p$ construct, is formally
  very close to Levy's CBPV. As explained in the Introduction, our
  intuition however is more related to Linear Logic than to CBPV and
  its general adjunction-based models. This explains why our syntax
  slightly departs from Levy's original syntax as described
  \Eg~in~\cite{LevyP02} and is much closer to the SFPL language
  of~\cite{MarzRohrStreicher99}: Levy's type constructor $F$ is kept
  implicit and $U$ is ``$\oc$''. We use LL inspired notations:
  $\STOP M$ corresponds to $\mathsf{thunk}(M)$ and $\GO M$ to
  $\mathsf{force}(M)$. Our syntax is also slightly simpler than that
  of SFPL in that our general types do not feature products and
  recursive types definitions, we will explain in
  Section~\ref{sec:PCF-products} that this is not a serious limitation
  in terms of expressiveness.
\end{remark}

Figure~\ref{fig:typing_system} provides the typing rules for these
terms. A typing context is an expression
$\cP=(x_1:\phi_1,\dots,x_k:\phi_k)$ where all types are positive
and the $x_i$s are pairwise distinct variables.

\begin{figure*}
  \centering
  \begin{center}
    \AxiomC{} \UnaryInfC{$\TSEQ{\cP,x:\phi}{x}{\phi}$}
    \DisplayProof \quad \AxiomC{$\TSEQ{\cP,x:\phi}{M}{\sigma}$}
    \UnaryInfC{$\TSEQ{\cP}{\ABST x\phi M}{\LIMPL\phi\sigma}$}
    \DisplayProof \quad \AxiomC{$\TSEQ{\cP}{M}{\LIMPL\phi\sigma}$}
    \AxiomC{$\TSEQ{\cP}{N}{\phi}$} \BinaryInfC{$\TSEQ{\cP}{\LAPP
        MN}{\sigma}$} \DisplayProof
  \end{center}
\vspace{0.3em}

\begin{center}
  \AxiomC{$\TSEQ\cP M\sigma$} \UnaryInfC{$\TSEQ\cP{\STOP
      M}{\EXCL\sigma}$} \DisplayProof \quad \AxiomC{}
  \UnaryInfC{$\TSEQ{\cP}{\ONELEM}{\ONE}$} \DisplayProof \quad
  \AxiomC{$\TSEQ\cP{M_\ell}{\phi_\ell}$}
  \AxiomC{$\TSEQ\cP{M_r}{\phi_r}$}
  \BinaryInfC{$\TSEQ\cP{\PAIR{M_\ell}{M_r}}{\TENS{\phi_\ell}{\phi_r}}$}
  \DisplayProof \quad \AxiomC{$\TSEQ\cP M{\phi_i}\quad \small i\in\{\ell,r\}$}
  \UnaryInfC{$\TSEQ\cP{\IN iM}{\PLUS{\phi_\ell}{\phi_r}}$}
  \DisplayProof
\end{center}
\vspace{0.3em}

\begin{center}
  \AxiomC{$\TSEQ\cP{M}{\EXCL \sigma}$} 
  \UnaryInfC{$\TSEQ\cP{\GO M}{\sigma}$} 
  \DisplayProof 
  \quad 
  \AxiomC{$\TSEQ\cP
    M{\TENS{\phi_\ell}{\phi_r}\quad \small i\in\{\ell,r\}}$} \UnaryInfC{$\TSEQ\cP{\PR
      iM}{\phi_i}$} \DisplayProof \quad \AxiomC{$\TSEQ{\cP,x:\EXCL
      \sigma}{M}{\sigma}$} \UnaryInfC{$\TSEQ\cP{\FIXT x\sigma
      M}\sigma$} \DisplayProof
\end{center}

\vspace{0.3em}

\begin{center}
  \AxiomC{$\TSEQ\cP M{\PLUS{\phi_\ell}{\phi_r}}$}
  \AxiomC{$\TSEQ{\cP,x_\ell:\phi_\ell}{M_\ell}\sigma$}
  \AxiomC{$\TSEQ{\cP,x_r:\phi_r}{M_r}\sigma$}
  \TrinaryInfC{$\TSEQ\cP{\CASE M{x_\ell}{M_\ell}{x_r}{M_r}}\sigma$}
  \DisplayProof
\end{center}
\vspace{0.8em}
\begin{center}
  \AxiomC{} \UnaryInfC{$\TSEQ{\cP}{\COIN p}{\PLUS\ONE\ONE}$}
  \DisplayProof
\end{center}
\vspace{0.3em}
\begin{center}
  \AxiomC{$\TSEQ{\cP}{M}{\Subst{\psi}{\TREC\zeta\psi}{\zeta}}$}
  \UnaryInfC{$\TSEQ{\cP}{\FOLD M}{\TREC\zeta\psi}$} \DisplayProof
  \quad \AxiomC{$\TSEQ{\cP}{M}{\TREC\zeta\psi}$}
  \UnaryInfC{$\TSEQ{\cP}{\UNFOLD
      M}{\Subst{\psi}{\TREC\zeta\psi}{\zeta}}$} \DisplayProof
\end{center}

\caption{Typing system for $\pCBPV$}
\label{fig:typing_system}
\end{figure*}

\subsection{Reduction rules}

\emph{Values} are particular $\pCBPV$ terms (they are not a new
syntactic category) defined by the following BNF syntax:
\begin{align*}
  V,W\dots{} \Bnfeq x \Bnfor \ONELEM \Bnfor \STOP M \Bnfor \PAIR{V}{W}
  \Bnfor \IN \ell {V} \Bnfor \IN r {V} \Bnfor \FOLD V\,.
\end{align*}

Figure~\ref{fig:reduction-rules} defines a deterministic \emph{weak} reduction
relation $\Wred$ and a probabilistic reduction $\Rel{\Redone p}$ relation.
This reduction is weak in the sense that we never reduce within a ``box'' $\STOP
M$ or under a $\lambda$. 

The distinguishing feature of this reduction system is the role played by
values in the definition of $\Wred$. Consider for instance the case of the term
$\PR \ell {\PAIR{M_\ell}{M_r}}$; one might expect this term to reduce directly to
$M_\ell$ but this is not the case. One needs first to reduce $M_\ell$ \emph{and}
$M_r$ to values before extracting the first component of the pair (the terms
$\PR \ell {\PAIR{M_\ell}{M_r}}$ and $M_\ell$ have not the same denotational
interpretation in general). Of course replacing $M_i$ with $\STOP{M_i}$ allows
a lazy behavior. Similarly, in the $\Wred$ rule for $\mathsf{case}$, the term
on which the test is made must be reduced to a value (necessarily of shape $\IN
\ell V$ or $\IN r V$ if the expression is well typed) before the reduction is performed. As
explained in the Introduction this allows to ``memoize'' the value $V$ for
further usage: the value is passed to the relevant branch of the
$\mathsf{case}$ through the variable $x_i$.

Given two terms $M$, $M'$ and a real number $p\in[0,1]$, $M\Rel{\Redone p} M'$
means that $M$ reduces in one step to $M'$ with probability $p$.

We say that $M$ is \emph{weak normal} if there is no reduction
$M\Rel{\Redone p}M'$. It is clear that any value is weak normal. When $M$ is
closed, $M$ is weak normal iff it is a value or an abstraction.

In order to simplify the presentation we \emph{choose} in
Figure~\ref{fig:reduction-rules} a reduction strategy. For instance we decide
that, for reducing $\PAIR{M_\ell}{M_r}$ to a value, one needs first to reduce
$M_\ell$ to a value, and then $M_r$; this choice is of course completely
arbitrary. A similar choice is made for reducing terms of shape $\LAPP{M}{N}$,
where we require the argument to be reduced first. This choice is less
arbitrary as it will simplify a little bit the proof of adequacy in
Section~\ref{sec:adequacy} (see for instance the proof of
Lemma~\ref{lemma:rel-app-closeness}).

We could perfectly define a more general weak reduction relation as
in~\cite{Ehrhard16a} for which we could prove a ``diamond'' confluence
property but we would then need to deal with a reduction transition
system where, at each node (term), several probability distributions
of reduced terms are available and so we would not be able to describe
reduction as a simple (infinite dimensional) stochastic matrix. We
could certainly also define more general reduction rules allowing to
reduce redexes anywhere in terms (apart for $\COIN p$ which can be
reduced only when in linear position) but this would require the
introduction of additional $\sigma$-rules as
in~\cite{EhrhardGuerrieri16}. As in that paper, confluence can
probably be proven, using ideas coming
from~\cite{EhrhardRegnier02,Vaux08} for dealing with reduction in
an algebraic lambda-calculus setting.
\begin{figure*}
  \centering
  \begin{center}
    \AxiomC{} \UnaryInfC{$\GO{\STOP M}\Rel\Wred M$} \DisplayProof
    \quad \AxiomC{} \UnaryInfC{$\LAPP{\ABST x\phi M}{V}\Rel\Wred\Subst
      MVx$} \DisplayProof \quad \AxiomC{}\RightLabel{$i\in\{\ell,r\}$} \UnaryInfC{$\PR
      i{\PAIR{V_\ell}{V_r}}\Rel\Wred V_i$} \DisplayProof
  \end{center}
  \vspace{0.8em}
  \begin{center}
    \AxiomC{} \UnaryInfC{$\FIXT x\sigma M\Rel\Wred\Subst
      M{\STOP{(\FIXT x\sigma M)}}x$} \DisplayProof \quad \AxiomC{}\RightLabel{$i\in\{\ell, r\}$}
    \UnaryInfC{$\CASE{\IN
        iV}{x_\ell}{M_\ell}{x_r}{M_r}\Rel\Wred\Subst{M_i}{V}{x_i}$}
    \DisplayProof
  \end{center}
  \vspace{0.8em}
  \begin{center}
    \AxiomC{} \UnaryInfC{$\UNFOLD{\FOLD{V}}\Rel{\Wred}V$}
    \DisplayProof
  \end{center}

  \vspace{0.3em}
  \centering
  \begin{center}
    \AxiomC{$M\Rel{\Wred}M'$} \UnaryInfC{$M\Rel{\Redone 1}M'$}
    \DisplayProof \quad \AxiomC{} \UnaryInfC{$\COIN p\Rel{\Redone p}
      \IN \ell \ONELEM$} \DisplayProof \quad \AxiomC{} \UnaryInfC{$\COIN
      p\Rel{\Redone{1-p}} \IN r \ONELEM$} \DisplayProof
  \end{center}

  \vspace{0.3em}

\begin{center}
  \AxiomC{$M\Rel{\Redone p} M'$} \UnaryInfC{$\GO M\Rel{\Redone
      p}\GO{M'}$} \DisplayProof \quad \AxiomC{$M\Rel{\Redone p} M'$}
  \UnaryInfC{$\LAPP MV\Rel{\Redone p}\LAPP{M'}V$} \DisplayProof \quad
  \AxiomC{$N\Rel{\Redone p} N'$} \UnaryInfC{$\LAPP MN\Rel{\Redone
      p}\LAPP{M}{N'}$} \DisplayProof \quad \AxiomC{$M\Rel{\Redone p}
    M'$}\RightLabel{$i\in\{\ell, r\}$} \UnaryInfC{$\PR iM\Rel{\Redone p}\PR i{M'}$} \DisplayProof
\end{center}
\vspace{0.3em}

\begin{center}
  \AxiomC{$M_\ell\Rel{\Redone p} M'_\ell$}
  \UnaryInfC{$\PAIR{M_\ell}{M_r}\Rel{\Redone p}\PAIR{M'_\ell}{M_r}$}
  \DisplayProof \quad \AxiomC{$M_r\Rel{\Redone p} M'_r$}
  \UnaryInfC{$\PAIR{V}{M_r}\Rel{\Redone p}\PAIR{V}{M'_r}$}
  \DisplayProof \quad 
  \AxiomC{$M\Rel{\Redone p} M'$} \RightLabel{$i\in\{\ell, r\}$}\UnaryInfC{$\IN
    iM\Rel{\Redone p}\IN i{M'}$} \DisplayProof
\end{center}

\vspace{0.3em}
\begin{center}
  \AxiomC{$M\Rel{\Redone p} M'$}
  \UnaryInfC{$\CASE{M}{x_\ell}{M_\ell}{x_r}{M_r} \Rel{\Redone
      p}\CASE{M'}{x_\ell}{M_\ell}{x_r}{M_r}$} \DisplayProof
\end{center}

\vspace{0.3em}

\begin{center}
  \AxiomC{$M\Rel{\Redone p} M'$}
  \UnaryInfC{$\FOLD M\Rel{\Redone p}\FOLD{M'}$}
  \DisplayProof
  \quad
  \AxiomC{$M\Rel{\Redone p} M'$}
  \UnaryInfC{$\UNFOLD M\Rel{\Redone p}\UNFOLD{M'}$}
  \DisplayProof
\end{center}

\caption{Weak and Probabilistic reduction axioms and rules for
  $\pCBPV$}
\label{fig:reduction-rules}
\end{figure*}

\subsection{Observational equivalence}\label{sec:obs-eq}
In order to define observational equivalence, we need to represent the
probability of convergence of a term to a normal form. As
in~\cite{DanosEhrhard08}, we consider the reduction as a discrete time
Markov chain whose states are terms and stationary states are weak
normal terms. We then define a stochastic matrix
$\Redmats\in[0,1]^{\pCBPV\times\pCBPV}$
(indexed by terms) as
\begin{equation*}
  \Redmats_{M,M'}=
  \begin{cases}
    p & \text{if }M\Rel{\Redone p}M'\\
    1 & \text{if $M$ is weak-normal and $M'=M$}\\
    0 & \text{otherwise.}
  \end{cases}
\end{equation*}
Saying that $\Redmats$ is stochastic means that the coefficients of
$\Redmats$ belong to $\Rseg 01$ and that, for any given term $M$, one
has $\sum_{M'}\Redmats_{M,M'}=1$ (actually there are at most two terms
$M'$ such that $\Redmats_{M,M'}\not=0$).

For all $M,M'\in\pCBPV$, if $M'$ is weak-normal then the sequence
$(\Redmats^n_{M,M'})_{n=1}^\infty$ is monotone and included in
$[0,1]$, and therefore has a lub that we denote as
$\Redmats^\infty_{M,M'}$ which defines a sub-stochastic matrix (taking
$\Redmats^\infty_{M,M'}=0$ when $M'$ is not weak-normal).

When $M'$ is weak-normal, the number $p=\Redmats^\infty_{M,M'}$ is the
probability that $M$ reduces to $M'$ after a finite number of steps.

Let us say when two closed terms $M_1$, $M_2$ of type $\sigma$ are
\emph{observationally equivalent}:
\begin{center}
  $M_1\Rel\Obseq M_2$, if for all closed term $C$ of type
  $\LIMPL{\EXCL\sigma}{\ONE}$, $ \Redmats^\infty_{\LAPP
    C{\STOP{M_1}},\ONELEM} =\Redmats^\infty_{\LAPP
    C{\STOP{M_1}},\ONELEM}$.
\end{center}
For simplicity we consider only closed terms $M_1$ and $M_2$. We could also
define an observational equivalence on non closed terms, replacing the term $C$
with a context $C[\ ]$ which could bind free variables of the $M_i$'s, this
would not change the results of the paper.

\subsection{Examples}\label{subsec:exsyn}

For the sake of readability, we drop the fold/unfold constructs
associated with recursive types definitions; they can easily be
inserted at the right places.
This also means that, in these examples, we consider the types
$\TREC\zeta\phi$ and $\Subst\phi{\TREC\zeta\phi}\zeta$ as identical.

\subsubsection*{Ever-looping program.}
Given any type $\sigma$, we define $\LOOP\sigma=\FIXT x\sigma{\GO x}$
which satisfies $\TSEQ{}{\LOOP\sigma}{\sigma}$. It is clear that
$\LOOP\sigma\Rel\Wred\GO{\STOP{(\LOOP\sigma)}}\Rel\Wred\LOOP\sigma$ so
that we can consider $\LOOP\sigma$ as the ever-looping program of type
$\sigma$.

\subsubsection*{Booleans.}
We define the type $\BOOL=\PLUS\ONE\ONE$, so that $\TSEQ\cP{\COIN p}\BOOL$. We define the ``true'' constant as
$\True =\IN \ell \ONELEM$ and the ``false'' constant as $\False =\IN r \ONELEM$.
The corresponding eliminator is defined as follows.
Given terms $M$, $N_\ell$ and $N_r$
we set $\IFB M{N_\ell}{N_r}=\CASE
M{x_\ell}{N_\ell}{x_r}{N_r}$ where $x_i$ is not free in $N_i$ for $i\in\{\ell,r\}$, so that
\[
  \AxiomC{$\TSEQ\cP M\BOOL$} \AxiomC{$\TSEQ\cP{N_\ell}\sigma$}
  \AxiomC{$\TSEQ{\cP}{N_r}\sigma$}
  \TrinaryInfC{$\TSEQ{\cP}{\IFB M{N_\ell}{N_r}}\sigma$}
  \DisplayProof
\]
We have the following weak and probabilistic reduction rules, derived
from Figure~\ref{fig:reduction-rules}:
\[
 \AxiomC{}
    \UnaryInfC{$\IFB \True{N_\ell}{N_r}\Rel\Wred{N_\ell}$}
    \DisplayProof\quad \AxiomC{}
    \UnaryInfC{$\IFB \False{N_\ell}{N_r}\Rel\Wred{N_r}$}
    \DisplayProof
\quad
      \AxiomC{$M\Rel{\Redone p} M'$}
  \UnaryInfC{$\IFB M{N_\ell}{N_r} \Rel{\Redone
      p}\IFB {M'}{N_\ell}{N_r}$} \DisplayProof
  \]
\subsubsection*{Natural numbers.}
We define the type $\NAT$ of unary natural numbers by
$\NAT=\PLUS\ONE\NAT$ (by this we mean that
$\NAT=\TREC\zeta{(\PLUS\ONE\zeta)}$). We define $\NUM 0=\IN \ell \ONELEM$
and $\NUM{n+1}=\IN r{\NUM n}$ so that we have $\TSEQ\cP{\NUM
  n}{\NAT}$ for each $n\in\Nat$.

Then, given a term $M$, we define the term $\TSUCC M=\IN r{M}$, so that
we have
\[
  \AxiomC{$\TSEQ\cP M\NAT$} \UnaryInfC{$\TSEQ\cP{\TSUCC M}\NAT$}
  \DisplayProof
\]
  Last, given terms $M$, $N_\ell$ and $N_r$ and a
variable $x$, we define an ``\textsf{ifz}'' conditional by
$\IFV M{N_\ell}{x}{N_r}=\CASE Mz{N_\ell}x{N_r}$ where $z$ is not free
in $N_\ell$, so that
\[
  \AxiomC{$\TSEQ\cP M\NAT$} \AxiomC{$\TSEQ\cP{N_\ell}\sigma$}
  \AxiomC{$\TSEQ{\cP,x:\NAT}{N_r}\sigma$}
  \TrinaryInfC{$\TSEQ{\cP}{\IFV M{N_\ell}{x}{N_r}}\sigma$}
  \DisplayProof
\]
We have the following weak and probabilistic reduction rules, derived
from Figure~\ref{fig:reduction-rules}:
\[
 \AxiomC{}
 \RightLabel{$i\in\{\ell, r\}$}
    \UnaryInfC{$\IFV{\IN
        iV}{M_\ell}{x}{M_r}\Rel\Wred\Subst{M_i}{V}{x}$}
    \DisplayProof\quad
      \AxiomC{$M\Rel{\Redone p} M'$}
  \UnaryInfC{$\IFV M{N_\ell}x{N_r} \Rel{\Redone
      p}\IFV {M'}{N_\ell}x{N_r}$} \DisplayProof
\]
These conditionals will be used in the examples below.

\subsubsection*{Streams.}
Let $\phi$ be a positive type and $\STREAM\phi$ be the positive type
defined by $\STREAM\phi=\EXCL{(\TENS\phi{\STREAM\phi})}$, that is
$\STREAM\phi=\TREC\zeta{\EXCL{(\TENS\phi\zeta)}}$. We can define a
term $M$ such that $\TSEQ{}{M}{\LIMPL{\STREAM\phi}{\LIMPL\NAT \phi}}$
which computes the $n$th element of a stream:
\begin{align*}
  M=\FIXTP f{\LIMPL{\STREAM\phi}{\LIMPL\NAT \phi}}{\ABST
    x{\STREAM\phi}{\ABST y\NAT{}}{}}\IFV y{\PR \ell {(\GO x)}}{z}{\LAPP{\GO
      f}{\PR r {(\GO x)}}{\,z}}
\end{align*}
Let $O=\STOP{(\LOOP{\TENS\phi{\STREAM\phi}})}$, a term which represents the
``undefined stream'' (more precisely, it is a stream which is a value, but
contains nothing, not to be confused with $\LOOP{\STREAM\phi}$ which has the
same type but is not a value). We have $\Tseq{}{O}{\STREAM\phi}$, and observe
that the reduction of $\LAPP MO$ converges (to an abstraction) and that
$\LAPP M{O\Appsep\NUM 0}$ diverges.

Conversely, we can define a term $N$ such that $\TSEQ
{}N{\LIMPL{\EXCL{(\LIMPL\NAT\phi)}}{\STREAM\phi}}$ which turns a
function into the stream of its successive applications to an integer.
\begin{align*}
  N=\FIXTP F{\LIMPL{\EXCL{(\LIMPL\NAT\phi)}}{\STREAM\phi}} {\ABST
    f{\EXCL{(\LIMPL\NAT\phi)}}{}} \STOP{\PAIR{\LAPP{\GO f}{\NUM
      0}}{\LAPP{\GO F} { \STOP{(\ABST x\NAT{\LAPP{\GO f}{\TSUCC
              x}})} } }}
\end{align*}
Observe that the recursive call of $F$ is encapsulated into a box,
which makes the construction lazy.
As a last example, consider the following term $P$ such that
$\Tseq{}{P}{\LIMPL{(\TENS{\STREAM\phi}{\STREAM\phi})}{\LIMPL{(\PLUS\NAT\NAT)}{\phi}}}$
given by
\begin{align*}
  P=\ABST{y}{\TENS{\STREAM\phi}{\STREAM\phi}}
  {\ABST{c}{\PLUS\NAT\NAT}{\CASE{c}{x}{\LAPP M{x\Appsep\PR \ell y}}{x}{\LAPP M{x\Appsep\PR r y}}}}
\end{align*}
Take $\phi=\ONE$ and consider the term $Q=\PAIR{\STOP{\PAIR{\ONELEM}{O}}}O$,
then we have $\Tseq{}Q{\TENS{\STREAM\One}{\STREAM\One}}$, and observe that
$\LAPP P{Q\Appsep{\IN \ell {\NUM 0}}}$ converges to $\ONELEM$ whereas $\LAPP
P{Q\Appsep{\IN r {\NUM 0}}}$ diverges.

These examples suggest that $\STREAM\phi$ behaves as should behave a type of
streams of elements of type $\phi$.

\subsubsection*{Lists.} There are various possibilities for defining a
type of lists of elements of a positive type $\phi$. The simplest
definition is $\lambda_0=\PLUS\ONE{(\TENS\phi{\lambda_0})}$. This
corresponds to the ordinary ML type of lists. But we can also define
$\lambda_1=\EXCL{(\PLUS\ONE{{(\TENS\phi{\lambda_1})}})}$ and then we have
a type of lazy lists (or terminable streams) where the tail of the
list is computed only when required. Here is an example of a term $L$
such that $\Tseq{}{L}{\lambda_1}$, with $\phi=\BOOL=\PLUS\ONE\ONE$
which is a list of random length containing random booleans:
\begin{align*}
  L=\FIXT x{\lambda_1}{\STOP{(\IFB{\COIN{1/4}}
  {\IN \ell \ONELEM}{\IN r {{\PAIR{\COIN{1/2}}{\GO x}}}})}}
\end{align*}
Then $\GO L$ will reduce with probability $\frac 14$ to the empty list
$\IN \ell \ONELEM$, and with probability $\frac 38$ to 
%
%
each of the values
$\IN r{\PAIR\True L}$ and $\IN r{\PAIR\False L}$.

We can iterate this process, defining a term $R$ of type
$\LIMPL{\lambda_1}{\lambda_0}$ which evaluates completely a terminable
stream to a list:
\begin{align*}
  R=\FIXT{f}{(\LIMPL{\lambda_1}{\lambda_0})}
  {\ABST x{\lambda_1}{\CASE{\GO x}{z}{\IN \ell \ONELEM}{z}{{\PAIR{\PR\ell z}{\LAPP{\GO f}{\PR r z}}}}}}\,.
\end{align*}
Then $\LAPP RL$, which is a closed term of type $\lambda_0$,
terminates with probability $1$. The expectation of the length of this
``random list'' is $\sum_{n=0}^\infty n(\frac 34)^n=12$.


\subsubsection*{Probabilistic tests.}

If $\TSEQ{\cP}{M_i}{\sigma}$ for $i=1,2$, we set
$\DICE{p}{M_1}{M_2}=\IFB{\COIN p}{M_1}{M_2}$  and this term satisfies $\Tseq{\cP}{\DICE{p}{M_1}{M_2}}{\sigma}$. If
$M_i$ reduces to a value $V_i$ with probability $q_i$, then
$\DICE{p}{M_1}{M_2}$ reduces to $V_1$ with probability $p\,q_1$ and to $V_2$ with probability $(1-p)q_2$.

Let $n\in\Nat$ and let $\Vect p=(\List p0n)$ be such that
$p_i\in[0,1]\cap\Rational$ and $p_0+\cdots+p_n\leq 1$. Then one defines a
closed term $\Ran{\Vect p}$, such that $\Tseq{}{\Ran{\Vect p}}{\Tnat}$, which reduces to $\Num i$ with probability $p_i$ for
each $i\in\{0,\dots,n\}$. The definition is by induction on $n$.
\begin{equation*}
  \Ran{\Vect p}=
  \begin{cases}
    \Num 0 & \hspace{-10em}\text{if }p_0=1\text{ whatever be the value of }n\\
    \IFB{\COIN{p_0}}{\Num 0}{\LOOP\Tnat}& \hspace{-10em}\text{if }n=0\\
    \IFB{\COIN{p_0}}
    {\Num 0}
    {\TSUCC{\Ran{\frac{p_1}{1-p_0},\dots,\frac{p_n}{1-p_0}}}}
    &\text{otherwise}
  \end{cases}
\end{equation*}

\medskip
As an example of use of the test to zero conditional,  we define, by induction on $k$, a family of terms $\Probe_k$ such that
$\Tseq{}{\Probe_k}{\LIMPL\Tnat\ONE}$ and that tests the equality to $k$:
\begin{align*}
  \Probe_0 = \Abst x\Tnat{\IFV{x}{\ONELEM}{z}{\LOOP\ONE}}\qquad
  \Probe_{k+1} = \Abst x\Tnat{\IFV{x}{\LOOP\ONE}{z}{\LAPP{\Probe_k}z}}
\end{align*}
For $M$ such that $\Tseq{}{M}{\Tnat}$, the term $\LAPP{\Probe_k}M$ reduces to
$\ONELEM$ with a probability which is equal to the probability of $M$ to reduce
to $\Num k$.

\subsubsection{Notation.}\label{subsec:not-terms}

Now, we introduce terms that will be used in the definition of testing terms
in the proof of Full Abstraction in Section~\ref{sec:FA}.

First, we  define $\Pprod_k$ such that
$\Tseq{}{\Pprod_k}{\LIMPL{\ONE}{\cdots\LIMPL{}{\LIMPL\ONE{\LIMPL{\phi}\phi}}}}$
(with $k$ occurrences of $\ONE$):
\begin{align*}
  \Pprod_0 = \ABST y{\phi} y\qquad
  \Pprod_{k+1}
  = \ABST x\ONE{\Pprod_k}\,.
\end{align*}

Given for each $i\in\{0,\dots,k\}$, $M_i$  such that $\Tseq{\cP}{M_i}\ONE$ and $\Tseq{\cP}{N}{\phi}$, the term
$\LAPP{\LAPP{\LAPP{\Pprod_{k+1}}{M_0}\cdots}{M_k}}N$ reduces to a value $V$ with probability
$p_0\,\cdots\,p_k\,q$ where $p_i$ is the probability of $M_i$ to reduce to $\ONELEM$ and $q$ is the probability of $N$ to reduce to $V$. 
We use the notations:
\begin{equation*}
M_0\cdot N=\LAPP{\LAPP{\Pprod_1}{M_0}}N\qquad  M_0\AND\cdots\AND M_{k-1} =\left\{ 
\begin{array}{l}
\ONELEM\text{ if $k=0$}\\
\LAPP{\LAPP{\Pprod_{k}}{M_0}{\cdots} }M_{k-1}\text{ otherwise}, 
\end{array}\right.
\end{equation*}
so that 
 $\TSEQ{\cP}{M_0\cdot N}{\phi}$ and the probability that $M_0\cdot N$ reduces to $V$ is $p_0\, q$
and
$\TSEQ{\cP}{M_0\AND\cdots\AND M_{k-1}}{\ONE}$ and $M_0\AND\cdots\AND M_{k-1}$ reduces to $\ONELEM$ with probability $p_0\,\cdots\,p_{k-1}$.

\medskip

Given a general type $\sigma$ and terms $\List M0{k-1}$ such that, for
any $i\in\{0,\dots,k-1\}$, $\Tseq{}{M_i}\sigma$, we define close terms
$\Pchoose^\sigma_i(\List M 0{k-1})$ for $i\in\{0,\dots,n-1\}$ such
that $\Tseq{}{\Pchoose^\sigma_i(\List
  M0{k-1})}{\LIMPL\Tnat{{\sigma}}}$
\begin{align*}
  \Pchoose^\sigma_0(\List M 0{k-1}) &= \Abst z\Tnat{\LOOP\sigma}\\
  \Pchoose^\sigma_{i+1}(\List M 0{k-1}) &=
  \Abstpref z\Tnat
  \IFV z{M_0}{y}{\LAPP{\Pchoose^\sigma_{i}(\List M 1{k-1})}{y}} \text{ if }i\le k-1
\end{align*}
Given a term $P$ such that $\Tseq{\cP}{P}{\Tnat}$ and $p_i$ the
probability of $P$ to reduce   to $\Num i$ for any $i$, the
first steps of the reduction are probabilistic:
\begin{align*}
\forall i\in\{0,\dots, k\},\ {\LAPP{\Pchoose^\sigma_{k+1}(\List M0k)}{P}}\Rel{\Redonetr {p_i}}  {\LAPP{\Pchoose^\phi_{k+1}(\List M0k)}{\Num i}}
\end{align*}
the next steps of the reduction are deterministic:
\begin{align*}
 \LAPP{\Pchoose^\phi_{k+1}(\List M0k)}{\Num i}  \Rel{\Wredtr} M_i
\end{align*}

\medskip As we will see more precisely in
Paragraph~\ref{subsec:sem-type-term}, a term of type $\Tnat$ can be
seen as a sub-probability distribution over $\Nat$. Given integers
$0\le l\le r$, we define by induction the term $\Pext lr$ of type
$\Limpl\Tnat\Tnat$:
\begin{eqnarray*}
  \Pext 00&=&\ABST z\Tnat{\IFV z{\Num 0}x{\LOOP\Tnat}}\\
  \forall r>0,\ \Pext 0r&=&
      \Pchoose^\Tnat_{r+1}{\TUPLE{\Num 0,\dots,\Num r}}\\
\Pext {l+1}{r+1}&=&\ABST{z}{\Tnat}{\IFV z{\LOOP\Tnat}{x}{\TSUCC{\LAPP{\Pext lr}x}}}
\end{eqnarray*}
such that if $\Tseq{}{P}{\Tnat}$, then $\LAPP{\Pext lr}P$ extracts the
sub-probability distribution with support
$\subseteq\{l,\dots,r\}$. Indeed, for any $i\in\{l,\dots,r\}$ $\LAPP{\Pext lr}P$ reduces to
$\Num i$ with probability $p_i$ where $p_i$ is the probability of $P$ to reduce to $\Num i$.

We also introduce, for $\Vect n=(\List n0k)$ a sequence of $k+1$ natural
numbers, a term $\Pwin{i}{\Vect n}$ of type
$\Limpl\Tnat\Tnat$ which extracts the sub-probability distribution
whose support is in the $i^{th}$ window of length $n_i$ for
$0\le i< k$:
\begin{eqnarray*}
  \Pwin 0{\Vect n}&=& \Pext 0{n_0-1}\\
  \Pwin {i+1}{\Vect n}&=&\Pext{n_0+\cdots+n_{i}}{n_0+\cdots+n_i+ n_{i+1}-1}.
\end{eqnarray*}

\subsection{On products and recursive definitions of general types}\label{sec:PCF-products}
This section has nothing to do with probabilities, so we consider the
deterministic language $\CBPV$, which is just $\pCBPV$ without the
$\COIN p$ construct.  Our $\CBPV$ general types, which are similar to
Levy's CBPV computation types~\cite{LevyP04} or to the SFPL general
types~\cite{MarzRohrStreicher99}, have $\LIMPL{}{}$ as only type
constructor. This may seem weird when one
compares our language with CBPV where products and recursive
definitions are possible on computation types and are used for
encoding CBN functional languages such as $\PCF$ extended with
products.

For keeping our presentation reasonably short, we will not consider
the corresponding extensions of $\CBPV$ in this paper. Instead, we
will shortly argue that such a $\PCF$ language with products can be
easily encoded in our $\CBPV$.


Concerning recursive type definitions, it is true that adding them as
well as the cartesian product $\IWith$ at the level of general types
would allow to define interesting types such as
$\TREC\zeta{\With\ONE{(\LIMPL{\EXCL\zeta}{\zeta})}}$, yielding a
straightforward encoding of the pure lambda-calculus in our
language. This goal can nevertheless be reached (admittedly in a
slightly less natural way) by using the positive recursive type
definition $\phi=\TREC\zeta{\EXCL{(\LIMPL{\zeta}{\zeta})}}$. A pure
term $t$ will be translated into a $\CBPV$ term $\Cterm t$ such that
$\Tseq{x_1:{\phi_1},\dots,x_n:\phi}{\Cterm t}{\LIMPL\phi\phi}$, where
the list $\List x1n$ contains all the free variables of $t$. This
translation is defined inductively as follows: $\Cterm x=\GO x$,
$\Cterm{(\App st)}=\GO{(\LAPP{\Cterm s}{\STOP{(\Cterm t)}})}$ and
$\Cterm{(\Abs xs)}=\Abs x{\STOP{(\Cterm s)}}$. A simple computation
shows that $\beta$-reduction is validated by this interpretation, but
observe that it is not the case for $\eta$.  The examples we provide
in Section~\ref{subsec:exsyn} also show that our recursive definitions
restricted to positive types allow to introduce a lot of useful data
types (integers, lists, trees, streams etc). So we do not see any real
motivations for adding recursive general type definitions (and their
addition would make the proof of adequacy in
Section~\ref{sec:adequacy} even more complicated).

Coming back to the encoding of products, consider the
following grammar of types
\begin{align*}
  A,B,\dots\Bnfeq \Cnat\Bnfor \Impl AB\Bnfor\Cprod AB
\end{align*}
and the following language of terms
\begin{align*}
  s,t,u,\dots\Bnfeq x\Bnfor\Num n\Bnfor\Csuc s\Bnfor\Cpred s
  \Bnfor\Cifz stu \Bnfor
  \Abst xAs\Bnfor\App st\Bnfor\Cpair st
  \Bnfor \Cproj \ell s\Bnfor \Cproj r s\Bnfor\Cfix xAs
\end{align*}
We call this language $\PCF$ as it is a straightforward extension of the
original $\PCF$ of~\cite{Plotkin77}.

The typing rules are described in Figure~\ref{fig:typing-PCF}.  A typing
context is a sequence $\Gamma=(x_1:A_1,\dots,x_n:A_n)$ where the
variables are pairwise distinct.

We explain now how we interpret this simple programming language in
$\CBPV$.

\subsubsection*{Translating $\PCF$ types.}
With any type $A$, we associate a \emph{finite sequence} of general
types $\Ctype A=(\Ctype A_1,\dots,\Ctype A_n)$ whose length
$n=\Clen A$ is given by: $\Clen\Cnat=1$, $\Clen{\Impl AB}=\Clen B$ and
$\Clen{\Cprod AB}=\Clen A+\Clen B$.

Given a sequence $\Vect\sigma=(\List\sigma 1n)$ of general types we
define $\EXCL{\Vect\sigma}$ by induction: $\EXCL{()}=\ONE$ and
$\EXCL{(\sigma,\Vect\sigma)}=\TENS{\EXCL\sigma}{\EXCL{\Vect\sigma}}$. Given
a positive type $\phi$ and a sequence $\Vect\sigma=(\List\sigma 1n)$
of general types, we define
$\LIMPL\phi{\Vect\sigma}=(\LIMPL\phi{\sigma_1},\dots,\LIMPL\phi{\sigma_n})$.

Using these notations, we can now define $\Ctype A$ as follows:
\begin{itemize}
\item $\Ctype\Cnat=(\NAT)$ (a one element sequence) where $\NAT$ is
  the type of integers introduced in Section~\ref{subsec:exsyn},
  $\NAT=\TREC\zeta{(\PLUS\ONE\zeta)}$,
\item $\Ctype{(\Impl AB)}=\LIMPL{\EXCL{\Ctype A}}{\Ctype B}$,
\item $\Ctype{(\Cprod AB)}=\Ccat{\Ctype A}{\Ctype B}$ (list concatenation).
\end{itemize}

\subsubsection*{Translating $\PCF$ terms.}
Let now $s$ be a $\PCF$ term with typing judgment\footnote{So our
  translation depends on the typing judgment and not only on the term;
  this is fairly standard and can be avoided by considering typed free
  variables.} $\Tseq\Gamma sA$. Let $n=\Clen A$, we define a sequence
$\Cterm s$ of length $n$ such that
$\Tseq{\EXCL\Gamma}{\Cterm s_i}{\Ctype A_i}$ (for $i=1,\dots,n$) as
follows (if $\Gamma=(x_1:C_1,\dots,x_k:C_k)$, then
$\EXCL\Gamma=(x_1:\EXCL{\Ctype{C_1}},\dots,x_k:\EXCL{\Ctype{C_k}})$).

If $s=x_j$, so that $A=C_j$ (for some $j\in\{1,\dots,k\}$), let
$(\List\sigma 1n)=\Ctype A$. Since
$\EXCL{\Ctype
  A}=\TENS{\EXCL{\sigma_1}}{(\TENS{\EXCL{\sigma_2}}
  {\cdots(\TENS{\EXCL{\sigma_{n-1}}}{(\TENS{\EXCL{\sigma_n}}\ONE)})\cdots})}$
we can set $\Cterm x=(\Cterm x_1,\dots,\Cterm x_n)$ where
$\Cterm x_i=\GO{\PR\ell {\PR r {\PR r {\cdots\PR r {x}}}}}$ (with $i-1$ occurrences of
$\PR r {}$).

We set $\Cterm{\Num n}=\IN r {\cdots\IN r{\IN \ell \ONELEM}}$ ($n$ occurrences of
$\IN r{}$), $\Cterm{(\Csuc s)}=\IN r{\Cterm s}$,
$\Cterm{(\Cpred s)}=\CASE{\Cterm s}{x}{\IN \ell \ONELEM}{x}{x}$. Assume that
$s=\Cifz tuv$ with $\Tseq\Gamma t\NAT$, $\Tseq\Gamma uA$ and $\Tseq\Gamma vA$
for some $\PCF$ type $A$. Let $l=\Clen A$. By inductive hypothesis we have
$\Tseq{\EXCL{\Ctype\Gamma}}{\Cterm t}{\NAT}$,
$\Tseq{\EXCL{\Ctype\Gamma}}{\Cterm u_i}{\Ctype A_i}$ and
$\Tseq{\EXCL{\Ctype\Gamma}}{\Cterm v_i}{\Ctype A_i}$ for $i=1,\dots,l$. So we
set $\Cterm s=(\CASE{\Cterm t}{z}{\Cterm u_i}{z}{\Cterm v_i})_{i=1}^l$ where
$z$ is a variable which does not occur free in $u$ or $v$.

Assume now that $\Tseq{\Gamma,x:A}{s}{B}$, we set
$\Cterm{(\Abst xAs)}=(\ABST x{\EXCL{\Ctype A}}{\Cterm s_i})_{i=1}^{\Clen B}$.
Assume that $\Tseq\Gamma s{\Impl AB}$ and $\Tseq\Gamma tA$. Then, setting
$n=\Clen B$, we have
$\Tseq{\EXCL{\Ctype\Gamma}}{\Cterm s_i}{\LIMPL{\EXCL{\Ctype A}}{\Ctype B_i}}$
for $i=1,\dots,n$ and $\Tseq{\EXCL{\Ctype\Gamma}}{\Cterm t_j}{\Ctype A_j}$ for
$j=1,\dots,m$ where $m=\Clen A$. Then, setting
\[
N=\PAIR{\STOP{(\Cterm{t_1})}}{\PAIR{\cdots}{\PAIR{\STOP{(\Cterm{t_{m-1}})}}
    {\PAIR{\STOP{(\Cterm{t_{m}})}}\ONELEM}\cdots}}
\]
we have $\Tseq{\EXCL{\Ctype\Gamma}}{N}{\EXCL{\Ctype A}}$ and we set
$\Cterm{(\App st)}=(\LAPP{\Cterm s_i}N)_{i=1}^n$. Assume that
$s=\Cpair{s_1}{s_2}$ with $\Tseq\Gamma{s_i}{A_i}$ for $i=1,2$ and
$\Tseq\Gamma s{\Cprod{A_1}{A_2}}$. Then we set
$\Cterm s=\Ccat{\Cterm{s_1}}{\Cterm{s_2}}$ (list concatenation). Assume that
$\Tseq\Gamma s{\Cprod{A_1}{A_2}}$, with $n_i=\Clen{A_i}$ for $i=1,2$. Then we
set $\Cterm{(\Cproj \ell s)}=(\Cterm s_1,\dots,\Cterm s_{n_1})$ and
$\Cterm{(\Cproj r s)}=(\Cterm s_{n_1+1},\dots,\Cterm s_{n_1+n_2})$.

Last assume that $s=\Cfix xAt$ with $\Tseq{\Gamma,x:A}{t}A$ so that, setting
$n=\Clen A$, we have
$\Tseq{\EXCL{\Ctype\Gamma},x:\EXCL{\Ctype A}}{\Cterm t_i}{\Ctype A_i}$ for
$i=1,\dots,n$. Let $\List x1n$ be pairwise distinct fresh variables, and set
\[
M_i=\Subst{\Cterm
  t_i}{\PAIR{{x_1}}{\cdots\PAIR{{x_{n-1}}}{{\PAIR{x_{n}}{\ONELEM}}}\cdots}}{x}
\]
for $i=1,\dots,n$.  We have
$\Tseq{\EXCL{\Ctype\Gamma},x_1:\EXCL{\Ctype
    A_1},\dots,x_n:\EXCL{\Ctype A_n}}{M_i}{\Ctype A_i}$.
We are in position of applying the usual trick for encoding mutual
recursive definitions using fixpoints operators.
For the sake of readability, assume that $n=2$, so we have
$\Tseq{\EXCL{\Ctype\Gamma},x_1:\EXCL{\Ctype A_1},x_2:\EXCL{\Ctype
    A_2}}{M_i}{\Ctype A_i}$ for $i=1,2$.
Let $N_1=\FIXT{x_1}{\Ctype{A}_1}{M_1}$ so that
$\Tseq{\EXCL{\Ctype\Gamma},x_2:\EXCL{\Ctype A_2}}{N_1}{\Ctype A_1}$.
Then we have
$\Tseq{\EXCL{\Ctype\Gamma,x_2:\EXCL{\Ctype
      A_2}}}{\Subst{M_2}{\STOP{N_1}}{x_1}}{\Ctype A_2}$.
Therefore we can set
$\Cterm s_2=\FIXT{x_2}{\Ctype A_2}{\Subst{M_2}{\STOP{N_1}}{x_1}}$ and
we have $\Tseq{\EXCL{\Ctype\Gamma}}{\Cterm s_2}{\Ctype A_2}$. Finally
we set $\Cterm s_1=\Subst{N_1}{\STOP{(\Cterm s_2)}}{x_2}$ with
$\Tseq{\EXCL{\Ctype\Gamma}}{\Cterm s_1}{\Ctype A_1}$.

We should check now that this translation is compatible with the
operational semantics of our extended $\PCF$ language. A simple way to do
so would be to choose a simple model of linear logic (for instance,
the relational model) and to prove that the semantics of a $\PCF$ term is
equal to the semantics of its translation in $\pCBPV$ stripped from
its probabilistic construct $\COIN p$ (interpreting tuples of types
using the ``additive'' cartesian product $\IWith$). This is a long and
boring exercise.

\begin{figure*}
  \centering
  \begin{center}
  \AxiomC{}
  \UnaryInfC{$\Tseq{\Gamma,x:A}xA$}
  \DisplayProof
  \quad
  \AxiomC{}
  \UnaryInfC{$\Tseq\Gamma{\Num n}\Cnat$}
  \DisplayProof
  \quad
  \AxiomC{$\Tseq\Gamma s\Cnat$}
  \UnaryInfC{$\Tseq\Gamma{\Csuc s}\Cnat$}
  \DisplayProof
  \quad
  \AxiomC{$\Tseq\Gamma s\Cnat$}
  \UnaryInfC{$\Tseq\Gamma{\Cpred s}\Cnat$}
  \DisplayProof    
  \end{center}
  \begin{center}
    \AxiomC{$\Tseq\Gamma s\Cnat$}
    \AxiomC{$\Tseq\Gamma tA$}
    \AxiomC{$\Tseq\Gamma uA$}
    \TrinaryInfC{$\Tseq\Gamma{\Cifz stu}A$}
    \DisplayProof
    \quad
    \AxiomC{$\Tseq{\Gamma,x:A}sB$}
    \UnaryInfC{$\Tseq\Gamma{\Abst xAs}{\Impl AB}$}
    \DisplayProof
  \end{center}
  \begin{center}
    \AxiomC{$\Tseq\Gamma s{\Impl AB}$}
    \AxiomC{$\Tseq\Gamma tA$}
    \BinaryInfC{$\Tseq\Gamma{\App st}{B}$}
    \DisplayProof
    \quad
    \AxiomC{$\Tseq\Gamma{s_\ell}{A_\ell}$}
    \AxiomC{$\Tseq\Gamma{s_r}{A_r}$}
    \BinaryInfC{$\Tseq\Gamma{\Cpair{s_\ell}{s_r}}{\Cprod{A_\ell}{A_r}}$}
    \DisplayProof
    \quad
    \AxiomC{$\Tseq\Gamma s{\Cprod{A_\ell}{A_r}}$}
    \RightLabel{$i\in\{\ell,r\}$}
    \UnaryInfC{$\Tseq\Gamma{\Cproj is}{A_i}$}
    \DisplayProof
  \end{center}
  \begin{center}
    \AxiomC{$\Tseq{\Gamma,x:A}{M}{A}$}
    \UnaryInfC{$\Tseq\Gamma{\Cfix xAM}{A}$}
    \DisplayProof
  \end{center}
  \caption{Typing rules for a simple call-by-name language with products, PCF}
  \label{fig:typing-PCF}
\end{figure*}

\section{Probabilistic Coherent Spaces}\label{sec:PCS}

\subsection{Semantics of LL, in a nutshell}
\label{sec:LL-semantics-short}

\label{sec:LL-based-models}
The kind of denotational models we are interested in, in this paper, are those
induced by a model of LL, as explained in~\cite{Ehrhard16a}. We remind the basic
definitions and notations, referring to that paper for more details.

\subsubsection{Models of Linear Logic.}\label{sec:LL-models}
A model of LL consists of the following data.

A symmetric monoidal closed category
$(\cL,\ITens,\One,\Leftu,\Rightu,\Assoc,\Sym)$ where we use simple
juxtaposition $g\Compl f$ to denote composition of morphisms $f\in\cL(X,Y)$ and
$g\in\cL(Y,Z)$. We use $\Limpl XY$ for the object of linear morphisms from $X$
to $Y$, $\Evlin\in\cL(\Tens{(\Limpl XY)}{X},Y)$ for the evaluation morphism and
$\Curlin\in \cL(\Tens ZX,Y)\to\cL(Z,\Limpl XY)$ for the linear curryfication
map. For convenience, and because it is the case in the concrete models we
consider (such as Scott Semantics~\cite{Ehrhard16a} or Probabilistic Coherent
Spaces here), we assume this SMCC to be a $*$-autonomous category with
dualizing object $\Bot$. We use $\Orth X$ for the object $\Limpl X\Bot$ of
$\cL$ (the dual, or linear negation, of $X$).

The category $\cL$ is cartesian with terminal object $\Top$, product
$\IWith$, projections $\Proj i$. By $*$-autonomy $\cL$ is co-cartesian
with initial object $\Zero$, coproduct $\IPlus$ and injections
$\Inj i$. By monoidal closeness of $\cL$, the tensor product $\ITens$
distributes over the coproduct $\IPlus$.

We are given a comonad $\Excl\_:\cL\to\cL$ with co-unit $\Der X\in\cL(\Excl
X,X)$ (\emph{dereliction}) and co-multiplication $\Digg X\in\cL(\Excl
X,\Excl{\Excl X})$ (\emph{digging}) together with a strong symmetric monoidal
structure (Seely isos $\Seelyz$ and $\Seelyt$) for the functor $\Excl\_$, from
the symmetric monoidal category $(\cL,\IWith)$ to the symmetric monoidal
category $(\cL,\ITens)$ satisfying an additional coherence condition
wrt.~$\Digg{}$.

We use $\Int\_$ for the ``De Morgan dual'' of $\Excl\_$: $\Int
X=\Orthp{\Exclp{\Orth X}}$ and similarly for morphisms. It is a monad on $\cL$.

\subsubsection{The Eilenberg-Moore category.}\label{sec:EM-general} 
It is then standard to define the category $\EM\cL$ of $\IExcl$-coalgebras. An
object of this category is a pair $P=(\Coalgc P,\Coalgm P)$ where $\Coalgc
P\in\Obj\cL$ and $\Coalgm P\in\cL(\Coalgc P,\Excl{\Coalgc P})$ is such that
$\Der{\Coalgc P}\Compl\Coalgm P=\Id$ and $\Digg{\Coalgc P}\Compl\Coalgm
P=\Excl{\Coalgm P}\Compl\Coalgm P$.  Then $f\in\EM\cL(P,Q)$ iff
$f\in\cL(\Coalgc P,\Coalgc Q)$ such that $\Coalgm Q\Compl f=\Excl
f\Compl\Coalgm P$.  The functor $\Excl\_$ can be seen as a functor from $\cL$
to $\EM\cL$ mapping $X$ to $(\Excl X,\Digg X)$ and $f\in\cL(X,Y)$ to $\Excl
f$. It is right adjoint to the forgetful functor
$\Forgca:\EM\cL\to\cL$. Given $f\in\cL(\Coalgc P,X)$, we use $\Prom
f\in\EM\cL(P,\Excl X)$ for the morphism associated with $f$ by this adjunction,
one has $\Prom f=\Excl f\Compl\Coalgm P$. If $g\in\EM\cL(Q,P)$, we have
$\Prom f\Compl g=\Promp{f\Compl g}$.

Then $\EM\cL$ is cartesian (with product of shape $\Tens PQ=(\Tens{\Coalgc
  P}{\Coalgc Q},\Coalgm{\Tens PQ})$ and terminal object $(\One,\Coalgm\One)$,
still denoted as $\One$). This category is also co-cartesian with coproduct of
shape $\Plus PQ=(\Plus{\Coalgc P}{\Coalgc Q},\Coalgm{\Plus PQ})$ and initial
object $(\Zero,\Coalgm\Zero)$ still denoted as $\Zero$. The complete
definitions can be found in~\cite{Ehrhard16a}. We use $\Contr P\in\EM\cL(P,\Tens
PP)$ (\emph{contraction}) for the diagonal and $\Weak P\in\EM\cL(P,\One)$
(\emph{weakening}) for the unique morphism to the terminal object.

We also consider occasionally the \emph{Kleisli category}\footnote{It is the
  full subcategory of $\EM\cL$ of free coalgebras, see any introductory text on
  monads and co-monads.} $\Kl\cL$ of the
comonad $\oc$: its objects are those of $\cL$ and $\Kl\cL(X,Y)=\cL(\Excl
X,Y)$. The identity at $X$ in this category is $\Der X$ and composition of
$f\in\Kl\cL(X,Y)$ and $g\in\Kl\cL(Y,Z)$ is defined as
\begin{align*}
  g\Comp f=g\Compl\Excl f\Compl\Digg X\,.
\end{align*}
This category is cartesian closed but this fact will not play an essential role
in this work.

\subsubsection{Fixpoints.}\label{parag:sem-fix-points}
For any object $X$, we assume to be given $\Sfix_X\in\cL(\Exclp{\Limpl{\Excl
    X}X},X)$, a morphism such that\footnote{It might seem natural to require
  the stronger uniformity conditions of \emph{Conway
    operator}~\cite{PlotkinSimpson00}. This does not seem to be necessary as
  far as soundness of our semantics is concerned even if the fixpoint
  operators arising in concrete models satisfy these further properties.}
$\Evlin\Compl(\Tens{\Der{\Limpl{\Excl X}{X}}}{\Prom{\Sfix_X}})
\Comp\Contr{\Excl{(\Limpl{\Excl X}{X})}}=\Sfix_X$ which will allow to interpret
term fixpoints.

In order to interpret fixpoints of types, we assume that the category $\cL$ is
equipped with a notion of embedding-retraction pairs, following a standard
approach. We use $\Embr\cL$ for the corresponding category. It is
equipped with a functor $\Funofemb:\Embr\cL\to\Op\cL\times\cL$ such that
$\Funofemb(X)=(X,X)$ and for which we use the notation
$(\Ret\phi,\Emb\phi)=\Funofemb(\phi)$ and assume that
$\Ret\phi\Compl\Emb\phi=\Id_X$.  We assume furthermore that $\Embr\cL$ has all
countable directed colimits and that the functor
$\Embf=\Proj2\Compl\Funofemb:\Embr\cL\to\cL$ is continuous. We also assume that
all the basic operations on objects ($\ITens$, $\IPlus$, $\Orth{(\_)}$ and
$\Excl\_$) are continuous functors from $\Embr\cL$ to itself\footnote{This is a
  rough statement; one has to say for instance that if
  $\phi_i\in\Embr\cL(X_i,Y_i)$ for $i=1,2$ then
  $\Ret{(\Tens{\phi_1}{\phi_2})}=\Tens{\Ret{\phi_1}}{\Ret{\phi_2}}$ etc. The
  details can be found in~\cite{Ehrhard16a}.}.

Then it is easy to carry this notion of embedding-retraction pairs to $\EM\cL$,
defining a category $\Embr{\EM\cL}$, to show that this category has all
countable directed colimits and that the functors $\ITens$ and $\IPlus$ are
continuous on this category: $\Embr{\EM\cL}(P,Q)$ is the set of all
$\phi\in\Embr\cL(\Coalgc P,\Coalgc Q)$ such that $\Emb\phi\in\EM\cL(P,Q)$.  One
checks also that $\Excl{}$ defines a continuous functor from $\Embr\cL$ to
$\Embr{\EM\cL}$. This allows to interpret recursive types, more details can be
found in~\cite{Ehrhard16a}.

\subsubsection{Interpreting types.}
\label{sec:cbpv-interpretation}

Using straightforwardly the object $\One$ and the operations $\ITens$,
$\IPlus$, $\Excl{}$ and $\Limpl{}{}$ of the model $\cL$ as well as the
completeness and continuity properties explained in
Section~\ref{parag:sem-fix-points}, we associate with any positive type $\phi$
and any repetition-free list $\Vect\zeta=(\List\zeta 1n)$ of type variables
containing all free variables of $\phi$ a continuous functor
$\Tsemca\phi_{\Vect\zeta}:(\Embr{\EM\cL})^n\to\Embr{\EM\cL}$ and with any
general type $\sigma$ and any list $\Vect\zeta=(\List\zeta 1n)$ of pairwise
distinct type variables containing all free variables of $\sigma$ we associate
a continuous functor $\Tsem\sigma_{\Vect\zeta}:(\Embr{\EM\cL})^n\to\Embr{\cL}$.

When we write $\Tsem\sigma$ or $\Tsemca\phi$ (without subscript), we assume
implicitly that the types $\sigma$ and $\phi$ have no free type variables. Then
$\Tsem\sigma$ is an object of $\cL$ and $\Tsemca\phi$ is an object of
$\EM\cL$. We have
$\Tsem\phi=\Coalgc{\Tsemca\phi}$ that is, considered as a generalized type, the
semantics of a positive type $\phi$ is the carrier of the coalgebra
$\Tsemca\phi$.

Given a typing context $\cP=(x_1:\phi_1,\dots,x_k:\phi_k)$, we define
$\Tsem\cP=\Tsemca{\phi_1}\ITens\cdots\ITens\Tsemca{\phi_k}\in\EM\cL$. 

In the model or probabilistic coherence spaces considered in this paper, we
define $\Embr\cL$ in such a way that the only isos are the identity maps. This
implies that the types $\TREC\zeta\phi$ and $\Subst\phi{(\TREC\zeta\phi)}\zeta$
are interpreted as \emph{the same object} (or functor). Such definitions of
$\Embr\cL$ are possible in many other models (relational, coherence spaces,
hypercoherences etc).

We postpone the description of term interpretation because this will 
require constructions specific to our probabilistic semantics, in addition to
the generic categorical ingredients introduced so far.

\subsection{The model of probabilistic coherence spaces}\label{subsec:model-pcoh}

Given a countable set $I$ and $u,u'\in\Realpto I$, we set
$\Eval{u}{u'}=\sum_{i\in I}u_iu'_i$. Given $\cF\subseteq\Realpto I$, we set
$\Orth\cF=\{u'\in\Realpto I\St\forall u\in\cF\ \Eval{u}{u'}\leq 1\}$.

A \emph{probabilistic coherence space} (PCS) is a pair $X=(\Web X,\Pcoh X)$
where $\Web X$ is a countable set and $\Pcoh X\subseteq\Realpto{\Web X}$
satisfies
\begin{itemize}
\item $\Biorth{\Pcoh X}=\Pcoh X$ (equivalently, $\Biorth{\Pcoh X}\subseteq\Pcoh
  X$),
\item for each $a\in\Web X$ there exists $u\in\Pcoh X$ such that $u_a>0$,
\item for each $a\in\Web X$ there exists $A>0$ such that $\forall u\in\Pcoh X\
  u_a\leq A$.
\end{itemize}
If only the first of these conditions holds, we say that $X$ is a
\emph{pre-probabilistic coherence space} (pre-PCS).

The purpose of the second and third conditions is to prevent infinite
coefficients to appear in the semantics. This property in turn will be
essential for guaranteeing the morphisms interpreting proofs to be analytic
functions, which will be the key property to prove full abstraction. So these
conditions, though aesthetic at first sight, are important for our ultimate
goal.

\begin{lemma}\label{lemma:pcoh-charact}
  Let $X$ be a pre-PCS. The following conditions are
  equivalent:
  \begin{itemize}
  \item $X$ is a PCS,
  \item $\forall a\in\Web X\,\exists u\in\Pcoh X\,\exists
    u'\in\Orth{\Pcoh X}\ u_a>0\text{ and }u'_a>0$,
  \item $\forall a\in\Web X\,\exists A>0\,\forall u\in\Pcoh X\,\forall
    u'\in\Orth{\Pcoh X}\ u_a\leq A\text{ and }u'_a\leq A$.
  \end{itemize}
\end{lemma}
The proof is straightforward.

We equip $\Pcoh X$ with the most obvious partial order relation: $u\leq v$ if
$\forall a\in\Web X\ u_a\leq v_a$ (using the usual order relation on $\Real$).

\begin{theorem}\label{th:Pcoh-prop}
  $\Pcoh X$ is an $\omega$-continuous domain. Given $u,v\in\Pcoh X$ and
  $\alpha,\beta\in\Realp$ such that $\alpha+\beta\leq 1$, one has $\alpha
  u+\beta v\in\Pcoh X$.
\end{theorem}
This is an easy consequence of the hypothesis
$\Biorth{\Pcoh X}\subseteq\Pcoh X$. See~\cite{DanosEhrhard08} for
details; from this result, we will only use the closure properties:
$\Pcoh X$ is closed under sub-probabilistic linear combinations and
under lubs of monotonic sequences. Though the $\omega$-continuity
property (and the associated way-below relation) does not play any
technical role, it is an intuitively satisfactory
fact\footnote{The $\omega$-continuity is similar to separability for topological vector spaces.}
which means that the ``size'' of our domains remains bounded.

\subsubsection{Morphisms of PCSs}\label{subsec:morphisms}

Let $X$ and $Y$ be PCSs. Let $t\in(\Realp)^{\Web X\times\Web Y}$ (to be
understood as a matrix). Given $u\in\Pcoh X$, we define $\Matapp
tu\in\Realpc^{\Web Y}$ by $(\Matapp tu)_b=\sum_{a\in\Web X}t_{a,b}u_a$
(application of the matrix $t$ to the vector $u$)\footnote{This is an unordered
  sum, which is infinite in general. It makes sense because all its terms are
  $\geq 0$.}.  We say that $t$ is a \emph{(linear) morphism} from $X$ to $Y$ if
$\forall u\in\Pcoh X\ \Matapp tu\in\Pcoh Y$, that is
\begin{align*}
  \forall u\in\Pcoh X\,\forall{v'}\in\Orth{\Pcoh Y}\quad\sum_{(a,b)\in\Web
    X\times\Web Y}t_{a,b}u_av'_b\leq 1\,.
\end{align*}
The diagonal matrix $\Id\in(\Realp)^{\Web X\times\Web X}$ given by
$\Id_{a,b}=1$ if $a=b$ and $\Id_{a,b}=0$ otherwise is a morphism. In that way
we have defined a category $\PCOH$ whose objects are the PCSs and whose
morphisms have just been defined. Composition of morphisms is defined as matrix
multiplication: let $s\in\PCOH(X,Y)$ and $t\in\PCOH(Y,Z)$, we define $\Matapp
ts\in(\Realp)^{\Web X\times\Web Z}$ by
\begin{align*}
  (\Matapp ts)_{a,c}=\sum_{b\in\Web Y}s_{a,b}t_{b,c}
\end{align*}
and a simple computation shows that $\Matapp ts\in\PCOH(X,Z)$. More precisely,
we use the fact that, given $u\in\Pcoh X$, one has $\Matapp{(\Matapp
  ts)}{u}=\Matapp t{(\Matapp su)}$. Associativity of composition holds because
matrix multiplication is associative. $\Id_X$ is the identity morphism at $X$.

Given $u\in\Pcoh X$, we define $\Norm u_X=\sup\{\Eval u{u'}\St u'\in\Pcoh{\Orth
X}\}$. By definition, we have $\Norm u_X\in[0,1]$.

\subsubsection{Multiplicative constructs}\label{sec:multiplicatives}
One sets $\Orth X=(\Web X,\Orth{\Pcoh X})$. It results straightforwardly from
the definition of PCSs that $\Orth X$ is a PCS. Given $t\in\PCOH(X,Y)$, one has
$\Orth t\in\PCOH(\Orth Y,\Orth X)$ if $\Orth t$ is the transpose of $t$, that
is $(\Orth t)_{b,a}=t_{a,b}$.

One defines $\Tens{X}{Y}$ by
$\Web{\Tens{X}{Y}}=\Web{X}\times\Web{Y}$ and
\begin{equation*}
  \Pcohp{\Tens XY}=\Biorth{\{\Tens uv\St u\in\Pcoh X\text{ and }v\in\Pcoh Y\}}
\end{equation*}
where $\Tensp uv_{(a,b)}=u_av_b$.  Then $\Tens XY$ is a pre-PCS.

We have 
\[
\Pcoh{\Orth{(\Tens{X}{\Orth Y})}}=\Orth{\{\Tens{u}{v'}\St u\in\Pcoh X\text{ and
  }v'\in\Pcoh{\Orth Y}\}}=\PCOH(X,Y)\,.
\]
It follows that $\Limpl XY=\Orth{(\Tens{X}{\Orth Y})}$ is a pre-PCS. Let
$(a,b)\in\Web X\times\Web Y$. Since $X$ and $\Orth Y$ are PCSs, there is $A>0$
such that $u_av'_b<A$ for all $u\in\Pcoh X$ and $v'\in\Pcoh{\Orth Y}$. Let
$t\in(\Realp)^{\Web{\Limpl XY}}$ be such that $t_{(a',b')}=0$ for
$(a',b')\not=(a,b)$ and $t_{(a,b)}=1/A$, we have $t\in\Pcoh{(\Limpl XY)}$. This
shows that $\exists t\in\Pcoh{(\Limpl XY)}$ such that $t_{(a,b)}>0$. Similarly
we can find $u\in\Pcoh X$ and $v'\in\Pcoh{\Orth Y}$ such that
$\epsilon=u_av'_b>0$. It follows that $\forall t\in\Pcoh{(\Limpl XY)}$ one has
$t_{(a,b)}\leq 1/\epsilon$. We conclude that $\Limpl XY$ is a PCS, and
therefore $\Tens XY$ is also a PCS.

\begin{lemma}\label{lemma:PCS-moprh-charact}
  Let $X$ and $Y$ be PCSs. One has $\Pcoh{(\Limpl XY)}=\PCOH(X,Y)$. That is,
  given $t\in\Realpto{\Web X\times\Web Y}$, one has $t\in\Pcoh{(\Limpl XY)}$
  iff for all $u\in\Pcoh X$, one has $t\Compl u\in\Pcoh Y$.
\end{lemma}
This results immediately from the definition above of $\Limpl XY$.

\begin{lemma}\label{lemma:tens-morph-charact}
  Let $X_1$, $X_2$ and $Y$ be PCSs. Let
  $t\in(\Realp)^{\Web{\Limpl{\Tens{X_1}{X_2}}{Y}}}$. One has
  $t\in\PCOH(\Tens{X_1}{X_2},Y)$ iff for all $u_1\in\Pcoh{X_1}$ and
  $u_2\in\Pcoh{X_2}$ one has $\Matapp t{(\Tens{u_1}{u_2})}\in\Pcoh Y$.
\end{lemma}
\Beginproof
The condition stated by the lemma is clearly necessary. Let us prove that it is
sufficient: under this condition, it suffices to prove that
\begin{align*}
  \Orth t\in\PCOH(\Orth Y,\Orth{(\Tens{X_1}{X_2})})\,.
\end{align*}
Let $v'\in\Pcoh{\Orth Y}$, it suffices to prove that $\Matapp{\Orth
  t}{v'}\in\Pcoh{\Orth{(\Tens{X_1}{X_2})}}$. So let $u_1\in\Pcoh{X_1}$ and
$u_2\in\Pcoh{X_2}$, it suffices to prove that $\Eval{\Matapp{\Orth
    t}{v'}}{\Tens{u_1}{u_2}}\leq 1$, that is $\Eval{\Matapp
  t{(\Tens{u_1}{u_2})}}{v'}\leq 1$, which follows from our assumption.
\Endproof

Let $s_i\in\PCOH(X_i,Y_i)$ for $i=1,2$. Then one defines
\begin{align*}
\Tens{s_1}{s_2}\in(\Realp)^{\Web{\Limpl{\Tens{X_1}{X_2}}{\Tens{Y_1}{Y_2}}}}  
\end{align*}
by
$(\Tens{s_1}{s_2})_{((a_1,a_2),(b_1,b_2))}=(s_1)_{(a_1,b_1)}(s_2)_{(a_2,b_2)}$
and one must check that
\[
\Tens{s_1}{s_2}\in\PCOH({\Tens{X_1}{X_2},\Tens{Y_1}{Y_2}})\,.
\] 
This follows directly from Lemma~\ref{lemma:tens-morph-charact}. Let
$\One=(\{*\},[0,1])$. There are obvious choices of natural isomorphisms
\begin{align*}
  \Leftu_X&\in\PCOH(\Tens\One X,X)\\
  \Rightu_X&\in\PCOH(\Tens X\One,X)\\
  \Assoc_{X_1,X_2,X_3}&\in\PCOH(\Tens{\Tensp{X_1}{X_2}}{X_3},
      \Tens{X_1}{\Tensp{X_2}{X_3}})\\
  \Sym_{X_1,X_2}&\in\PCOH(\Tens{X_1}{X_2},\Tens{X_2}{X_1})  
\end{align*}
which satisfy the standard coherence properties. This shows that 
$(\PCOH,\One,\Leftu,\Rightu,\Assoc,\Sym)$ is a symmetric monoidal category.

\subsubsection{Internal linear hom.}
Given PCSs $X$ and $Y$, let us define
$\Evlin\in(\Realp)^{\Web{\Limpl{\Tens{(\Limpl XY)}{X}}{Y}}}$ by
\begin{align*}
\Evlin_{(((a',b'),a),b)}=
\begin{cases}
  1 & \text{if }(a,b)=(a',b')\\
  0 & \text{otherwise.}
\end{cases}
\end{align*}
Then it is easy to see that $(\Limpl XY,\Evlin)$ is an internal linear hom
object in $\PCOH$, showing that this SMCC is closed. If $t\in\PCOH(\Tens
ZX,Y)$, the corresponding linearly curryfied morphism
$\Curlin(t)\in\PCOH(Z,\Limpl XY)$ is given by
$\Curlin(t)_{(c,(a,b))}=t_{((c,a),b)}$. 

\subsubsection{$*$-autonomy.}
Take $\Bot=\One$, then one checks readily that 
$(\PCOH,\One,\Leftu,\Rightu,\Assoc,\Sym,\Bot)$ is a $*$-autonomous category. The
duality functor $X\mapsto(\Limpl X\Bot)$ can be identified with the strictly
involutive contravariant functor $X\mapsto\Orth X$.

\subsubsection{Additives}\label{sec:sem-additives}
Let $(X_i)_{i\in I}$ be a countable family of PCSs. We define a PCS
$\Bwith_{i\in I}X_i$ by $\Web{\Bwith_{i\in I}X_i}=\bigcup_{i\in
  I}\{i\}\times\Web{X_i}$ and $u\in\Pcohp{\Bwith_{i\in I}X_i}$ if, for all
$i\in I$, the family $u(i)\in(\Realp)^{\Web{X_i}}$ defined by
$u(i)_a=u_{(i,a)}$ belongs to $\Pcoh{X_i}$.

\begin{lemma}
  Let $u'\in(\Realp)^{\Web{\Bwith_{i\in I}X_i}}$. One has
  $u'\in\Pcoh{\Orthp{\Bwith_{i\in I}X_i}}$ iff
  \begin{itemize}
  \item $\forall i\in I\ u'(i)\in\Pcoh{\Orth{X_i}}$
  \item and $\sum_{i\in I}\Norm{u'(i)}_{\Orth{X_i}}\leq 1$.
  \end{itemize}
\end{lemma}
The proof is quite easy. It follows that $\Bwith_{i\in I}X_i$ is a
PCS. Moreover we can define $\Proj i\in\PCOH(\Bwith_{j\in I}X_j,X_i)$ by
\begin{align*}
(\Proj i)_{(j,a),a'}=
\begin{cases}
  1 & \text{if }j=i\text{ and }a=a'\\
  0 & \text{otherwise.}
\end{cases}
\end{align*}
Then $(\Bwith_{i\in I}X_i,(\Proj i)_{i\in I})$ is the cartesian product of the
family $(X_i)_{i\in I}$ in the category $\PCOH$. The coproduct $(\Bplus_{i\in
  I}X_i,(\Inj i)_{i\in I})$ is the dual operation, so that
\begin{align*}
  \Web{\Bplus_{i\in I}X_i}=\Union_{i\in I}\{i\}\times\Web{X_i}
\end{align*}
and $u\in\Pcoh{(\Bplus_{i\in I}X_i)}$ if $\forall i\in I\ u(i)\in\Pcoh{X_i}$
and $\sum_{i\in I}\Norm{u(i)}_{X_i}\leq 1$. The injections
$\Inj j\in\PCOH(X_j,\Bplus_{i\in I}X_i)$ are given by
\begin{align*}
(\Inj i)_{a',(j,a)}=
\begin{cases}
  1 & \text{if }j=i\text{ and }a=a'\\
  0 & \text{otherwise.}
\end{cases}
\end{align*}
Given morphisms $s_i\in\PCOH(X_i,Y)$ (for each $i\in I$), then the
unique morphism $s\in\PCOH(\Bplus_{i\in I} X_i, Y)$ is given by
$s_{(i,a),b}=(s_i)_{a,b}$ and denoted as $\Case_{i\in I}s_i$ (in the binary
case, we use $\Case(s_1,s_2)$).

\subsubsection{Exponentials}
Given a set $I$, a \emph{finite multiset} of elements of $I$ is a function
$b:I\to\Nat$ whose \emph{support} $\Supp b=\{a\in I\St b(a)\not=0\}$ is
finite. We use $\Mfin I$ for the set of all finite multisets of elements of
$I$. Given a finite family $\List a1n$ of elements of $I$, we use $\Mset{\List
  a1n}$ for the multiset $b$ such that $b(a)=\Card{\{i\St a_i=a\}}$. We use
additive notations for multiset unions: $\sum_{i=1}^kb_i$ is the multiset
$b$ such that $b(a)=\sum_{i=1}^k b_i(a)$. The empty multiset is denoted as
$0$ or $\Msetempty$. If $k\in\Nat$, the multiset $kb$ maps $a$ to $k\,b(a)$.

Let $X$ be a PCS. Given $u\in\Pcoh X$ and $b\in\Mfin{\Web X}$, we define
$u^b=\prod_{a\in\Web X}u_a^{b(a)}\in\Realp$. Then we set $\Prom
u=(u^b)_{b\in\Mfin{\Web X}}$ and finally
\begin{align*}
  \Excl X=(\Mfin{\Web X},\Biorth{\{\Prom u\St u\in\Pcoh X\}})
\end{align*}
which is a pre-PCS. 

We check quickly that $\Excl X$ so defined is a PCS. Let $b=\Mset{\List
  a1n}\in\Mfin{\Web X}$.  Because $X$ is a PCS, and by
Theorem~\ref{th:Pcoh-prop}, for each $i=1,\dots,n$ there is $u(i)\in\Pcoh X$
such that $u(i)_{a_i}>0$. Let $(\alpha_i)_{i=1}^n$ be a family of strictly
positive real numbers such that $\sum_{i=1}^n\alpha_i\leq 1$. Then
$u=\sum_{i=1}^n\alpha_iu(i)\in\Pcoh X$ satisfies $u_{a_i}>0$ for each
$i=1,\dots,n$. Therefore $\Prom u_b=u^b>0$. This shows that there is
$U\in\Pcoh{(\Excl X)}$ such that $U_b>0$.

Let now $A\in\Realp$ be
such that $\forall u\in\Pcoh X\,\forall i\in\{1,\dots,n\}\ u_{a_i}\leq A$. For
all $u\in\Pcoh X$ we have $u^b\leq A^n$. We have
\begin{align*}
  \Orth{(\Pcoh{(\Excl X)})}=\Triorth{\{\Prom u\St u\in\Pcoh X\}}
  =\Orth{\{\Prom u\St u\in\Pcoh X\}}\,.
\end{align*}
Let $t\in\Realpto{\Web{\Excl X}}$ be defined by $t_c=0$ if $c\not=b$ and
$t_b=A^{-n}>0$; we have $t\in\Orth{(\Pcoh{(\Excl X)})}$. We have exhibited an
element $t$ of $\Orth{(\Pcoh{(\Excl X)})}$ such that $t_b>0$. By
Lemma~\ref{lemma:pcoh-charact} it follows that $\Excl X$ is a PCS.

\subsubsection{Kleisli morphisms as functions.}\label{subsec:Kleisli_fun}
Let $s\in\Realpto{\Web{\Limpl{\Excl X}{Y}}}$. We define a function $\Fun
s:\Pcoh X\to\Realpcto{\Web Y}$ as follows.  Given $u\in\Pcoh X$, we set
\begin{align*}
  \Fun s(u)=s\Compl {\Prom u}=\left(\sum_{c\in\Web{\Excl
        X}}s_{c,b}u^c\right)_{b\in\Web Y}\,.
\end{align*}

\begin{theorem}\label{prop:kleisli-morph-charact}
  One has $s\in\Pcoh{(\Limpl{\Excl X}{Y})}$ iff, for all $u\in\Pcoh X$, one has
  $\Fun s(u)\in\Pcoh Y$.  
\end{theorem}
This is an immediate consequence of the definition.

\begin{theorem}\label{th:pcoh-functional}
  Let $s\in\PCOH(\Excl X,Y)$.  The function $\Fun s$ is
  Scott-continuous. Moreover, given $s,s'\in\PCOH(\Excl X,Y)$, one has $s=s'$
  (as matrices) iff $\Fun s=\Fun{s'}$ (as functions $\Pcoh X\to\Pcoh Y$).
\end{theorem}
This is an easy consequence of the fact that two polynomials of $n$ variables
with real coefficients are identical iff they are the same function on any open
subset of $\Real^n$.

\emph{Terminology.} We say that $s\in\Pcoh{(\Limpl{\Excl X}Y)}$ is a
\emph{power series} whose monomial $u^c$ has coefficient
$s_{c,b}$. Since $s$ is characterized by the function $\Fun s$ we
sometimes say that $\Fun s$ is a power series.

\medskip
We can consider the elements of $\Kl\PCOH(X,Y)$ (the morphisms of the
Kleisli category of the comonad $\Excl\_$ on the category $\PCOH$) as
particular Scott continuous functions $\Pcoh X\to\Pcoh Y$ and this
identification is compatible with the definition of identity maps and of
composition in $\Kl\PCOH$, see Section~\ref{sec:EM-general}. Of course, not all
Scott continuous function are morphisms in $\Kl\PCOH$.

\begin{theorem}\label{prop:order-fun-kleiseli}
  Let $s,s'\in\Kl\PCOH(X,Y)$ be such that $s\leq s'$ (as elements of
  $\Pcoh{(\Limpl{\Excl X}Y)}$). Then $\forall u\in\Pcoh X\ \Fun
  s(u)\leq\Fun{s'}(u)$. Let $(s(i))_{i\in\Nat}$ be a monotone sequence of
  elements of $\Kl\PCOH(X,Y)$ and let $s=\sup_{i\in\Nat}s(i)$. Then $\forall
  u\in\Pcoh X\ \Fun s(u)=\sup_{i\in I}\Fun{s_i}(u)$.
\end{theorem}
The first statement is obvious. The second one results from the monotone
convergence Theorem.

Given a multiset $b\in\Mfin I$, we define its \emph{factorial}
$\Factor b=\prod_{i\in I}\Factor{b(i)}$ and its \emph{multinomial
  coefficient} $\Multinom{}{b}=\Factor{(\Card b)}/\Factor b\in\Natnz$ where
$\Card b=\sum_{i\in I}b(i)$ is the cardinality of $b$. Remember that,
given an $I$-indexed family $a=(a_i)_{i\in I}$ of elements of a commutative
semi-ring, one has the multinomial formula
\begin{align*}
  \Big(\sum_{i\in I}a_i\Big)^n=\sum_{b\in\Mfinc nI}\Multinom{}b a^b
\end{align*}
where $\Mfinc nI=\{b\in\Mfin I\St\Card b=n\}$.

Given $c\in\Web{\Excl X}$ and $d\in\Web{\Excl Y}$ we define
$\Mexpset cd$ as the set of all multisets $r$ in $\Mfin{\Web X\times\Web
  Y}$ such that
\begin{align*}
  \forall a\in\Web X\ \sum_{b\in\Web Y}r(a,b)=c(a)
  \quad\text{and}\quad
  \forall b\in\Web Y\ \sum_{a\in\Web X}r(a,b)=d(b)\,.
\end{align*}

Let $t\in\PCOH(X,Y)$, we define $\Excl t\in\Realpto{\Limpl{\Excl X}{\Excl Y}}$
by
\begin{align*}
  (\Excl t)_{c,d}
  =\sum_{r\in\Mexpset cd}\frac{\Factor d}{\Factor r}t^r\,.
\end{align*}
Observe that the coefficients in this sum are all non-negative integers.

\begin{lemma}\label{lemma:excl-morph-app}
  For all $u\in\Pcoh X$ one has $\Excl t\Compl\Prom u=\Prom{(t\Compl u)}$.
\end{lemma}
This results from a simple computation applying the multinomial formula.

\begin{theorem}
  For all $t\in\PCOH(X,Y)$ one has $\Excl t\in\PCOH(\Excl X,\Excl Y)$ and the
  operation $t\mapsto\Excl t$ is functorial.
\end{theorem}
Immediate consequences of Lemma~\ref{lemma:excl-morph-app} and
Theorem~\ref{th:pcoh-functional}.

\subsubsection{Description of the exponential comonad.}
We equip now this functor with a structure of comonad: let $\Der
X\in\Realpto{\Web{\Limpl{\Excl X}X}}$ be given by $(\Der
X)_{b,a}=\Kronecker{\Mset a}{b}$ (the value of the Kronecker symbol
$\Kronecker ij$ is $1$ if $i=j$ and $0$ otherwise) and $\Digg
X\in\Realpto{\Web{\Limpl{\Excl X}{\Excl{\Excl X}}}}$ be given by $(\Digg
X)_{b,\Mset{\List b1n}}=\Kronecker{\sum_{i=1}^n b_i}{b}$. Then we have
$\Der X\in\PCOH(\Excl X,X)$ and $\Digg X\in\PCOH(\Excl X,\Excl{\Excl X})$
simply because
\begin{align*}
  \Fun{\Der X}(u)=u\quad\text{and}\quad\Fun{\Digg X}(u)=\Prom{(\Prom u)}
\end{align*}
for all $u\in\Pcoh X$, as easily checked. Using these equations, one also
checks easily the naturality of these morphisms, and the fact that
$(\Excl\_,\Der{},\Digg{})$ is a comonad.

As to the monoidality of this comonad, we introduce
$\Expmonisoz\in\Realpto{\Web{\Limpl{\One}{\Excl\Top}}}$ by
$\Expmonisoz_{\Onelem,\Mset{}}=1$ and $\Expmonisob
XY\in\Realpto{\Web{\Limpl{\Tens{\Excl X}{\Excl Y}}{\Excl{(\With XY)}}}}$ by
$(\Expmonisob XY)_{b,c,d}=\Kronecker{d}{\Injms 1b+\Injms
  2c}$ 
where $\Injms i{\Mset{\List a1n}}=\Mset{(i,a_1),\dots,(i,a_n)}$. It is easily
checked that the required commutations hold (again, we refer
to~\cite{Mellies09}).

\subsubsection{Fixpoints in $\Kl\PCOH$.}
For any object $Y$ of $\PCOH$, a morphism $t\in\Kl\PCOH(Y,Y)$ defines a
Scott-continuous function $f=\Fun t:\Pcohp Y\to\Pcohp Y$ which has a least
fixpoint $\sup_{n\in\Nat}f^n(0)$. Let $X$ be an object of $\PCOH$ and set
$Y=\Limpl{\Exclp{\Limpl{\Excl X}{X}}}X$. Then we have a morphism $t=\Curlin
s\in\Kl\PCOH(Y,Y)$ where 
$s\in\PCOH(
\Tens{\Excl Y} 
     {\Exclp{\Limpl{\Excl X}{X}}}
,X)$ is defined as the following composition of morphisms in $\PCOH$:
\begin{center}
  \begin{tikzpicture}[->, >=stealth]
    \node (1) {$\Tens{\Excl Y}{\Excl{(\Limpl{\Excl X}{X})}}$};
    \node (2) [right of=1, node distance=60mm] 
    {$\Excl Y\ITens\Excl{(\Limpl{\Excl X}{X})}
      \ITens\Excl{(\Limpl{\Excl X}{X})}$};
    \node (3) [below of=2, node distance=12mm] 
      {$\Excl X\ITens(\Limpl{\Excl X}{X})$};
    \node (4) [below of=2, node distance=24mm] {$X$};
    \tikzstyle{every node}=[midway,auto,font=\scriptsize]
   \draw (1) -- node [below] {$s$} (4);
    \draw (1) -- node {$\Tens{\Excl Y}{\Contr{\Limpl{\Excl X}{X}}}$} (2);
    \draw (2) -- node {$\Tens
      {\Prom{(\Evlin\Compl(\Tens{\Der Y}{\Excl{(\Limpl{\Excl X}{X})}}))}}
      {\Der{(\Limpl{\Excl X}{X})}}$} (3);
    \draw (3) -- node {$\Evlin\Compl\Sym$} (4);    

  \end{tikzpicture}
\end{center}
 Then $\Fun t$ is a Scott continuous function
$\Pcoh Y\to\Pcoh Y$ whose least fixpoint is $\Sfix$, considered as a
morphism $\Sfix\in\Kl\PCOH(\Limpl{\Excl X}{X},X)$, satisfies
$\Fun\Sfix(u)=\sup_{n=0}^\infty\Fun u^n(0)$.

\subsubsection{The partially ordered class of probabilistic coherence spaces}

\newcommand\Subobj{\subseteq}
\newcommand\Subwit[2]{\eta_{#1,#2}}

We define the category $\Embr\PCOH$. This category is actually a partially
ordered class whose objects are those of $\PCOH$. The order relation, denoted
as $\Subobj$, is defined as follows: $X\Subobj Y$ if $\Web X\subseteq\Web Y$
and the matrices $\Emb{\Subwit XY}$ and $\Ret{\Subwit XY}$ defined, for
$a\in\Web X$ and $b\in\Web Y$, by $(\Emb{\Subwit XY})_{a,b}=(\Ret{\Subwit
  XY})_{b,a}=\Kronecker ab$ satisfy $\Emb{\Subwit XY}\in\PCOH(X,Y)$ and
$\Ret{\Subwit XY}\in\PCOH(Y,X)$. In other words: given $u\in\Pcoh X$, the
element $\Matapp{\Emb{\Subwit XY}}{u}$ of $\Realpto{\Web Y}$ obtained by
extending $u$ with $0$'s outside $\Web X$ belongs to $\Pcoh Y$. And conversely,
given $v\in\Pcoh Y$, the element $\Matapp{\Ret{\Subwit XY}}{v}$ of
$\Realpto{\Web X}$ obtained by restricting $v$ to $\Web X$ belongs to $\Pcoh
X$. Considering $\Embr\PCOH$ as a category, $\Subwit XY$ is a
notation for the unique element of $\Embr\PCOH(X,Y)$ when $X\Subobj Y$, in
accordance with the notations of Paragraph~\ref{parag:sem-fix-points}.

\begin{lemma}\label{lemma:sem-orth-subobj}
  If $X\Subobj Y$ then $\Orth X\Subobj\Orth Y$, $\Emb{\Subwit{\Orth X}{\Orth
    Y}}=\Orth{(\Ret{\Subwit XY})}$ and  $\Ret{\Subwit{\Orth X}{\Orth
    Y}}=\Orth{(\Emb{\Subwit XY})}$.
\end{lemma}
The proof is a straightforward verification.

We contend that $\Embr\PCOH$ is directed co-complete.  Let
$(X_\gamma)_{\gamma\in\Gamma}$ be a countable directed family in $\Embr\PCOH$
(so $\Gamma$ is a countable directed poset and $\gamma\leq\gamma'\Implies
X_\gamma\Subobj X_{\gamma'}$), we have to check that this family has a least
upper bound $X$. We set $\Web X=\bigcup_{\gamma\in\Gamma}\Web{X_\gamma}$ and
$\Pcoh X=\{w\in\Realpto{\Web X}\St\forall\gamma\in\Gamma\ \Matapp{\Ret{\Subwit
    XY}}{w}\in\Pcoh{X_\gamma}\}$. This defines an object of $\PCOH$ which
satisfies $\Pcoh
X=\Biorth{\{\Matapp{\Emb{\Subwit{X_\gamma}{X}}}{u}\St\gamma\in\Gamma\text{ and
  }u\in\Pcoh{X_\gamma}\}}$ and is therefore the lub of the family
$(X_\gamma)_{\gamma\in\Gamma}$ in $\Embr\PCOH$. This object $X$ is denoted
$\bigcup_{\gamma\in\Gamma}X_\gamma$. One checks easily that
$\Orth{(\bigcup_{\gamma\in\Gamma}X_\gamma)}
=\bigcup_{\gamma\in\Gamma}\Orth{X_\gamma}$.

Then the functor $\Embf:\Embr\PCOH\to\PCOH$ defined by $\Embf(X)=X$ and
$\Embf(\Subwit XY)=\Emb{\Subwit XY}$ is continuous: given a directed family
$(X_\gamma)_{\gamma\in\Gamma}$ whose lub is $X$ and given a collection of
morphisms $t_\gamma\in\PCOH(X_\gamma,Y)$ such that
$t_{\gamma'}\Compl\Emb{\Subwit{X_{\gamma}}{X_{\gamma'}}}=t_\gamma$ for any
$\gamma,\gamma'\in\Gamma$ such that $\gamma\leq\gamma'$, there is exactly one
morphism $t\in\PCOH(X,Y)$ such that
$t\Compl\Emb{\Subwit{X_\gamma}{X}}=t_\gamma$
for each $\gamma\in\Gamma$. Given $a\in\Web X$ and $b\in\Web Y$,
$t_{a,b}=(t_\gamma)_{a,b}$ for any $\gamma$ such that $a\in\Web{X_\gamma}$ (our
hypothesis on the $t_\gamma$'s means that $(t_\gamma)_{a,b}$ does not depend on
the choice of $\gamma$).

All the operations of Linear Logic define monotone continuous functionals on
$\Embr\PCOH$ which moreover commute with the functor $\Funofemb$. This means
for instance that if $X\Subobj Y$ then $\Excl X\Subobj\Excl Y$,
$\Emb{\Subwit{\Excl X}{\Excl Y}}=\Excl{(\Emb{\Subwit{X}{Y}})}$,
$\Ret{\Subwit{\Excl X}{\Excl Y}}=\Excl{(\Ret{\Subwit{X}{Y}})}$ and
$\Excl{(\bigcup_{\gamma\in\Gamma}X_\gamma)}=
\bigcup_{\gamma\in\Gamma}\Excl{X_\gamma}$ and similarly for $\ITens$ and
$\IPlus$. As a consequence, and as a consequence of
Lemma~\ref{lemma:sem-orth-subobj}, if $X_i\Subobj Y_i$ for
$i=1,2$ then $\Limpl{X_1}{X_2}\Subobj\Limpl{Y_1}{Y_2}$,
$\Emb{\Subwit{\Limpl{X_1}{X_2}}{\Limpl{Y_1}{Y_2}}}
=\Limpl{\Ret{\Subwit{X_1}{Y_1}}}{\Emb{\Subwit{X_1}{Y_1}}}$ and
$\Ret{\Subwit{\Limpl{X_1}{X_2}}{\Limpl{Y_1}{Y_2}}}
=\Limpl{\Emb{\Subwit{X_1}{Y_1}}}{\Ret{\Subwit{X_1}{Y_1}}}$ and $\Linarrow$
commutes with directed colimits in $\Embr\PCOH$.

\newcommand\Coalgmt[1]{\Coalgm{#1}}

This notion of inclusion on probabilistic coherence spaces extends to
coalgebras as outlined in Section~\ref{parag:sem-fix-points} (again, we refer
to~\cite{Ehrhard16a} for more details). We describe briefly this notion of
inclusion in the present concrete setting. 

Let $P$ and $Q$ be object of $\EM\PCOH$, we have $P\Subobj Q$ in
$\Embr{\EM\PCOH}$ if $\Coalgc P\Subobj\Coalgc Q$ and $\Coalgm
Q\Compl\Emb{\Subwit{\Coalgc P}{\Coalgc Q}}=\Excl{(\Emb{\Subwit{\Coalgc
      P}{\Coalgc Q}})}\Compl\Coalgm P$. The lub of a directed family
$(P_\gamma)_{\gamma\in\Gamma}$ of coalgebras (for this notion of substructure)
is the coalgebra $P=\bigcup_{\gamma\in\Gamma}P_\gamma$ defined by $\Coalgc
P=\bigcup_{\gamma\in\Gamma}\Coalgc{P_\gamma}$ and $\Coalgm P$ is characterized
by the equation $\Coalgm P\Compl\Emb{\Subwit{\Coalgc{P_\gamma}}{\Coalgc
    P}}=\Excl{\Emb{\Subwit{\Coalgc{P_\gamma}}{\Coalgc
      P}}}\Compl\Coalgm{P_\gamma}$ which holds for each $\gamma\in\Gamma$.

As outlined in Section~\ref{sec:cbpv-interpretation}, this allows to interpret
any type $\sigma$ as an object $\Tsem\sigma$ of $\PCOH$ and any positive type
$\phi$ as an object $\Tsemca\phi$ such that $\Coalgc{\Tsemca\phi}=\Tsem\phi$,
in such a way that the coalgebras $\Tsemca{\TREC\zeta\phi}$ and
$\Tsemca{\Subst{\phi}{\TREC\zeta\phi}{\phi}}$ are \emph{exactly the same
  objects of $\EM\PCOH$}. We use $\Coalgmt\phi$ for $\Coalgm{\Tsemca\phi}$.

\subsubsection{Dense coalgebras}\label{sec:reg-coalgebras}

Let $P$ be an object of $\EM\PCOH$, so that $P=(\Coalgc P,\Coalgm P)$ where
$\Coalgc P$ is a probabilistic coherence space and $\Coalgm P\in\PCOH(\Coalgc
P,\Excl{\Coalgc P})$ satisfies $\Digg{\Coalgc P}\Compl\Coalgm P=\Excl{\Coalgm
  P}\Compl\Coalgm P$. Given coalgebras $P$ and $Q$, a morphism
$t\in\PCOH(\Coalgc P,\Coalgc Q)$ is coalgebraic (that is $t\in\EM\PCOH(P,Q)$)
if $\Coalgm Q\Compl t=\Excl t\Compl\Coalgm P$. In particular, we say that
$u\in\Pcoh{(\Coalgc P)}$ is coalgebraic if, considered as a morphism from
$\One$ to $\Coalgc P$, $u$ belongs to $\EM\PCOH(\One,P)$. This means that
$\Prom u={\Coalgm P}\Compl u$. 

\begin{definition}\label{def:dense}
  Given an object $P$ of $\EM\PCOH$, we use $\PcohEM{(P)}$ for the set of
  coalgebraic elements of $\Pcoh{(\Coalgc P)}$.
\end{definition}

The following lemma is useful in the sequel and holds in any model of Linear
Logic. 

\begin{lemma}\label{lemma:sem-values}
  Let $X$ be a probabilistic coherence space, one has $\PcohEM{(\Excl
    X)}=\{\Prom u\St u\in\Pcoh X\}$. Let $P_\ell$ and $P_r$ be objects of
  $\EM\PCOH$.

  $\Tens{P_\ell}{P_r}$ is the cartesian product of $P_\ell$ and $P_r$
  in $\EM\PCOH$. The function
  $\PcohEM{(P_\ell)}\times\PcohEM{(P_r)}\to\PcohEM{(\Tens{P_\ell}{P_r})}$
  which maps $(u,v)$ to $\Tens uv$ is a bijection. The projections
  $\Projt i\in\EM\PCOH(\Tens{P_\ell}{P_r},P_i)$ are characterized by
  $\Projt i(\Tens{u_\ell}{u_r})=u_i$.

  The function
  $\{\ell\}\times\PcohEM{(P_\ell)}\cup\{r\}\times\PcohEM{(P_r)}
  \to\PcohEM{(\Plus{P_\ell}{P_r})}$
  which maps $(i,u)$ to $\Inj i(u)$ is a bijection. The injection
  $u\mapsto\Inj i(u)$ has a left inverse
  $\Proj i\in\PCOH(\Plus{\Coalgc{P_\ell}}{\Coalgc{P_r}},\Coalgc{P_i})$
  defined by $(\Proj i)_{(j,a),b}=\Kronecker{i}{j}\Kronecker{a}{b}$,
  which is not a coalgebra morphism in general.  
\end{lemma}
\begin{proof}
  Let $v\in\PcohEM{(\Excl X)}$, we have $\Prom
  v={\Coalgm{\Excl X}}\Compl v=\Digg X\Compl v$ hence $\Prom{(\Der X\Compl
    v)}=\Excl{\Der X}\Compl\Prom v=\Excl{\Der X}\Compl\Digg X v=v$. The
  other properties result from the fact that the Eilenberg-Moore category
  $\EM\PCOH$ is cartesian and co-cartesian with $\ITens$ and $\IPlus$ as
  product and co-product, see~\cite{Mellies09} for more details.
\end{proof}
Because of these properties we write sometimes $(u_\ell,u_r)$ instead of
$\Tens{u_\ell}{u_r}$ when $u_i\in\PcohEM{P_i}$ for $i\in\{\ell,r\}$.

\begin{definition}\label{def:coalg-dense}
  An object $P$ of $\EM\PCOH$ is \emph{dense} if, for any object
  $Y$ of $\PCOH$ and any two morphisms $t,t'\in\PCOH(\Coalgc P,Y)$, if $t\Compl
  u=t'\Compl u$ for all $u\in\PcohEM{(P)}$, then $t=t'$.
\end{definition}

\begin{theorem}\label{th:coalg-dense-colsed}
  For any probabilistic coherence space $X$, $\Excl X$ is a dense
  coalgebra. If $P_\ell$ and $P_r$ are dense coalgebras then $\Tens{P_\ell}{P_r}$
  and $\Plus{P_\ell}{P_r}$ are dense. The colimit in 
  $\Embr{(\EM\PCOH)}$ of a directed family of dense coalgebras is 
  dense.
\end{theorem}
\begin{proof}
  Let $X$ be an object of $\PCOH$, one has $\PcohEM{(\Excl X)}=\{\Prom u\St
  u\in\Pcoh{X}\}$ by Lemma~\ref{lemma:sem-values}.  It follows that $\Excl X$
  is a dense coalgebra by Theorem~\ref{th:pcoh-functional}. Assume that $P_\ell$
  and $P_r$ are dense coalgebras. Let
  $t,t'\in\PCOH(\Tens{\Coalgc{P_\ell}}{\Coalgc{P_r}},Y)$ be such that $t\Compl
  w=t'\Compl w$ for all $w\in\PcohEM{(\Tens{P_\ell}{P_r})}$. We have
  $\Curlin{(t)},\Curlin{(t')}\in\PCOH(\Coalgc{P_\ell},\Limpl{\Coalgc P_r}{Y})$ so,
  using the density of $P_\ell$, it suffices to prove that $\Curlin{(t)}\Compl
  u_\ell=\Curlin{(t')}\Compl u_\ell$ for each $u_\ell\in\PcohEM{(P_\ell)}$. So let
  $u_\ell\in\PcohEM{(P_\ell)}$ and let $s=\Curlin{(t)}\Compl u_\ell$ and
  $s'=\Curlin{(t')}\Compl u_\ell$. Let $u_r\in\PcohEM{(P_r)}$, we have $s\Compl
  u_r=t\Compl(\Tens{u_\ell}{u_r})=t'\Compl(\Tens{u_\ell}{u_r})=s'\Compl u_r$ since
  $\Tens{u_\ell}{u_r}\in\PcohEM{(\Tens{P_\ell}{P_r})}$ and therefore $s=s'$ since
  $P_r$ is dense. Let now $t,t'\in\PCOH(\Plus{\Coalgc{P_\ell}}{\Coalgc{P_r}},Y)$
  be such that $t\Compl w=t'\Compl w$ for all
  $w\in\PcohEM{(\Plus{P_\ell}{P_r})}$. To prove that $t=t'$, it suffices to prove
  that $t\Compl\Inj i=t'\Compl\Inj i$ for $i\in\{\ell,r\}$. Since $P_i$ is dense, it
  suffices to prove that $t\Compl\Inj i\Compl u=t'\Compl\Inj i\Compl u$ for
  each $u\in\PcohEM{(P_i)}$ which follows from the fact that $\Inj i\Compl
  u\in\PcohEM{P_i}$. Last let $(P_\gamma)_{\gamma\in\Gamma}$ be a directed
  family of dense coalgebras (in $\Embr{\EM\PCOH}$) and let
  $P=\bigcup_{\gamma\in\Gamma}P_\gamma$, and let $t,t'\in\EM\PCOH(\Coalgc P,Y)$
  be such that $t\Compl w=t'\Compl w$ for all $w\in\PcohEM{(P)}$. It suffices
  to prove that, for each $\gamma\in\Gamma$, one has
  $t\Compl\Emb{\Subwit{\Coalgc{P_\gamma}}{\Coalgc
      P}}=t'\Compl\Emb{\Subwit{\Coalgc{P_\gamma}}{\Coalgc P}}$ and this results
  from the fact that $P_\gamma$ is dense and
  $\Emb{\Subwit{\Coalgc{P_\gamma}}{\Coalgc P}}$ is a coalgebra morphisms (and
  therefore maps $\PcohEM{(P_\gamma)}$ to $\PcohEM{(P)}$).
\end{proof}

The sub-category $\EM\PCOH$ of dense coalgebras is cartesian and
co-cartesian and is well-pointed by
Theorem~\ref{th:coalg-dense-colsed}. We use $\EMR\PCOH$ for this
sub-category and $\Embr{(\EMR\PCOH)}$ for the sub-class of
$\Embr{\EM\PCOH}$ whose objects are the dense coalgebras (with the
same order relation).

\subsubsection{Interpreting types and terms in $\PCOH$}\label{subsec:sem-type-term}

Given a type $\sigma$ with free type variables contained in the repetition-free
list $\Vect\zeta$, and given a sequence $\Vect P$ of length $n$ of objects of
$\EM\PCOH$, we define $\Tsem\sigma_{\Vect\zeta}(\Vect P)$ as an object of
$\PCOH$ and when $\phi$ is a positive type (whose free variables are contained
in $\Vect\zeta$) we define $\Tsemca\phi_{\Vect\zeta}(\Vect P)$ as an object of
$\EM\PCOH$. These operations are continuous and their definition follows the
general pattern described in Section~\ref{sec:cbpv-interpretation}.

\begin{theorem}\label{th:pos-types-dense}
  Let $\phi$ be a positive type and let $\Vect\zeta=(\List\zeta 1n)$
  be a repetition-free list of type variables which contains all the
  free variables of $\phi$. Let $\Vect P$ be a sequence of $n$ dense
  coalgebras. Then $\Tsemca\phi_{\Vect\zeta}(\Vect P)$ is a dense
  coalgebra. In particular, when $\phi$ is closed, the coalgebra~
  $\Tsemca\phi$ is dense.
\end{theorem}
This is an immediate consequence of the definition of $\Tsemca\phi$ and
of Theorem~\ref{th:coalg-dense-colsed}.

\begin{remark}
  It turns out that the interpretation of positive types in the
  model $\PCOH/\EM\PCOH$ are dense coalgebras. This is mainly due to
  the fact that the colimit of a directed family of dense coalgebras
  \emph{in the partially ordered class $\Embr{\EM\PCOH}$} is dense,
  see Theorem~\ref{th:coalg-dense-colsed}. From the viewpoint of
  Levy's CBPV~\cite{LevyP04}, whose semantics is described in terms of
  adjunctions, we are using a resolution of the comonad ``$\oc$''
  through the category $\EM\PCOH$ (or, equivalently, through the
  category $\EMR\PCOH$). As pointed out to us by one of the referees
  and already mentioned in the introduction, there is another
  resolution through a category of families and introduced
  in~\cite{Abramsky1998}, which is initial among all resolutions that
  model CBPV. This other option will be explored in further work.
\end{remark}

Then $\BOOL=\PLUS\ONE\ONE$ satisfies
$\Web{\Tsem\BOOL}=\{(\ell,\Onelem),(r,\Onelem)\}$ and
$u\in\Realpto{\Web{\Tsem\BOOL}}$ satisfies $u\in\Pcoh{\Tsem\BOOL}$ iff
$u_{(\ell,\Onelem)}+u_{(r,\Onelem)}\leq 1$. The coalgebraic structure of this
object is given by
\begin{align*}
  (\Coalgmt\BOOL)_{(j,\Onelem),\Mset{(j_1,\Onelem),\dots,(j_k,\Onelem)}}=
  \begin{cases}
    1 & \text{if }j=j_1=\dots=j_k\\
    0 & \text{otherwise.}
  \end{cases}
\end{align*}
  The object $\Snat=\Tsem\NAT$ satisfies
$\Snat=\PLUS\One\Snat$ so that
$\Web\SNat=\{(\ell,\Onelem),(r,(\ell,\Onelem)),(r,(r,(\ell,\Onelem))),\dots\}$
and we use $\Snum n$ for the element of $\Web\Snat$ which has $n$
occurrences of $r$. Given $u\in\Realpto{\Web\Snat}$, we use $l(u)$ for
the element of $\Realpto{\Web\Snat}$ defined by
$l(u)_{\Snum n}=u_{\Snum{n+1}}$. By definition of $\Snat$, we have
$u\in\Pcoh\Snat$ iff $u_{\Snum 0}+\Norm{l(u)}_\Snat\leq 1$, and then
$\Norm u_\Snat=u_{\Snum 0}+\Norm{l(u)}_\Snat$. It follows that
$u\in\Pcoh\Snat$ iff $\sum_{n=0}^\infty u_{\Snum n}\leq 1$ and of
course $\Norm u_\Snat=\sum_{n=0}^\infty u_{\Snum n}$. Then the
coalgebraic structure $\Coalgmt\NAT$ is defined exactly as
$\Coalgmt\BOOL$ above. In the sequel, we identify $\Web\Snat$ with
$\Nat$.

Given a typing context $\cP=(x_1:\phi_1,\dots,x_k:\phi_k)$, a type $\sigma$ and a term
$M$ such that $\TSEQ\cP M\sigma$, $M$ is interpreted as a morphism $\Psem
M^\cP\in\PCOH(\Tsem\cP,\Tsem\sigma)$. For all constructs of the language but
probabilistic choice, this interpretation uses the generic structures of the
model described in Section~\ref{sec:LL-semantics-short}, the description of
this interpretation can be found in~\cite{Ehrhard16a}. 
We set $\Psem{\COIN p}{}=p\Base{(\ell,*)}+(1-p)\Base{(r,*)}$.

If $\TSEQ{x_1:\phi_1,\dots,x_k:\phi_k}{M}{\sigma}$, the morphism $\Psem{M}^\cP$
is completely characterized by its values on
$(u_1,\dots,u_k)\in\PcohEM{(\Tsemca\cP)}$. We describe now the interpretation of
terms using this observation.

\pagebreak
\begin{itemize}
\item $\Psem\ONELEM=1\in\Pcoh{\One}=[0,1]$.
\item $\Psem{x_i}^\cP(\List u1k)=u_i$.
\item $\Psem{\STOP N}^\cP(\List u1k)=\Prom{(\Psem N^\cP(\List u1k))}$.
\item $\Psem{\PAIR{M_\ell}{M_r}}^\cP(\List u1k)=
  \Tens{\Psem{M_\ell}^\cP(\List u1k)}{\Psem{M_r}^\cP(\List u1k)}$.
\item $\Psem{\IN i{N}}^\cP(\List u1k)=\Inj i(\Psem N^\cP(\List u1k))$, $i\in\{\ell,r\}$.
\item $\Psem{\GO N}^{\cP}(\List u1k)=\Der{\Tsem\sigma}(\Psem{N}^\cP(\List
  u1k))$, assuming that $\TSEQ{\cP}{N}{\EXCL\sigma}$.
\item If $\TSEQ{\cP}{N}{\LIMPL{\phi}{\sigma}}$ and $\TSEQ{\cP}{R}{\phi}$ then
  $\Psem{N}^\cP(\List u1k)\in\Pcoh{(\Limpl{\Tsem\phi}{\Tsem\sigma})}$,
and  $\Psem{R}^\cP(\List u1k)\in\Pcoh{(\Tsem\phi)}$ and using the application of a matrix to a vector we have $\Psem{\LAPP
    NR}^\cP(\List u1k)=\Psem{N}^\cP(\List u1k)\Compl\Psem{R}^\cP(\List u1k)$.
\item If $\TSEQ{\cP,x:\phi}{N}{\sigma}$ then $\Psem{\ABST{x}{\phi}{N}}^\cP(\List
  u1k)\in\Pcoh{(\Limpl{\Tsem\phi}{\Tsem\sigma})}$ is completely described
  by the fact that, for all $u\in\PcohEM{(\Tsemca\phi)}$, one has
  $\Psem{\ABST{x}{\phi}{N}}^\cP(\List u1k)\Compl u=\Psem{N}^{\cP,x:\phi}(\List
  u1k,u)$. This is a complete characterization of this interpretation by
  Theorem~\ref{th:pos-types-dense}.
\item If $\TSEQ{\cP}{N}{\PLUS{\phi_\ell}{\phi_r}}$ and
  $\TSEQ{\cP,y_i:\phi_i}{R_i}{\sigma}$ for $i\in\{\ell,r\}$, then\\
  $\Psem{\CASE{N}{y_\ell}{R_\ell}{y_r}{R_r}}^\cP(\List
  u1k)=\Psem{R_\ell}^{\cP,y_\ell:\phi_\ell}(\List u1k,\Proj \ell(\Psem{N}^\cP(\List
  u1k)))+\Psem{R_r}^{\cP,y_r:\phi_r}(\List u1k,\Proj r(\Psem{N}^\cP(\List
  u1k)))$ where $\Proj i\in\PCOH{(\Plus{P_\ell}{P_r},P_i)}$ is the $i$th
  ``projection'' introduced in~\ref{sec:reg-coalgebras}, left inverse for $\Inj
  i$.
\item If $\TSEQ{\cP,x:\EXCL\sigma}{N}{\sigma}$ then
  $\Psem{N}^{\cP,x:\EXCL\sigma}
  \in\PCOH(\Tens{\Coalgc{\Tsem\cP}}{\Excl{\Tsem\sigma}},\Tsem\sigma)$ and
  $\Psem{\FIXT x\sigma N}^\cP(\List u1k)=\sup_{n=0}^\infty f^n(0)$ where
  $f:\Pcoh{\Tsem\sigma}\to \Pcoh{\Tsem\sigma}$ is the Scott-continuous function
  given by $f(u)=\Psem{N}^{\cP,x:\Excl\sigma}(\List u1k,\Prom u)$.
\item If $\TSEQ{\cP}{N}{\Subst\psi{\TREC\zeta\psi}\zeta}$ then $\Psem{\FOLD
    N}^\cP=\Psem N^\cP$ which makes sense since
  $\Tsem{\Subst\psi{\TREC\zeta\psi}\zeta}=\Tsem{\TREC\zeta\psi}$.
\item If $\TSEQ{\cP}{N}{\TREC\zeta\psi}$ then $\Psem{\UNFOLD
    N}^\cP=\Psem N^\cP$.
\end{itemize}

\begin{theorem}[Soundness]\label{th:soundness}
  If $M$ satisfies $\TSEQ{\cP}{M}{\sigma}$ then
  \begin{align*}
    \Psem
    M^\cP=\sum_{\TSEQ{\cP}{M'}{\sigma}}\Redmats_{M,M'}\Psem{M'}^{\cP}
  \end{align*}
\end{theorem}
The proof is done by induction and is a straightforward verification.

\begin{corollary}\label{th:soundness-ineq}
  Let $M$ be a term such that $\Tseq{}{M}{\One}$ so that $\Psem M\in\Izu$. Then
  $\Psem M\geq\Redmats^\infty_{M,\ONELEM}$.
\end{corollary}
This is an immediate consequence of Theorem~\ref{th:soundness} and of the
definition of $\Redmats^\infty$, see Section~\ref{sec:obs-eq}.

\subsection{Examples of term interpretations}\label{subsec:exden}

We give the interpretation of terms that we gave as examples in Subsection~\ref{subsec:exsyn}.
\begin{itemize}
\item $\Psem{\LOOP\sigma}=\Psem{\FIXT x\sigma{\GO x}}=0$
\item $\Psem{\True}= \Base{(\ell,*)}$ and $\Psem{\False}=\Base{(r,*)}$
\item $\Psem{\IFB M{N_\ell}{N_r}}^\cP(\List u1k)=\Psem{M}^\cP_{(\ell,*)}(\List u1k)\Psem{N_\ell}^\cP(\List u1k)\\+\Psem{M}^\cP_{(r,*)}(\List u1k)\Psem{N_r}^\cP(\List u1k)$
\item $\Psem{\DICE p{M_\ell}{M_r}}^\cP(\List u1k)=p\Psem{M_\ell}^\cP(\List
  u1k)+(1-p)\Psem{M_r}^\cP(\List u1k)$
\item  $\Psem{\NUM n}=\Snum n$ for $n\in\Nat$
\item $\Psem{\TSUCC M}^\cP_{n+1}(\List u1k)= \Psem{M}_n(\List u1k)$
\item $\Psem{\IFV M{N_\ell}x{N_r}}^\cP(\List u1k)=\Psem{M}^\cP_0(\List u1k)\Psem{N_\ell}^\cP(\List u1k)\\+\sum_{n=0}^\infty\Psem{M}^\cP_{n+1}(\List u1k)\Psem{N_r}^\cP(\List u1k)(\Snum n)$
\item $\Psem{\Ran{\Vect p}}=\sum_{i=1}^n p_i \Base{\Snum i}$
\item $\Psem{\LAPP{\Probe_\ell}M}^\cP(\List u1k)= \Psem{M}^\cP(\List u1k)_\ell\Base{*}$
\item $\Psem{M_0\cdot N}^\cP(\List u1k)=\Psem{M_0}^\cP(\List u1k)\Psem N^\cP(\List u1k)$
\item $\Psem{M_0\AND\cdots\AND M_l}^\cP(\List u1k)=\prod_{i=0}^l\Psem{M_i}^\cP(\List u1k)$
\item $\Psem{\LAPP{\Pchoose^\sigma_{l+1}(\List N0l)}P}^\cP(\List u1k)=\sum_{i=0}^l \Psem{P}_i^\cP(\List u1k)\cdot  \Psem{N_i}^\cP(\List u1k)$
\item $\forall u \in\Pcoh{(\Tsem{\Tnat})},\ \Psem{\Pext lr}(u)=\sum_{i=l}^ru_i \Base{\Snum i}$ and \[\Psem{\Pwin l{\Vect n}}(u)=\sum_{i=n_1+\cdots+n_{l-1}+1}^{n_1+\cdots+n_l}u_i \Base{\Snum{i}}\]
\end{itemize}

\section{Adequacy}\label{sec:adequacy}

Our goal is to prove the converse of Corollary~\ref{th:soundness-ineq}: for any
closed term $M$ such that $\TSEQ{}{M}{\ONE}$, the probability that $M$ reduces
to $\ONELEM$ is larger than or equal to
$\Psem M{}\in\Pcoh{\Tsem\ONE}\Isom[0,1]$, so that we shall know that these two
numbers are actually equal.

In spite of its very simple statement, the proof of this property is rather
long mainly because we have to deal with the recursive type definitions allowed
by our syntax. As usual, the proof is based on the definition of a logical
relations between terms and elements of the model (more precisely, given any
type $\sigma$, we have to define a relation between closed terms of types
$\sigma$ and elements of $\Pcoh{\Tsem\sigma}$; let us call such a relation a
\emph{$\sigma$-relation}).

Since we have no positivity restrictions on the occurrence of type variables
wrt.~which recursive types are defined so that types are neither covariant nor
contravariant wrt.~these type variables, we use a very powerful technique
introduced in~\cite{Pitts93} for defining this logical relation. 

Indeed a type variable $\zeta$ can have positive and negative
occurrences in a positive\footnote{Warning: the word ``positive'' has
  two different meanings here!} type $\phi$, consider for instance the
case $\phi=\EXCL{(\LIMPL\zeta\zeta)}$ where the type variable $\zeta$
has a positive (on the right of the $\LIMPL{}{}$) and a negative
occurrence (on the left). To define the logical relation associated
with $\TREC\zeta\phi$, we have to find a fixpoint for the operation
which maps a $(\TREC\zeta{\phi})$-relation $\cR$ to the relation
$\Phi(\cR)=\EXCL{(\LIMPL{\cR}{\cR})}$ (which can be defined using
$\cR$ as a ``logical relation'' in a fairly standard way). Relations
are naturally ordered by inclusion, and this strongly suggests to
define the above fixpoint using this order relation by \Eg{}~Tarski's
Fixpoint Theorem. The problem however is that $\Phi$ is neither a
monotone nor an anti-monotone operation on relations, due to the fact
that $\zeta$ has a positive and a negative occurrence in $\phi$.

It is here that Pitts's trick comes in: we replace the relations $\cR$ with
pairs of relations $\cR=(\Nrel\cR,\Prel\cR)$ ordered as follows:
$\cR\Subrel\cS$ if $\Prel\cR\subseteq\Prel\cS$ and $\Nrel\cS\subseteq\Nrel\cR$.
Then we define accordingly $\Phi(\cR)$ as a pair of relations by
$\Nrel{\Phi(\cR)}=\EXCL{(\LIMPL{\Prel\cR}{\Nrel\cR})}$ and
$\Prel{\Phi(\cR)}=\EXCL{(\LIMPL{\Nrel\cR}{\Prel\cR})}$. Now the operation
$\Phi$ is monotone wrt.~the $\Subrel$ relation and it becomes possible to apply
Tarski's Fixpoint Theorem to $\Phi$ and get a pair of relations $\cR$ such that
$\cR=\Phi(\cR)$. The next step consists in proving that
$\Nrel\cR=\Prel\cR$. This is obtained by means of an analysis of the definition
of the interpretation of fixpoints of types as colimits in the category
$\Embr\PCOH$. One is finally in position of proving a fairly standard ``Logical
Relation Lemma'' from which adequacy follows straightforwardly.

In this short description of our adequacy proof, many technicalities have
obviously been hidden, the most important one being that values are handled in
a special way so that we actually consider two kinds of pairs of
relations. Also, a kind of ``biorthogonality closure'' plays an essential role
in the handling of positive types, no surprise for the readers acquainted with
Linear Logic, see for instance the proof of normalization in~\cite{Girard87}.

\subsection{Pairs of relations and basic operations}
Given a \emph{closed} type $\sigma$, we define $\Rels\sigma$ as the set of all
pairs of relations $\cR=(\Nrel\cR,\Prel\cR)$ such that, for
$\epsilon\in\{\POS,\NEG\}$, each element of $\Srel\cR\epsilon$ is a pair
$(M,u)$ where $\TSEQ{}{M}{\sigma}$ and $u\in\Pcoh{\Tsem\sigma}$. For a
\emph{closed} positive type $\phi$, we also define $\Relsv\phi$ as the set of
all pairs of relations $\cV=(\Nrel\cV,\Prel\cV)$ such that, for
$\epsilon\in\{\POS,\NEG\}$, each element of $\Srel\cV\epsilon$ is a pair
$(V,v)$ where $\TSEQ{}{V}{\phi}$ is a value and $v\in\PcohEM{\Tsem\phi}$.

Given $\cR,\cS\in\Rels\sigma$, we write $\cR\Subrel\cS$ if
$\Prel\cR\subseteq\Prel\cS$ and $\Nrel\cS\subseteq\Nrel\cR$. We define similarly
$\cV\Subrel\cW$ for $\cV,\cW\in\Relsv\phi$. Then $\Rels\sigma$ is a complete
meet-lattice, the infimum of a collection $(\cR_i)_{i\in I}$ being
$\Infrel_{i\in I}\cR_i=(\bigcup_{i\in I}\Nrel\cR_i,\bigcap_{i\in
  I}\Prel\cR_i)$. The same holds of course for $\Relsv\phi$ and we use the same
notations.

We define $\cR(\ONE)$ as the set of all pairs $(M,p)$ such that
$\TSEQ{}{M}{\ONE}$, $p\in[0,1]$ and $\Redmats^\infty_{M,\ONELEM}\geq p$.

We define in Figure~\ref{fig:adeq-rel-log-op} logical operations on these
pairs of relations.  The last one is the aforementioned
biorthogonality closure operation on pairs of relations.
\begin{figure}
  \centering
\begin{itemize}
\item Let $\cR\in\Rels\sigma$, we define $\EXCL\cR\in\Relsv{\EXCL\sigma}$ by:
  $\Srel{\EXCL\cR}\epsilon=\{(\STOP M,\Prom u)\St (M,u)\in\Srel\cR\epsilon\}$
  for $\epsilon\in\{\NEG,\POS\}$.
\item Let $\cV_i\in\Relsv{\phi_i}$ for $i\in\{\ell,r\}$. We define
  $\Srel{(\TENS{\cV_\ell}{\cV_r})}\epsilon=\{(\PAIR{V_\ell}{V_r},\Tens{v_\ell}{v_r})\}\St
  (V_i,v_i)\in\Srel{\cV_i}\epsilon\}$ for $\epsilon\in\{\NEG,\POS\}$, so that
  $\TENS{\cV_\ell}{\cV_r}\in\Relsv{\TENS{\phi_\ell}{\phi_r}}$. 
\item Let $\cV_i\in\Relsv{\phi_i}$ for $i\in\{\ell,r\}$. We define
  $\Srel{(\PLUS{\cV_\ell}{\cV_r})}\epsilon=\{(\IN i{V},\Inj i(v))\}\St
  i\in\{\ell,r\}\text{ and } (V,v)\in\Srel{\cV_i}\epsilon\}$ for
  $\epsilon\in\{\NEG,\POS\}$, so that
  $\PLUS{\cV_\ell}{\cV_r}\in\Relsv{\PLUS{\phi_\ell}{\phi_r}}$.
\item Let $\cV\in\Relsv\phi$ and $\cR\in\Rels\sigma$. We define
  $\LIMPL\cV\cR\in\Rels{\LIMPL\phi\sigma}$ as follows:
  $\Srel{(\LIMPL\cV\cR)}\epsilon=\{(M,u)\St\
  \TSEQ{}{M}{\LIMPL\phi\sigma},\
  u\in\Pcoh{\Tsem{\LIMPL\phi\sigma}}\text{ and }\forall
  (V,v)\in\Srel\cV{-\epsilon}\ (\LAPP MV,u\Matapp
  v)\in\Srel\cR\epsilon\}$.
\item Last, given $\cV\in\Relsv{\phi}$, we define $\Crel\cV\in\Rels{\phi}$ as
  follows: $\Srel{\Crel\cV}\epsilon$ is the set of all $(M,u)$ such that
  $\TSEQ{}{M}{\phi}$, $u\in\Pcoh{\Tsem\phi}$ and, for all
  $(T,t)\in\Srel{(\LIMPL\cV\cR(\ONE))}{-\epsilon}$, one has $(\LAPP
  TM,\Matapp tu)\in\cR(\ONE)$.
\end{itemize}  
  \caption{Logical operations for pairs of relations}
  \label{fig:adeq-rel-log-op}
\end{figure}

Observe that all these operations are monotone wrt.~$\Subrel$. For instance
$\cV\Subrel\cW\Andc\cR\Subrel\cS\Implies(\LIMPL\cV\cR)\Subrel(\LIMPL\cW\cS)$,
and $\cV\Subrel\cW\Implies\Crel\cV\Subrel\Crel\cW$.

\subsection{Fixpoints of pairs of relations}

To deal with fixpoint types $\TREC\zeta\phi$, we need to consider types
parameterized by relations as follows.

Let $\sigma$ be a type and let $\Vect\zeta=(\List\zeta 1n)$ be a list of type
variables without repetitions and which contains all free variables of
$\sigma$. For all list $\Vect\phi=(\List\phi 1n)$ of \emph{closed} positive
types, we define
\begin{align*}
  \Trel\sigma{\Vect\zeta}:\prod_{i=1}^n\Relsv{\phi_i}
  \to\Rels{\Subst\sigma{\Vect\phi}{\Vect\zeta}}\,.
\end{align*}
Let also $\phi$ be a positive type whose free variables are contained in
$\Vect\zeta$, we define
\begin{align*}
  \Trelv\phi{\Vect\zeta}:\prod_{i=1}^n\Relsv{\phi_i}
  \to\Relsv{\Subst\phi{\Vect\phi}{\Vect\zeta}}\,.
\end{align*}
The definition is by simultaneous induction on $\sigma$ and $\phi$.  All cases
but one consist in applying straightforwardly the above defined logical
operations on pairs of relations, for instance
\begin{align*}
  \Trel{\LIMPL\phi\tau}{\Vect\zeta}(\Vect\cV)
  =\LIMPL{\Trelv\phi{\Vect\zeta}(\Vect\cV)}{\Trel\sigma{\Vect\zeta}(\Vect\cV)}
  \quad\text{and}\quad
  \Trel{\phi}{\Vect\zeta}(\Vect\cV)=\Crel{\Trelv\phi{\Vect\zeta}(\Vect\cV)}\,.
\end{align*}
We are left with the case of recursive definitions of types, so assume that
$\phi=\TREC\zeta\psi$. Let $\Vect\phi=(\List\phi 1n)$ be a list of closed
positive types and let $\Vect\cV\in\prod_{i=1}^n\Relsv{\phi_i}$, we set
\begin{align}
  \label{eq:bi-rel-fixpoint-def}
  \Trelv\phi{\Vect\zeta}(\Vect\cV)=
  \Infrel\{\cV\in\Relsv{\Subst\phi{\Vect\phi}{\Vect\zeta}}\St
  \FOLD{\Trelv{\psi}{\Vect\zeta,\zeta}(\Vect\cV,\cV)}\Subrel\cV\}
\end{align}
where we use the following notation: given
$\cW\in\Relsv{\Substbis\psi{\Vect\phi/\Vect\zeta,\Subst\phi{\Vect\phi}{\Vect\zeta}/\zeta}}$,
$\FOLD{\cW}\in\Relsv{\Subst{\phi}{\Vect\phi}{\Vect\zeta}}$ is given by
$\Srel{\FOLD\cW}{\epsilon}=\{(\FOLD W,v)\St(W,v)\in\Srel\cW\epsilon\}$ for
$\epsilon\in\{\POS,\NEG\}$. 

We recall the statement of Tarski's fixpoint theorem.
\begin{theorem}
  Let $S$ and $T$ be complete meet semi-lattices and let $f:S\times T\to T$ be
  a monotone function. For $x\in S$, let $g(x)$ be the meet of the set $\{y\in
  T\St f(x,y)\leq y\}$. Then the function $g$ is monotone and satisfies
  $f(x,g(x))=g(x)$ for each $x\in S$.
\end{theorem}

Applying this theorem we obtain, by induction on types, the following property.

\begin{lemma}\label{lemma:bi-rel-monotone}
  For any type $\sigma$ and any positive type $\phi$, the maps
  $\Trel\sigma{\Vect\zeta}$ and $\Trelv\phi{\Vect\zeta}$ are monotone wrt.~the
  $\Subrel$ order relation. If $\psi$ is a positive type,
  $\Vect\zeta=(\List\zeta 1n,\zeta)$ is a repetition-free list of type
  variables containing all the free variables of $\psi$ and
  $\phi=\TREC\zeta\psi$ and $\Vect\cV=(\List\cV 1n)$ is a list of pairs of
  relations such that $\cV_i\in\Relsv{\phi_i}$ for each $i$, then
  $\cV=\Trel\phi{\Vect\zeta}(\Vect\cV)$ satisfies
  $\cV=\FOLD{\Trel{\psi}{\Vect\zeta,\zeta}}(\Vect\cV,\cV)$.
\end{lemma}

\subsection{Some useful closeness lemmas}
We state and prove a series of lemmas expressing that our pairs of relations
are closed under various syntactic and semantic operations.

\begin{lemma}\label{lemma:prob-conv-indep-det}
  Let $M$ and $M'$ be terms such that $\TSEQ{}{M}{\ONE}$ and
  $\TSEQ{}{M'}{\ONE}$. If $M\Rel\Wred M'$ then $\Redmats^\infty_{M,\ONELEM}
  =\Redmats^\infty_{M',\ONELEM}$.
\end{lemma}
This is straightforward since any reduction path from $M$ to $\ONELEM$ must
start with the step $M\Rel\Wred M'$, and this is a probability $1$ step.

\begin{lemma}\label{lemma:rel-app-closeness}
  Let $\phi$ be a closed positive type and let $\sigma$ be a closed type. Let
  $(M,u)\in\Srel{\Trel\phi{}}{-\epsilon}$ and
  $(R,r)\in\Srel{\Trel{\LIMPL\phi\sigma}{}}{\epsilon}$. Then $(\LAPP
  RM,\Matapp ru)\in\Srel{\Trel{\sigma}{}}{\epsilon}$. 
\end{lemma}
\begin{proof}
  We can write $\sigma=\LIMPL{\phi_1}{\LIMPL\cdots{\LIMPL{\phi_n}\psi}}$ for
  some $n$ and $\List\phi 1n,\psi$ positive and closed. Given
  $(V_i,v_i)\in\Srel{\Trelv{\phi_i}{}}{-\epsilon}$ for $i=1,\dots,n$, we have
  to prove that
  \begin{align*}
    (\LAPP R{M\,V_1\cdots V_n},\Matapp ru\Appsep v_1\cdots v_n)
    \in\Srel{(\Crel{\Trelv{\psi}{}})}{\epsilon}
  \end{align*}
    so let $(T,t)\in\Srel{(\LIMPL{\Trelv\psi{}}{\Trel\ONE{}})}{-\epsilon}$, we
  have to prove that
  \begin{align*}
    (\LAPP{T}{(\LAPP R{M\,V_1\cdots V_n})},t(\Matapp ru\Appsep v_1\cdots v_n))
    \in\Trel{\ONE}{}\,.
  \end{align*}
  Let $S=\ABST{x}{\phi}{\LAPP T{(\LAPP R{x\,V_1\cdots V_n})}}$ so that
  $\TSEQ{}{S}{\LIMPL\phi\ONE}$. Similarly let
  $s\in\Pcoh{\Tsem{\LIMPL\phi\ONE}}$ be the linear morphism defined by
  $\Matapp s{u'}=t(\Matapp r{u'}\Appsep v_1\cdots v_n)$ (the fact that $s$ so
  defined is actually a morphism in $\PCOH$ results from the symmetric
  monoidal closeness of that category and from the fact that $r$ and
  $t$ are morphisms in $\PCOH$).
  Let $(V,v)\in\Srel{\Trelv\phi{}}{-\epsilon}$, we have
  $(\LAPP RV,\Matapp rv)\in\Srel{\Trel\sigma{}}{\epsilon}$ and hence
  $(\LAPP R{V\Appsep V_1\cdots V_n}),\Matapp rv\Appsep v_1\cdots
  v_n\in\Srel{\Trel{\psi}{}}{\epsilon}$. Therefore
  \begin{align*}
    (\LAPP T{(\LAPP R{V\Appsep V_1\cdots
        V_n})}),t(\Matapp rv\Appsep v_1\cdots v_n)\in\Trel{\ONE}{}
  \end{align*}
  since we have assumed that
  $(T,t)\in\Srel{(\LIMPL{\Trelv\psi{}}{\Trel\ONE{}})}{-\epsilon}$. Since
  $t(\Matapp rv\Appsep v_1\cdots v_n)=\Matapp sv$, and by
  Lemma~\ref{lemma:prob-conv-indep-det}, it follows that
  $(\LAPP SV,\Matapp sv)\in\Trel\ONE{}$. Hence
  $(S,s)\in\Srel{\Trel{\LIMPL\phi\ONE}{}}{\epsilon}$ and therefore
  $(\LAPP SM,\Matapp su)\in\Trel\ONE{}$ since we have
  $(M,u)\in\Srel{\Trel\phi{}}{-\epsilon}$.

  We finish the proof by observing that
  $\Matapp su=t(\Matapp ru\Appsep v_1\cdots v_n)$ and by showing that
  \begin{align*}
    \Redmats^\infty_{\LAPP T{\LAPP R{M\,V_1\cdots V_n}},\ONELEM}
    =\Redmats^\infty_{\LAPP SM,\ONELEM}
  \end{align*}
  For this it suffices to observe (by inspection of the reduction rules) that
  each reduction path $\pi$ from $\LAPP T{(\LAPP R{M\,V_1\cdots V_n})}$ to
  $\ONELEM$ is of shape $\pi=\lambda\rho$ where
  \begin{itemize}
  \item $\lambda$ is a reduction path \[\LAPP T{(\LAPP R{M_1\,V_1\cdots
        V_n})}\Rel{\Redone{p_1}}\LAPP T{(\LAPP R{M_2\,V_1\cdots
        V_n})}\Rel{\Redone{p_2}}\cdots\Rel{\Redone{p_k}}\LAPP T{(\LAPP
      R{M_{k+1}\,V_1\cdots V_n})}\] where $M_1=M$, $M_{k+1}$ is a value $V$ and
    $M=M_1\Rel{\Redone{p_1}}M_2\Rel{\Redone{p_2}}\cdots\Rel{\Redone{p_{k}}}
    M_{k+1}=V$
  \item and $\rho$ is a reduction path from $\LAPP T{(\LAPP R{V\,V_1\cdots
        V_n})}$ to $\ONELEM$.
  \end{itemize}  
  Then we have $\LAPP SM=\LAPP S{M_1}\Rel{\Redone{p_1}}\LAPP
  S{M_2}\Rel{\Redone{p_2}}\cdots\Rel{\Redone{p_{k}}}\LAPP SV\Rel{\Redone
    1} \LAPP T{(\LAPP R{V\,V_1\cdots V_n})}$, the last step resulting from the
  definition of $S$. In that way, we have defined a probability preserving
  bijection between the reduction paths from $\LAPP T{(\LAPP R{M_1\,V_1\cdots
      V_n})}$ to $\ONELEM$ and the reduction paths from $\LAPP SM$ to
  $\ONELEM$, proving our contention.
\end{proof}

\begin{lemma}\label{lemma:tens-crel}
  Let $\phi_i$ be closed positive types and
  $(M_i,u_i)\in\Srel{\Trel{\phi_i}{}}{\epsilon}$ for $i\in\{\ell,r\}$. Then
  $(\PAIR{M_\ell}{M_r},\TENS{u_\ell}{u_r})
  \in\Srel{\Trel{\TENS{\phi_\ell}{\phi_r}}{}}{\epsilon}$.
\end{lemma}
\begin{proof}
  Let
  $(T,t)\in\Srel{(\Limpl{\Trelv{\TENS{\phi_\ell}{\phi_r}}{}}
    {\Trel{\ONE}{}})}{-\epsilon}$,
  we must prove that
  $(\LAPP
  T{\PAIR{M_\ell}{M_r}},t(\Tens{u_\ell}{u_r}))\in\Trel{\ONE}{}$.
  Let
  $S=\ABST{x_\ell}{\phi_\ell}{\ABST{x_r}{\phi_r}{\LAPP
      T{\PAIR{x_\ell}{x_r}}}}$
  and $s\in\Pcoh{\Tsem{\LIMPL{\phi_\ell}{(\LIMPL{\phi_r}{\ONE})}}}$ be
  defined by $\Matapp s{u_\ell}\Appsep u_r=t(\Tens{u_\ell}{u_r})$
  (again, $s$ is a morphism in $\PCOH$ by symmetric monoidal closeness
  of that category). It is clear that
  $(S,s)\in\Srel{(\LIMPL{\Trelv{\phi_\ell}{}}
    {(\LIMPL{\Trelv{\phi_r}{}}{\Trel{\ONE}{}}{})}{})}{-\epsilon}$.
  By Lemma~\ref{lemma:rel-app-closeness} we get
  $(\LAPP S{M_\ell},\Matapp
  s{u_\ell})\in\Srel{(\LIMPL{\Trelv{\phi_r}{}}{\Trel{\ONE}{}})}{-\epsilon}$
  and then
  $(\LAPP{S}{M_\ell\,M_r},t(\Tens{u_\ell}{u_r})\in\Trel{\ONE}{}$. Observing
  that there is a probability preserving bijection between the
  reduction paths from $\LAPP{S}{M_\ell\,M_r}$ to $\ONELEM$ and the
  reduction paths from $\LAPP T{\PAIR{M_\ell}{M_r}}$ to $\ONELEM$, we
  conclude that
  $(\LAPP
  T{\PAIR{M_\ell}{M_r}},t(\Tens{u_\ell}{u_r}))\in\Trel{\ONE}{}$
  as contended (in both terms one has to reduce first $M_\ell$ and
  then $M_r$ to a value).
\end{proof}

\begin{lemma}\label{lemma:proj-crel}
  Let $\phi_\ell$ and $\phi_r$ be closed positive types. If
  $(M,u)\in\Srel{\Trel{\TENS{\phi_\ell}{\phi_r}}{}}\epsilon$ then
  $(\PR iM,\Matapp{\Proj i}u)\in\Srel{\Trel{\phi_i}{}}\epsilon$.
\end{lemma}
\begin{proof}
  Let
  $(T,t)\in\Srel{(\Limpl{\Trelv{\phi_i}{}}{\Trel\ONE{}})}{-\epsilon}$,
  we have to prove that
  $(\LAPP T{\PR iM},\Matapp t{(\Matapp{\Proj i}u)})\in\Trel\ONE{}$.
  Let $S=\ABST x{\TENS{\phi_\ell}{\phi_r}}{\LAPP T{\PR ix}}$ and
  $s\in\Pcoh{\Tsem{\LIMPL{\TENS{\phi_\ell}{\phi_r}}\ONE}}$ be defined
  by $\Matapp s{u_0}=\Matapp t{(\Matapp{\Proj i}{u_0})}$ for all
  $u_0\in\Pcoh{\Tsem{\TENS{\phi_\ell}{\phi_r}}}$. Let
  $(W,w)\in\Srel{\Trelv{\TENS{\phi_\ell}{\phi_r}}{}}\epsilon$, which
  means that $W=\PAIR{V_\ell}{V_r}$ and $w=\Tens{v_\ell}{v_r}$ with
  $(V_j,v_j)\in\Srel{\Trelv{\phi_j}{}}\epsilon$ for
  $j\in\{\ell,r\}$. We have $\LAPP SW\Rel\Wred\LAPP T{V_i}$ and
  $\Matapp sw=\Matapp t{v_i}$ and we know that
  $(\LAPP T{V_i},\Matapp t{v_i})\in\Trel\ONE{}$ from which it follows
  by Lemma~\ref{lemma:prob-conv-indep-det} that
  $(\LAPP SW,\Matapp sw)\in\Trel{\ONE}{}$. So we have proven that
  $(S,s)\in\Srel{(\Limpl{\Trelv{\TENS{\phi_\ell}{\phi_r}}{}}
    {\Trel\ONE{}})}{-\epsilon}$
  and hence $(\LAPP SM,\Matapp su)\in\Trel{\ONE}{}$. We have
  $\Matapp su=\Matapp t{(\Matapp{\Proj i}u)}$. Moreover we have a
  probability preserving bijection between the reduction paths from
  $\LAPP T{\PR iM}$ to $\ONELEM$ and the reduction paths from
  $\LAPP SM$ to $\ONELEM$, and hence we have
  $(\LAPP T{\PR iM},\Matapp t{(\Matapp{\Proj i}u)})\in\Trel\ONE{}$ as
  contended.

  Indeed, any reduction path $\pi$ from $\LAPP T{\PR iM}$ to $\ONELEM$ has
  shape $\pi=\lambda\rho$ where $\lambda$ is a reduction path $\LAPP T{\PR
    iM}=\LAPP T{\PR i{M_\ell}}\Rel{\Redone{p_1}}\LAPP T{\PR
    i{M_r}}\Rel{\Redone{p_2}}\cdots\Rel{\Redone{p_k}}\LAPP T{\PR
    i{W}}\Rel{\Redone 1}\LAPP T{V_i}$ (with $W=\PAIR{V_\ell}{V_r}$) and $\rho$ is
  a reduction path from $\LAPP T{V_i}$ to $\ONELEM$. The first steps $\lambda$
  of this reduction are determined by the reduction path
  $M=M_\ell\Rel{\Redone{p_1}}\cdots\Rel{\Redone{p_k}}W$ from $M$ to the value
  $W$. This reduction path determines uniquely the reduction path $\LAPP
  SM=\LAPP S{M_\ell}\Rel{\Redone{p_1}}\cdots\Rel{\Redone{p_k}}\LAPP SW\Rel{\Redone
    1}\LAPP T{\PR iW}\Rel{\Redone 1}\LAPP T{V_i}$ followed by the reduction
  $\rho$ from $\LAPP T{V_i}$ to $\ONELEM$ by $\rho$.
\end{proof}

\begin{lemma}\label{lemma:plus-crel}
  Let $\phi_\ell$ and $\phi_r$ be closed positive types and let
  $(M,u)\in\Srel{\Trel{\phi_i}{}}{\epsilon}$ for $i=\ell$ or $i=r$. Then $(\IN
  i{M},\Matapp{\Inj i}u)\in\Srel{\Trel{\PLUS{\phi_\ell}{\phi_r}}{}}{\epsilon}$.
\end{lemma}
\begin{proof}
  Let
  $(T,t)\in\Srel{(\Limpl{\Trelv{\PLUS{\phi_\ell}{\phi_r}}{}}
    {\Trel{\ONE}{}})}{-\epsilon}$,
  we must prove that
  $(\LAPP T{\IN i{M}},\Matapp t{(\Matapp{\Inj i}u)})\in\Trel{\ONE}{}$.
  Let $S=\ABST x{\phi_i}{\LAPP t{\IN i(x)}}$ and let
  $s\in\Pcoh{\Tsem{\LIMPL{\phi_i}{\ONE}}}$. It is clear that
  $(S,s)\in\Srel{(\Limpl{\Trelv{\phi_i}{}}{\Trel{\ONE}{}}{})}{-\epsilon}$
  and it follows that $(\LAPP S{M},\Matapp su)\in\Trel{\ONE}{}$ which
  implies
  $(\LAPP T{\IN i{M}},\Matapp t{(\Matapp{\Inj i}u)})\in\Trel{\ONE}{}$
  by the usual bijective and probability preserving bijection on
  reductions.
\end{proof}

The next lemma uses notations introduced in Section~\ref{sec:sem-additives}.
\begin{lemma}\label{lemma:case-crel}
  Let $\phi_\ell$ and $\phi_r$ be closed positive type and $\sigma$ be a closed
  type. Let $(M,u)\in\Srel{\Trel{\PLUS{\phi_\ell}{\phi_r}}{}}\epsilon$. For
  $i\in\{\ell,r\}$, let $R_i$ be a term such that
  $\TSEQ{y_i:\phi_i}{R_i}{\sigma}$ and assume that
  $(\ABST{x_i}{\phi_i}{R_i},r_i)
  \in\Srel{\Trel{\LIMPL{\phi_i}{\sigma}}{}}{-\epsilon}$. Then
  $(\CASE{M}{y_\ell}{R_\ell}{y_r}{R_r},
  \Matapp{\Case(r_\ell,r_r)}u)\in\Srel{\Trel{\sigma}{}}{-\epsilon}$.
\end{lemma}
\begin{proof}
  We can write $\sigma=\LIMPL{\psi_1}{\LIMPL\cdots{\LIMPL{\psi_k}{\psi}}}$
  where $\psi$ and the $\psi_j$'s are closed and positive types. Given
  $(W_j,w_j)\in\Srel{\Trelv{\psi_j}{}}{\epsilon}$ for $j=1,\dots,k$, we have to
  prove that
  \begin{align}\label{eq:adeq-case-test1}
    (\LAPP{\CASE{M}{y_\ell}{R_\ell}{y_r}{R_r}}{\Vect W},
    \Matapp{\Case(r_\ell,r_r)}u\Appsep\Vect w)
    \in\Srel{\Trelv{\psi}{}}{-\epsilon}
  \end{align}
  so let $(T,t)\in\Srel{(\Limpl{\Trelv{\psi}{}}{\Trel\ONE{}})}{\epsilon}$, our
  goal is to prove that
  \begin{align}\label{eq:adeq-case-test2}
    (\LAPP T{\LAPP{\CASE{M}{y_\ell}{R_\ell}{y_r} {R_r}}{\Vect W}},
    \Matapp t{(\Matapp{\Case(r_\ell,r_r)}u\Appsep\Vect w)})\in\Trel\ONE{}\,.
  \end{align}
  Let
  $S=\ABST x{\PLUS{\phi_\ell}{\phi_r}}{\LAPP
    T{\LAPP{\CASE{x}{y_\ell}{R_\ell}{y_r} {R_r}}{\Vect W}}}$
  and $s\in\Pcoh{\Tsem{\LIMPL{\PLUS{\phi_\ell}{\phi_r}}{\ONE}}}$ be
  defined by
  $\Matapp s{u_0}=\Matapp
  t{(\Matapp{\Case(r_\ell,r_r)}{u_0}\Appsep\Vect w)}$
  for each $u_0\in\Pcoh{\Tsem{\PLUS{\phi_\ell}{\phi_r}}}$. Then we
  have
  $(S,s)\in\Srel{(\Limpl{\Trelv{\PLUS{\phi_\ell}{\phi_r}}{}}
    {\Trel\ONE{}})}{\epsilon}$.
  Let indeed $i\in\{\ell,r\}$ and let
  $(V,v)\in\Srel{\Trelv{\phi_i}{}}{-\epsilon}$ so that
  $(\IN iV,\Matapp{\Inj
    i}v)\in\Srel{\Trelv{\PLUS{\phi_\ell}{\phi_r}}{}}{-\epsilon}$.
  We have
  $\LAPP S{\IN iV}\Rel{\Transcl{\Wred}}\LAPP{T}
  {\LAPP{\Subst{R_i}{V}{y_i}}{\Vect W}}$
  and
  $\Matapp s{(\Matapp{\Inj i}v)}=\Matapp t{(\Matapp{r_i}v\Appsep\Vect
    w)}$
  and, by our assumptions and Lemma~\ref{lemma:prob-conv-indep-det},
  $(\Subst{R_i}{V}{y_i},\Matapp{r_i}v)\in\Srel{\Trel{\sigma}{}}{-\epsilon}$
  and hence
  $(\LAPP{\Subst{R_i}{V}{y_i}}{\Vect W},\Matapp{r_i}v\Appsep\Vect
  w)\in\Srel{\Trelv{\psi}{}}{-\epsilon}$.
  By Lemma~\ref{lemma:prob-conv-indep-det} it follows that
  $(\LAPP S{\IN iV},\Matapp s{(\Matapp{\Inj i}v)})\in\Trel\ONE{}$ and hence
  $(S,s)\in\Srel{(\Limpl{\Trelv{\PLUS{\phi_\ell}{\phi_r}}{}}
    {\Trel\ONE{}})}{\epsilon}$ as contended.

  Therefore $(\LAPP SM,\Matapp su)\in\Trel\ONE{}$. There is a bijective and
  probability preserving correspondence between the reductions from
  $\LAPP SM$ to $\ONELEM$ and the reductions from
  $\LAPP T{\LAPP{\CASE{M}{x_\ell}{\LAPP{R_\ell}{x_\ell}}{x_r}
      {\LAPP{R_r}{x_r}}}{\Vect W}}$
  to $\ONELEM$: as usual such reductions start with a reduction
  $M=M_1\Rel{\Redone{p_1}}M_2\Rel{\Redone{p_2}}\cdots\Rel{\Redone{p_k}}M_k=\IN
  iV$
  where $i\in\{\ell,r\}$ and $V$ is a value of type $\phi_i$ and
  (after a few $\Wred$-steps) continue with a reduction from
  $\LAPP{T}{\LAPP{R_i}{V\,\Vect W}}$ to
  $\ONELEM$. Therefore~\Eqref{eq:adeq-case-test2} holds and hence we
  have~\Eqref{eq:adeq-case-test1}, this ends the proof of the lemma.
\end{proof}

\begin{lemma}\label{lemma:der-crel}
  Let $\sigma$ be a closed type and let
  $(M,u)\in\Srel{\Trel{\EXCL\sigma}{}}\epsilon$. We have
  $(\GO M,\Matapp{\Der{}}u)\in\Srel{\Trel\sigma{}}\epsilon$.
\end{lemma}
\begin{proof}
  We can write $\sigma=\LIMPL{\psi_1}{\LIMPL\cdots{\LIMPL{\psi_k}{\psi}}}$
  where $\psi$ and the $\psi_j$'s are closed and positive types. Given
  $(W_j,w_j)\in\Srel{\Trelv{\psi_i}{}}{-\epsilon}$, we have to prove that
  \begin{align}\label{eq:adeq-der-test1}
    (\LAPP{\GO M}{\Vect W},
  \Matapp{\Der{}}u\Appsep\Vect w)\in\Srel{\Trelv{\psi}{}}{\epsilon}
  \end{align}
  so let $(T,t)\in\Srel{(\Limpl{\Trelv{\psi}{}}{\Trel\ONE{}})}{-\epsilon}$, our
  goal is to prove that
  \begin{align}\label{eq:adeq-der-test2}
    (\LAPP T{\LAPP{\GO M}{\Vect W}},
    \Matapp t{(\Matapp{\Der{}}u\Appsep\Vect w)})\in\Trel\ONE{}\,.
  \end{align}
  We set $S=\ABST x{\EXCL\sigma}{\LAPP T{\LAPP{\GO x}{\Vect W}}}$ and we define
  $s\in\Pcoh{\Tsem{\LIMPL{\EXCL\sigma}{\ONE}}}$ by
  $\Matapp s{u_0}=\Matapp t{(\Matapp{\Der{}}{u_0}\Appsep\Vect w)}$
  for all $u_0\in\Pcoh{\Tsem{\EXCL\sigma}}$, and we prove that
  $(S,s)\in\Srel{(\Limpl{\Trelv{\EXCL\sigma}{}}{\Trel\ONE{}})}{-\epsilon}$ as in
  the proof of Lemme~\ref{lemma:case-crel} (for instance). We finish the
  proof in the same way, showing~\Eqref{eq:adeq-der-test2} by establishing a
  bijective and probability preserving correspondence between reductions. Our
  contention~\Eqref{eq:adeq-der-test1} follows.
\end{proof}

\begin{lemma}\label{lemma:unfold-crel}
  Let $\phi$ be a closed positive type of shape $\phi=\TREC\zeta\psi$. If
  $(M,u)\in\Srel{\Trel{\phi}{}}\epsilon$ then $(\UNFOLD
  M,u)\in\Srel{\Trel{\Subst\psi\phi\zeta}{}}\epsilon$.
\end{lemma}
\begin{proof}
  Let
  $(T,t)\in\Srel{(\Limpl{\Trelv{\Subst\psi\phi\zeta}{}}{\Trel\ONE{}})}
  {-\epsilon}$,
  we must prove that $(\LAPP T{\UNFOLD M},u)\in\Trel{\ONE}{}$. As
  usual one defines $S=\ABST x{\phi}{\LAPP T{\UNFOLD x}}$ and one
  proves that
  $(S,t)\in\Srel{(\Limpl{\Trelv{\phi}{}}{\Trel\ONE{}})}
  {-\epsilon}$.
  This results from the fact that if
  $(V,v)\in\Srel{\Trelv{\phi}{}}{\epsilon}$ then $V=\FOLD W$ with
  $(W,v)\in\Srel{\Trelv{\Subst\psi\phi\zeta}{}}{\epsilon}$, from
  Lemma~\ref{lemma:prob-conv-indep-det} and from the fact that
  $\UNFOLD{\FOLD W}\Rel\Wred W$ (and of course from our assumption on
  $(T,t)$). It follows that $(\LAPP SM,\Matapp tu)\in\Trel{\ONE}{}$
  from which we deduce $(\LAPP T{\UNFOLD M},u)\in\Trel{\ONE}{}$ by the
  usual reasoning involving a bijective probability preserving
  correspondence on reductions.
\end{proof}

\begin{lemma}\label{lemma:fold-crel}
  Let $\phi$ be a closed positive type of shape $\phi=\TREC\zeta\psi$. If
  $(M,u)\in\Srel{\Trel{\Subst\psi\phi\zeta}{}}\epsilon$ then $(\FOLD
  M,u)\in\Srel{\Trel{\phi}{}}\epsilon$.
\end{lemma}
\begin{proof}
  Let
  $(T,t)\in\Srel{(\Limpl{\Trelv{\phi}{}}{\Trel\ONE{}})} {-\epsilon}$,
  we must prove that $(\LAPP T{\FOLD M},u)\in\Trel{\ONE}{}$. As usual
  one defines $S=\ABST x{\Subst\psi\phi\zeta}{\LAPP T{\FOLD x}}$ and
  one proves that
  $(S,t)\in\Srel{(\Limpl{\Trelv{\Subst\psi\phi\zeta}{}}{\Trel\ONE{}})}
  {-\epsilon}$.
  This results easily from the fact that, if
  $(V,v)\in\Srel{\Trelv{\Subst\psi\phi\zeta}{}}{\epsilon}$ then
  $(\FOLD V,v)\in\Srel{\Trelv{\phi}{}}{\epsilon}$, from
  Lemma~\ref{lemma:prob-conv-indep-det} and from our assumption about
  $(T,t)$. Therefore we have $(\LAPP SM,\Matapp tu)\in\Trel{\ONE}{})$
  from which we deduce $(\LAPP T{\FOLD M},u)\in\Trel{\ONE}{}$ by the
  usual reasoning.
\end{proof}

\begin{lemma}\label{lemma:Trel-sem-closeness}
  Let $\sigma$ be a closed type and let $M$ be a closed term of type
  $\sigma$. Then $(M,0)\in\Srel{\Trel{\sigma}{}}{\epsilon}$ and, if
  $D\subseteq\Pcoh{\Tsem\sigma}$ is directed and satisfies $\forall u\in D\
  (M,u)\in\Srel{\Trel{\sigma}{}}{\epsilon}$ then $(M,\sup
  D)\in\Srel{\Trel{\sigma}{}}{\epsilon}$. Last, if
  $(M,u)\in\Srel{\Trel{\sigma}{}}{\epsilon}$ and $u'\leq u$ then
  $(M,u')\in\Srel{\Trel{\sigma}{}}{\epsilon}$.
\end{lemma}
\begin{proof}
  We can write
  $\sigma=\LIMPL{\phi_1}{\LIMPL\cdots{\LIMPL{\phi_n}\psi}}$ for some $n$ and
  $\List\phi 1n,\psi$ positive and closed. Let us prove the second
  statement. For $i=1,\dots,n$, let
  $(V_i,v_i)\in\Srel{\Trelv{\phi_i}{}}{-\epsilon}$, we must prove that $(\LAPP
  M{V_1\cdots V_n},\Matapp{(\sup
  D)}{v_1}\cdots v_n)\in\Srel{\Crel{\Trelv{\psi}{}}}{\epsilon}$, knowing that
  \[\forall u\in D\quad (\LAPP M{V_1\cdots
    V_n},u\Appsep v_1\cdots v_n)\in\Srel{\Crel{\Trelv{\psi}{}}}{\epsilon}\,.
  \]
  This results from the fact that, given
  $t\in\Pcoh{\Tsem{\LIMPL\psi\ONE}}$, the map
  $u\mapsto \Matapp t{(u\Appsep v_1\cdots v_n)}$
  is Scott continuous from $\Pcoh{\Tsem\phi}$ to $[0,1]$. The first statement
  of the lemma results from the fact that this function maps $0$ to $0$. The
  last one results from the fact that this function is monotone.
\end{proof}

\subsection{Uniqueness of the relation}
With any closed type $\sigma$ we have associated a pair of relations
$\Trel\sigma{}$. We need now to prove that this pair satisfies
$\Srel{\Trel\sigma{}}\POS=\Srel{\Trel\sigma{}}\NEG$ so that we have actually
associated a unique relation with any type.

To this end we prove first that
$\Srel{\Trel\sigma{}}\POS\subseteq\Srel{\Trel\sigma{}}\NEG$. Defining, for any
pair of relations $\cR$, the relation $\Op\cR$ as $(\Srel\cR\POS,\Srel\cR\NEG)$,
this amounts to proving that ${\Trel\sigma{}}\Subrel\Op{\Trel\sigma{}}$. We use
the same notation for the elements of $\Relsv\phi$ for $\phi$ positive.

For the next lemma, we use the same notational conventions as above.
\begin{lemma}
  Let $\Vect\cV$ be a list of pairs of relations such that
  $\cV_i\in\Relsv{\phi_i}$ and $\cV_i\Subrel\Op{\cV_i}$ for each $i$. Then
  $\Trel\sigma{\Vect\zeta}(\Vect\cV)
  \Subrel\Op{\Trel\sigma{\Vect\zeta}(\Vect\cV)}$ and
  $\Trelv\phi{\Vect\zeta}(\Vect\cV)
  \Subrel\Op{\Trelv\phi{\Vect\zeta}(\Vect\cV)}$.
\end{lemma}
\begin{proof}
  The proof is by induction on types. All cases result straightforwardly from
  the monotony of the logical operations on pairs of relations, but the case of
  fixpoints of types. So assume that $\phi=\TREC\zeta\psi$, let
  $\cV=\Trelv\phi{\Vect\zeta}(\Vect\cV)$ and let us prove that $\cV
  \Subrel\Op\cV$. For this, because of the definition of $\cV$ as a
  glb~\Eqref{eq:bi-rel-fixpoint-def}, it suffices to show that
  \begin{align*}
    \FOLD{\Trelv{\psi}{\Vect\zeta,\zeta}(\Vect\cV,\Op\cV)}\Subrel\Op\cV\,.
  \end{align*}
  By the first statement of Lemma~\ref{lemma:bi-rel-monotone} and our
  assumption on the $\cV_i$'s we have
  \[
  \FOLD{\Trelv{\psi}{\Vect\zeta,\zeta}({\Vect\cV},\Op\cV)}
  \Subrel\FOLD{\Trelv{\psi}{\Vect\zeta,\zeta}(\Op{\Vect\cV},\Op\cV)}\,.
  \]
  By inductive hypothesis we have
  $\FOLD{\Trelv{\psi}{\Vect\zeta,\zeta}(\Op{\Vect\cV},\Op\cV)}
  \Subrel\Op{\FOLD{\Trelv{\psi}{\Vect\zeta,\zeta}({\Vect\cV},\cV)}}=\Op\cV$
  since $\cV=\FOLD{\Trelv{\psi}{\Vect\zeta,\zeta}({\Vect\cV},\cV)}$ by
  Lemma~\ref{lemma:bi-rel-monotone}.
\end{proof}

We are left with proving the converse property, namely that
$\Op{\Trel\sigma{}}\Subrel\Trel\sigma{}$ for each closed type $\sigma$. This
requires a bit more work, and is based on a notion of ``finite'' approximation
of elements of the model, that we define by syntactic means as follows. 

\subsubsection{Restriction operators}

We define\footnote{This definition as well as our
  reasoning below features some similarities with \emph{step-indexing} that we
  would like to understand better.}
closed terms $\Restr n\sigma$ and $\Restrv n\phi$ (for $n\in\Nat$, $\sigma$ a
type and $\phi$ a positive type) typed as follows:
$\TSEQ{}{\Restr n\sigma}{\LIMPL{\EXCL\sigma}{\sigma}}$ and
$\TSEQ{}{\Restrv n\phi}{\LIMPL{\phi}{\phi}}$.
\begin{align*}
  \Restr n\phi &= \ABST x{\Excl\phi}{\LAPP{\Restrv n\phi}{\GO x}} \\
  \Restr n{\LIMPL\phi\sigma} &= \ABST f{\EXCL{(\LIMPL\phi\sigma)}}
  {\ABST{x}{\phi}{\LAPP{\Restr{n}{\sigma}}
      {\STOP{(\LAPP{\GO f}{\LAPP{\Restrv n\phi}{x}})}}}}\\
  \Restrv n\ONE &= \ABST x\ONE x\\
  \Restrv{n}{\EXCL\sigma}
  &= \ABST{x}{\EXCL\sigma}{\STOP{(\LAPP{\Restr n\sigma}{x})}}\\
  \Restrv{n}{\TENS{\phi_\ell}{\phi_r}} &= \ABST
  x{\TENS{\phi_\ell}{\phi_r}}{\PAIR{\LAPP{\Restrv n{\phi_\ell}}{\PR
        \ell x}}{\LAPP{\Restrv n{\phi_r}}{\PR r x}}}\\
  \Restrv n{\PLUS{\phi_\ell}{\phi_r}}
  &= \ABST x{\PLUS{\phi_\ell}{\phi_r}}{\CASE{x}{x_\ell}
    {\IN \ell {\LAPP{\Restrv n{\phi_\ell}}{x_\ell}}}{x_r}
    {\IN r {\LAPP{\Restrv n{\phi_r}}{x_r}}}}\\
  \Restrv 0{\TREC\zeta\phi}
  &= \ABST x{\TREC\zeta\phi}{\LOOP{\TREC\zeta\phi}}\\
  \Restrv{n+1}{\TREC\zeta\phi}
  &= \ABST{x}{\TREC\zeta\phi}{\FOLD{\LAPP{\Restrv
        n{\Subst\phi{\TREC\zeta\phi}\zeta}}{\UNFOLD x}}}
\end{align*}
This is a well-founded lexicographic inductive definition on triples
$(n,\sigma,l)$ (where $l\in\{\Vsymb,\Gsymb\}$) if we order the symbols
$\Vsymb$ and $\Gsymb$ by $\Vsymb<\Gsymb$. We consider actually only
triples $(n,\sigma,l)$ such that $\sigma$ is positive if $l=\Vsymb$.

We describe similarly the interpretation of these terms: we give an
explicit description of the matrices $\Psem{\Restr n\sigma}{}$ and
$\Psem{\Restrv n\phi}{}$. To this end, we define a family of sets
$\Srestrp n\sigma l\subseteq\Web{\Tsem\sigma}$ 
by induction on $(n,\sigma,l)$ (where $n\in\Nat$,
$l\in\{\Vsymb,\Gsymb\}$ and $\sigma$ is a closed type which is
positive if $l=\Vsymb$).  

\begin{itemize}
\item $\Srestrv n{\Excl\sigma}=\Mfin{\Srestr n\sigma}$
\item $\Srestrv n{\TENS{\phi_\ell}{\phi_r}}=\Srestrv n{\phi_\ell}\times\Srestrv
  n{\phi_r}$
\item $\Srestrv n{\PLUS{\phi_\ell}{\phi_r}}=
  \{\ell\}\times\Srestrv n{\phi_\ell}\cup\{r\}\times\Srestrv n{\phi_r}$
\item $\Srestrv 0{\TREC\zeta\psi}=\emptyset$
\item
  $\Srestrv{n+1}{\TREC\zeta\psi}=\Srestrv{n}{\Subst\phi{\TREC\zeta\psi}\zeta}$
\item $\Srestr n\phi=\Srestrv n\phi$
\item $\Srestr n{\LIMPL\phi\sigma}=\Srestrv n\phi\times\Srestr n\sigma$.
\end{itemize}

\begin{lemma}\label{lemma:restr-expression}
  Let $n\in\Nat$, $\phi$ be a closed positive type and $\sigma$ be a closed
  type. One has
  \begin{align*}
    \Psem{\Restrv n\phi}{}_{(a,b)}=
    \begin{cases}
      1 & \text{if }a=b\in\Srestrv n\phi\\
      0 & \text{otherwise.}
    \end{cases}
    \quad
    \Psem{\Restr n\sigma}{}_{(c,b)}=
    \begin{cases}
      1 & \text{if }c=\Mset b\text{ and }b\in\Srestr n\sigma\\
      0 & \text{otherwise.}
    \end{cases}
  \end{align*}
\end{lemma}
\begin{proof}
  By Theorem~\ref{th:pos-types-dense}, for a closed positive type $\phi$ and for
  $u\in\PcohEM{\Tsemca\phi}$, it suffices to prove that
  \begin{align*}
    \Psem{\Restrv n\phi}(u)_a=
    \begin{cases}
      u_a & \text{if }a\in\Srestrv n\phi\\
      0 & \text{otherwise}
    \end{cases}
  \end{align*}
  And for a closed type $\sigma$ and for $u\in\Pcoh{\Tsem\sigma}$, it suffices
  to prove
  \begin{align*}
    \Psem{\Restr n\sigma}(\Prom u)_a=
    \begin{cases}
      u_a & \text{if }a\in\Srestr n\sigma\\
      0 & \text{otherwise}
    \end{cases}
  \end{align*}
  Both statements are easily proved by induction.
\end{proof}

We need now to prove that, given $u\in\Pcoh{\Tsem{\sigma}}$, the sequence
$\Psem{\Restr n\sigma}{(\Prom u)}$ is monotone and has $u$ as lub. 

\begin{lemma}\label{lemma:restr-monotone-limit}
  For any triple $(n,\sigma,l)$ where $\sigma$ is positive if
  $l=\Vsymb$, one has
  $\Srestrp n\sigma l\subseteq\Srestrp{n+1}\sigma l$. Moreover
  \begin{align*}
    \bigcup_{n=0}^\infty\Srestrp n\sigma l=\Web{\Tsem\sigma}\,.
  \end{align*}
\end{lemma}
\begin{proof}
  The first statement is straightforward, by induction on $(n,\sigma,l)$.
  For the second statement, we only have to prove the right-to-left
  inclusions. We define the \emph{height} $\Height a$ of an element
  $a$ of $\Web{\Tsem\sigma}$ 
  as follows.
  \begin{itemize}
  \item $\Height\Onelem=1$
  \item $\Height{a_1,a_2}=1+\max{(\Height{a_1},\Height{a_2})}$ (this definition
    is used when $\phi$ is a tensor and when $\sigma$ is a linear implication)
  \item $\Height{i,a}=1+\Height a$
  \item $\Height{\Mset{\List a1k}}=1+\max{(\Height{a_1},\dots,\Height{a_k})}$
  \end{itemize}
  Then by induction on $\Height a$ one proves that
  \begin{align*}
    \forall a\in\Web{\Tsem\sigma}\,\exists n\in\Nat\ a\in\Srestrp n\sigma l
  \end{align*}
    We deal only with the statement relative to $\Srestrv n\phi$. The
  closed positive type $\phi$ is of shape
  \begin{align*}
    \phi=\TREC{\zeta_1}{\cdots\TREC{\zeta_k}{\psi}}
  \end{align*}
  where $\psi$ is not of shape $\TREC\zeta\rho$. We introduce auxiliary closed
  types $\List \phi 1k$ as follows:
  \begin{align*}
    \phi_1&=\phi=\TREC{\zeta_1}{\cdots\TREC{\zeta_k}{\psi}}\\
    \phi_2&=
      \TREC{\zeta_2}{\cdots\TREC{\zeta_k}{\Subst\psi{\phi_1}{\zeta_1}}}\\
    &\ \ \vdots\\
    \phi_{k+1}&=
      \Substbis\psi{\phi_1/\zeta_1,\phi_2/\zeta_2,\dots,\phi_k/\zeta_k}
  \end{align*}
  and all these types have the same interpretation in $\EM\PCOH$. The type
  $\psi$ cannot be one of the type variables $\zeta_i$ as otherwise we would
  have $\Web{\Tsem\phi}=\emptyset$, contradicting our assumption that $a$
  belongs to this set. Assume that $\psi=\EXCL\sigma$ so that we must have
  $a=\Mset{\List b1l}$ with $b_i\in\Web{\Tsem{\sigma'}}$ (where
  $\sigma'=\Substbis\sigma{\phi_1/\zeta_1,\phi_2/\zeta_2,\dots,\phi_k/\zeta_k}$)
  for each $i=1,\dots,l$. We have $\Height{b_i}<\Height a$ for each $i$ so that
  we can apply the inductive hypothesis: for each $i$ there is $n_i$ such that
  $b_i\in\Srestr{n_i}{\sigma'}$. Using the monotonicity property (first
  statement of the lemma) and setting $n=\max(\List n1l)$ we have
  $b_i\in\Srestr n{\sigma'}$ and hence $a\in\Srestrv
  n{\Excl{\sigma'}}$. Therefore
    $a\in\Srestrv {n+k}{\phi}$ (coming
  back to the definition of this set). The other
  cases are dealt with similarly.
\end{proof}

\begin{lemma}\label{lemma:restr-sup-id}
  Let $\sigma$ be a closed type and let $\phi$ be a closed positive
  type. If $u\in\Pcoh{\Tsem\sigma}$ then the sequence
  $(\Matapp{\Psem{\Restr n\sigma}{}}{\Prom u})_{n\in\Nat}$ is monotone
  (in $\Pcoh{\Tsem\sigma}$) and has $u$ as lub. If
  $u\in\Pcoh{\Tsem\phi}$ then the sequence
  $(\Matapp{\Psem{\Restrv n\phi}{}}u)_{n\in\Nat}$ is monotone and has
  $u$ as lub.
\end{lemma}
\begin{proof}
  Immediate consequence of Lemmas~\ref{lemma:restr-expression}
  and~\ref{lemma:restr-monotone-limit}.
\end{proof}

\subsubsection{Main Inclusion Lemma}

Now we are in position of proving the key lemma in the proof of the uniqueness
of relations.

\begin{lemma}\label{lemma:restr-minus-sub-plus}
  Let $\sigma$ be a closed type, $n\in\Nat$ and
  $l\in\{\Vsymb,\Gsymb\}$. If $l=\Gsymb$ and
  $(M,u)\in\Nrel{\Trel\sigma{}}$ then
  $(M,\Matapp{\Psem{\Restr n\sigma}{}}{\Prom
    u})\in\Prel{\Trel\sigma{}}$.
  If $l=\Vsymb$ and $\sigma$ is a closed positive type $\phi$ then
  $(V,v)\in\Nrel{\Trelv{\phi}{}}\Implies(V,\Matapp{\Psem{\Restrv
      n\phi}{}}v)\in\Prel{(\Crel{\Trelv{\phi}{}})}=\Prel{\Trel{\phi}{}}$.
\end{lemma}
\begin{proof}
  By    lexicographic induction on triples $(n,\sigma,l)$ (with $\sigma$
  positive when $l=\Vsymb$). 

  Until further notice, we assume that $l=\Vsymb$.

  The only case where ``$n$ decreases'' in this induction is when
  $\phi=\TREC\zeta\psi$, we start with this case.

  Assume that $\phi=\TREC\zeta\psi$ and that
  $(V,v)\in\Nrel{\Trelv{\phi}{}}$. If $n=0$ we have
  $\Matapp{\Psem{\Restrv n\phi}{}}v=0$ and hence
  $(V,\Matapp{\Psem{\Restrv n\phi}{}}v)\in\Prel{\Trel{\phi}{}}$ by
  Lemma~\ref{lemma:Trel-sem-closeness}. Assume that the implication
  holds for $n$ and let us prove it for $n+1$. Let
  $(V,v)\in\Nrel{\Trelv{\phi}{}}$, that is $V=\FOLD W$ with
  $(W,v)\in\Nrel{\Trelv{\Subst\psi\phi\zeta}{}}$. We have
  $\Matapp{\Psem{\Restrv{n+1}{\TREC\zeta\psi}}{}}v=
  \Matapp{\Psem{\Restrv{n}{\Subst\psi\phi\zeta}}{}}v$
  by definition. By inductive hypothesis we have
  \begin{align}\label{eq:hyprec-unfold}
    (W,\Matapp{\Psem{\Restrv{n}{\Subst\psi\phi\zeta}}{}}v)
    \in\Prel{\Trel{\Subst\psi\phi\zeta}{}}
  \end{align}
  and we must prove that
  $(\FOLD
  W,\Matapp{\Psem{\Restrv{n}{\Subst\psi\phi\zeta}}{}}v)
  \in\Prel{\Trel{\phi}{}}$.
  Let $(T,t)\in\Nrel{(\LIMPL{\Trelv\phi{}}{\Trel\ONE{}})}$, we must
  prove that
  $(\LAPP T{\FOLD
    W},t(\Matapp{\Psem{\Restrv{n}{\Subst\psi\phi\zeta}}{}}v))\in\Trel\ONE{}$.
  Let $S=\ABST x{\Subst\psi\phi\zeta}{\LAPP T{\FOLD x}}$, we have
  $(S,t)\in\Nrel{(\LIMPL{\Trelv{\Subst\psi\phi\zeta}{}}{\Trel\ONE{}})}$
  by Lemma~\ref{lemma:prob-conv-indep-det} and therefore
  \[
  (\LAPP
  SW,t(\Matapp{\Psem{\Restrv{n}{\Subst\psi\phi\zeta}}{}}v))\in\Trel\ONE{}
  \]
  by~\Eqref{eq:hyprec-unfold} and Lemma~\ref{lemma:rel-app-closeness} and this
  implies $(\LAPP T{\FOLD
    W},\Matapp t{(\Matapp{\Psem{\Restrv{n}{\Subst\psi\phi\zeta}}{}}v)})
  \in\Trel\ONE{}$ by
  Lemma~\ref{lemma:prob-conv-indep-det}.

  Assume that $\phi=\EXCL\sigma$ and that
  $(V,v)\in\Nrel{\Trelv{\EXCL\sigma}{}}$, that is $V=\STOP M$ and
  $v=\Prom u$ with $(M,u)\in\Nrel{\Trel\sigma{}}$. By inductive
  hypothesis we have
  $(M,\Matapp{\Psem{\Restr{n}{\sigma}}}{\Prom u})
  \in\Prel{\Trel{\sigma}{}}$
  and hence
  $(\STOP M,\Prom{(\Matapp{\Psem{\Restr{n}{\sigma}}}{\Prom
      u})})\in\Prel{\Trelv{\EXCL\sigma}{}}$
  and since
  $\Prom{(\Matapp{\Psem{\Restr{n}{\sigma}}}{\Prom
      u})}=\Matapp{\Psem{\Restrv{n}{\EXCL\sigma}}}{\Prom u}$ we get
  \[
  (V,\Matapp{\Psem{\Restrv{n}{\EXCL\sigma}}}v)
  \in\Prel{\Trelv{\EXCL\sigma}{}}\subseteq\Prel{\Trel{\EXCL\sigma}{}}
  \]
  as expected.

  Assume that $\phi=\TENS{\phi_\ell}{\phi_r}$ and that
  $(V,v)\in\Nrel{\Trelv{\TENS{\phi_\ell}{\phi_r}}{}}$, that is
  $V=\PAIR{V_\ell}{V_r}$ and $v=\Tens{v_\ell}{v_r}$ with
  $(V_i,v_i)\in\Nrel{\Trelv{\phi_i}{}}$ for $i\in\{\ell,r\}$. By
  inductive hypothesis we have
  $(V_i,\Matapp{\Psem{\Restr{n}{\phi_i}}}{v_i})\in\Prel{\Trel{\phi_i}{}}$. By
  Lemma~\ref{lemma:tens-crel} we get
  $(\PAIR{V_\ell}{V_r},\Tens{\Matapp{\Psem{\Restr{n}{\phi_\ell}}}
    {v_\ell}}{\Matapp{\Psem{\Restr{n}{\phi_r}}}{v_r}})
  \in\Prel{\Trel{\TENS{\phi_\ell}{\phi_r}}{}}$,
  that is
  $(\PAIR{V_\ell}{V_r},\Psem{\Restr{n}{\TENS{\phi_\ell}{\phi_r}}}{}(\Tens{v_\ell}{v_r}))
  \in\Prel{\Trel{\TENS{\phi_\ell}{\phi_r}}{}}$.

  Assume that $\phi=\PLUS{\phi_\ell}{\phi_r}$ and
  $(V,v)\in\Nrel{\Trelv{\PLUS{\phi_\ell}{\phi_r}}{}}$. This means that
  for some $i\in\{\ell,r\}$, one has $V=\IN i{W}$ and
  $v=\Matapp{\Inj i}w$ for $(W,w)\in\Nrel{\Trelv{\phi_i}{}}$. By
  inductive hypothesis we have
  $(W,\Matapp{\Psem{\Restr{n}{\phi_i}}{}}w)\in\Prel{\Trel{\phi_i}{}}$
  and hence
  $(\IN i{W},\Matapp{\Inj
    i}{(\Matapp{\Psem{\Restr{n}{\phi_i}}{}}w)})
  \in\Prel{\Trel{\PLUS{\phi_\ell}{\phi_r}}{}}$
  by Lemma~\ref{lemma:plus-crel}, that is
  $(\IN i{W}, \Matapp{\Psem{\Restr{n}{\PLUS{\phi_\ell}{\phi_r}}}{}}w)
  \in\Prel{\Trel{\PLUS{\phi_\ell}{\phi_r}}{}}$.

  We assume now that  $l=\Gsymb$.

  If $\sigma$ is a closed positive type $\phi$ and let
  $(M,u)\in\Nrel{\Trel{\sigma}{}}$, we must prove that
  \[
  (M,\Matapp{\Psem{\Restr{n}{\sigma}}{}}{\Prom u})\in\Prel{\Trel\sigma{}}
  \]
  which follows directly from the definition of $\Restr{n}{\phi}$ and from the
  inductive hypothesis applied to $(n,\phi,\Vsymb)$.

  Assume last that $\sigma=\LIMPL\phi\tau$ and that
  $(M,u)\in\Nrel{\Trel{\LIMPL\phi\tau}{}}$, we must prove that
  $(M,\Matapp{\Psem{\Restr{n}{\LIMPL\phi\tau}}{}}{\Prom
  u})\in\Prel{\Trel{\IMPL\phi\tau}{}}$.  Let $(V,v)\in\Nrel{\Trelv{\phi}{}}$,
  we must prove that
  \begin{align}\label{eq:goal-app}
    (\LAPP MV,\Matapp{\Psem{\Restr{n}{\LIMPL\phi\tau}}{}}{\Prom u}\Appsep v)
    \in\Prel{\Trel\tau{}}
  \end{align}
  which follows from the fact that
  $\Matapp{\Psem{\Restr{n}{\LIMPL\phi\tau}}{}}{\Prom
    u}\Appsep v=
  \Psem{\Restr{n}{\tau}}{\Prom{(u(\Matapp{\Psem{\Restrv{n}{\phi}}{}}v))}}$.
  Indeed the inductive hypothesis applied to $(n,\phi)$ yields
  $(V,\Matapp{\Psem{\Restrv n\phi}{}}v)\in\Prel{\Trel{\phi}{}}$ and hence
  $(\LAPP MV,\Matapp u{(\Matapp{\Psem{\Restrv n\phi}{}}v)})
  \in\Prel{\Trel{\tau}{}}$ by
  Lemma~\ref{lemma:rel-app-closeness}, from which we derive~\Eqref{eq:goal-app}
  by Lemma~\ref{lemma:restr-sup-id} and Lemma~\ref{lemma:Trel-sem-closeness}.
\end{proof}

\begin{lemma}
  For any closed type $\sigma$ one has
  $\Nrel{\Trel\sigma{}}=\Prel{\Trel\sigma{}}$.
  \end{lemma}
\begin{proof}
  Immediate consequence of lemmas~\ref{lemma:Trel-sem-closeness},
  \ref{lemma:restr-sup-id} and~\ref{lemma:restr-minus-sub-plus}.
\end{proof}
From now on we simply use
the notation $\Trel\sigma{}$ instead of $\Nrel{\Trel\sigma{}}$ and
$\Prel{\Trel\sigma{}}$.

\subsection{Logical relation lemma}

We can prove now the main result of this section.  

\begin{theorem}[Logical Relation Lemma]\label{lemma:adequacy}
  Assume that $\TSEQ{x_1:\phi_1,\dots,x_k:\phi_k}{M}{\sigma}$ and let
  $(V_i,v_i)\in\Trel{\phi_i}{}$ (where $V_i$ is a value and
  $v_i\in\PcohEM{\Tsem{\phi_i}}$) for $i=1,\dots,k$. Then $(\Substbis
  M{V_1/x_1,\dots,V_k/x_k},\Psem{M}^{\List
    x1k}{\Vect v})\in\Trel\sigma{}$ where $\Vect v=(\List v1k)$.
\end{theorem}
\begin{remark}
  One would expect to have rather assumptions of the shape
  ``$(V_i,v_i)\in\Trelv{\phi_i}{}$''; the problem is that we don't know whether
  $\Srel{\Trelv{\phi_i}{}}\POS=\Srel{\Trelv{\phi_i}{}}\NEG$. 
\end{remark}
\begin{proof}
  By induction on the typing derivation of $M$, that is, on $M$. We set
  $\cP=(x_1:\phi_1,\dots,x_k:\phi_k)$ and, given a term $R$, we use $R'$ for
  the term $\Substbis R{V_1/x_1,\dots,V_k/x_k}$. We also use $\Vect v$ for the
  sequence $\List v1k$ and $\Vect x$ for the sequence $\List x1k$.

  The case $M=x_i$ is straightforward.

  Assume that $M=\STOP N$ and that $\phi=\EXCL\sigma$ with
  $\TSEQ\cP N\sigma$. By inductive hypothesis we have
  $(N',\Matapp{\Psem{N}^{\Vect x}}{\Vect v})\in\Trel\sigma{}$.
  Therefore
  $(\STOP{(N')},\Prom{(\Matapp{\Psem{N}^{\Vect x}}{\Vect
      v})})\in\Srel{\Trelv{\EXCL\sigma}{}}{\epsilon}$
  (for $\epsilon=\POS$ or $\epsilon=\NEG$)\footnote{It is not clear
    whether $\Nrel{\Trelv\phi{}}=\Prel{\Trelv\phi{}}$ for any closed
    positive type $\phi$, but we don't need this property in this
    proof, so we leave this technical question unanswered.}. We have
  $\Srel{\Trelv{\EXCL\sigma}{}}{\epsilon}
  \subseteq\Crel{\Srel{\Trelv{\EXCL\sigma}{}}{\epsilon}}
  ={\Srel{\Trel{\EXCL\sigma}{}}{\epsilon}}=\Trel{\EXCL\sigma}{}$
  and hence
  $(M',\Matapp{\Psem{M}^{\Vect x}}{\Vect v}))\in\Trel{\EXCL\sigma}{}$
  as contended.

  Assume that $M=\PAIR{N_\ell}{N_r}$ and
  $\sigma=\TENS{\psi_\ell}{\psi_r}$ with $\TSEQ{\cP}{N_i}{\psi_i}$ for
  $i\in\{\ell,r\}$. By inductive hypothesis we have
  $(N'_i,\Matapp{\Psem{N_i}^{\Vect x}}{\Vect v})\in\Trel{\psi_i}{}$. By
  Lemma~\ref{lemma:tens-crel} we get
  \[(\PAIR{N'_\ell}{N'_r},\Matapp{\Psem{\PAIR{N_\ell}{N_r}}^{\Vect
      x}}{\Vect v})\in\Trel{\TENS{\psi_\ell}{\psi_r}}{}\]
  as contended, since
  $\Matapp{\Psem{\PAIR{N_\ell}{N_r}}^{\Vect x}}{\Vect v}
  =\Tens{\Matapp{\Psem{N_\ell}^{\Vect x}}{\Vect
      v}}{\Matapp{\Psem{N_r}^{\Vect x}}{\Vect v}}$.

  The case $M=\IN iN$ (and $\sigma=\PLUS{\psi_\ell}{\psi_r}$) is handled
  similarly, using Lemma~\ref{lemma:plus-crel}.

  Assume that $M=\FOLD N$ and $\sigma=\phi=\TREC\zeta\psi$ with
  $\TSEQ{\cP}{N}{\Subst\psi\phi\zeta}$. By inductive hypothesis we have
  $(N',\Matapp{\Psem{N}^{\Vect x}}{\Vect v})
  \in\Trel{\Subst\psi\phi\zeta}{}$ which
  implies $(\FOLD{N'},\Matapp{\Psem{N}^{\Vect x}}{\Vect v})\in\Trel{\phi}{}$ by
  Lemma~\ref{lemma:fold-crel}.

  Assume that $M=\ABST x\phi N$ and $\sigma=\LIMPL\phi\tau$, with
  $\TSEQ{\cP,x:\phi}{N}{\tau}$. We must prove that
  $(\ABST{x}{\phi}{N'},\Matapp{\Psem{\ABST x\phi N}^{\Vect x}}{\Vect
    v})\in\Srel{(\Limpl{\Trelv\phi{}}{\Trel\tau{}})}\epsilon$
  for an arbitrary $\epsilon\in\{\NEG,\POS\}$. So let
  $(V,v)\in\Srel{\Trelv{\phi}{}}{-\epsilon}$. Since
  $\Srel{\Trelv{\phi}{}}{-\epsilon}\subseteq\Trel{\phi}{}$, we have
  $(\Subst{N'}{V}{x},\Psem{N}^{\Vect x,x}(\Vect v,v))\in\Trel{\tau}{}$
  by inductive hypothesis. It follows that
  $(\LAPP{\ABST{x}{\phi}{N'}}{V},\Matapp{\Psem{\ABST x\phi N}^{\Vect
      x}}{\Vect v}\Appsep v)\in\Trel\tau{}$
  by Lemma~\ref{lemma:prob-conv-indep-det}, proving our contention.

  Assume that $M=\LAPP RN$ with $\TSEQ{\cP}{R}{\LIMPL\phi\sigma}$ and
  $\TSEQ{\cP}{N}{\phi}$ where $\phi$ is a closed positive type. By
  inductive hypothesis we have
  $(R',\Psem R^{\Vect x}(\Vect v))\in\Trel{\LIMPL\phi\sigma}{}$ and
  $(N',\Matapp{\Psem N^{\Vect x}}{\Vect v})\in\Trel{\phi}{}$ and hence
  $(\LAPP{R'}{N'},\Matapp{\Psem R^{\Vect x}}{\Vect
    v}\Appsep(\Matapp{\Psem N^{\Vect x}}{\Vect v}))\in\Trel{\sigma}{}$
  by Lemma~\ref{lemma:rel-app-closeness}, that is
  $(M',\Matapp{\Psem M^{\Vect x}}{\Vect v})\in\Trel{\sigma}{}$.

  Assume that $M=\FIXT x\sigma N$ with
  $\TSEQ{\cP,x:\EXCL\sigma}{N}\sigma$. The function
  $f:\Pcoh{\Tsem\sigma}\to\Pcoh{\Tsem\sigma}$ defined by
  \begin{align*}
    f(u)=\Psem N^{\Vect x,x}(\Vect v,\Prom u)
  \end{align*}
  is Scott continuous and we have $\Matapp{\Psem M^{\Vect x}}{\Vect
  v}=\sup_{k\in\Nat}f^k(0)$. By induction on $k$, we prove that
  \begin{align}\label{eq:adeq-fixpoint-rec}
    \forall k\in\Nat\quad (M',f^k(0))\in\Trel{\sigma}{}\,. 
  \end{align}
  The base case is proven by Lemma~\ref{lemma:Trel-sem-closeness}. Assume that
  $(M',f^k(0))\in\Trel{\sigma}{}$. Choosing an arbitrary $\epsilon$, we have
  $(\STOP{(M')},\Prom{(f^k(0))})
  \in\Srel{\Trelv{\EXCL\sigma}{}}\epsilon\subseteq\Trel{\EXCL\sigma}{}$ and
  hence by our ``outermost'' inductive hypothesis we have
  $(\Subst{N'}{\STOP{(M')}}{x},f^{k+1}(0))\in\Trel\sigma{}$ from which we get
  $(M',f^{k+1}(0))\in\Trel\sigma{}$ by Lemma~\ref{lemma:prob-conv-indep-det}
  and this ends the proof of~\Eqref{eq:adeq-fixpoint-rec}. We conclude that
  $(M',\Matapp{\Psem M^{\Vect x}}{\Vect v})\in\Trel\sigma{}$ by
  Lemma~\ref{lemma:Trel-sem-closeness}.

  Assume that $M=\GO N$ with $\TSEQ{\cP}{N}{\EXCL\sigma}$. By
  inductive hypothesis we have
  $(N',\Matapp{\Psem{N}^{\Vect x}}{\Vect v})\in\Trel{\EXCL\sigma}{}$
  which implies
  $(\GO{N'},\Der{}(\Matapp{\Psem N^{\Vect x}}{\Vect
    v}))\in\Trel{\sigma}{}$
  by Lemma~\ref{lemma:der-crel}, that is
  $(M',\Matapp{\Psem{M}^{\Vect x}}{\Vect v})\in\Trel{\sigma}{}$.

  Assume that $M=\PR jN$ with $j\in\{\ell,r\}$,
  $\sigma=\Tens{\phi_\ell}{\phi_r}$ and
  $\TSEQ{\cP}{M}{\TENS{\phi_\ell}{\phi_r}}$. By inductive hypothesis
  we have
  $(N'\Matapp{,\Psem{N}^{\Vect x}}{\Vect
    v})\in\Trel{\TENS{\phi_\ell}{\phi_r}}{}$
  and hence
  $(\PR j{N'},\Matapp{\Proj j}{(\Matapp{\Psem{N}^{\Vect x}}{\Vect v})})
  \in\Trel{\phi_j}{}$
  by Lemma~\ref{lemma:proj-crel} that is
  $(M',\Matapp{\Psem{M}^{\Vect x}}{\Vect v})\in\Trel{\phi_j}{}$.

  Assume that $M=\CASE{N}{y_\ell}{R_\ell}{y_r}{R_r}$ with
  $\TSEQ{\cP}{N}{\PLUS{\phi_\ell}{\phi_r}}$ and
  $\TSEQ{\cP,y_j:\phi_j}{R_j}{\sigma}$ for $j\in\{\ell,r\}$. By
  inductive hypothesis we have
  $(N',\Matapp{\Psem N^{\Vect x}}{\Vect
  v})\in\Trel{\PLUS{\phi_\ell}{\phi_r}}{}$
  and
  $(\ABST{y_j}{\phi_j}{R'_j},\Matapp{\Psem{\ABST{y_j}{\phi_j}{R_j}}^{\Vect
    x}}{\Vect v})\in\Trel{\LIMPL{\phi_j}{\sigma}}{}$
  for $j\in\{\ell,r\}$ (to prove this latter fact, one chooses
  $\epsilon\in\{\NEG,\POS\}$ and considers an arbitrary
  $(V,v)\in\Srel{\Trelv{\phi_j}{}}{-\epsilon}$, we have
  $(V,v)\in\Trel{\phi_j}{}$ and hence
  $(\Subst{R'_j}{V}{y_j},\Psem{R_j}^{\Vect x,y_j}(\Vect
  v,v))\in\Trel{\sigma}{}$ by inductive hypothesis, which implies
  \[
  (\LAPP{\ABST{y_j}{\phi_j}{R'_j}}V,
  \Matapp{\Psem{\ABST{y_j}{\phi_j}{R_j}}^{\Vect
    x}}{\Vect v}\Appsep v)\in\Trel{\sigma}{}
  \]
  by Lemma~\ref{lemma:prob-conv-indep-det}). By Lemma~\ref{lemma:case-crel} we
  get
  \[(\CASE{N'}{y_\ell}{R'_\ell}{y_r}{R'_r},
  \Matapp{\Case(\Psem{\ABST{y_\ell}{\phi_\ell}{R_\ell}}^{\Vect
    x}}{\Vect v}),\Matapp{\Psem{\ABST{y_r}{\phi_r}{R_r}}^{\Vect
    x}}{\Vect v}))(\Matapp{\Psem N^{\Vect x}}{\Vect v}))\in\Trel{\sigma}{}
  \]
  that is $(M',\Psem{M}^{\Vect x}(\Vect v))\in\Trel\sigma{}$, by
  Lemma~\ref{lemma:prob-conv-indep-det}.

  Assume that $M=\UNFOLD N$ where $\TSEQ{\cP}{N}{\phi}$ with
  $\phi=\TREC\zeta\psi$. We apply Lemma~\ref{lemma:unfold-crel}
  straightforwardly.

  Assume that $M=\ONELEM$ and the typing derivation consists of the axiom
  $\TSEQ{\cP}{\ONELEM}{\ONE}$ so that $\sigma=\ONE$. We have $(M,\Psem
  M{})\in\Trel{\ONE}{}$ by definition since
  $\Redmats^\infty_{M,\ONELEM}=1=\Psem M{}$.

  Assume last that $M=\COIN p$ for some $p\in[0,1]\cap\Rational$ and the typing
  derivation consists of the axiom $\TSEQ{\cP}{\COIN p}{\PLUS\ONE\ONE}$ so
  that $\sigma=\PLUS\ONE\ONE$. We must prove that $(\COIN p,\Psem{\COIN
    p}{})\in\Trel{\PLUS\ONE\ONE}{}$. Remember that $\Psem{\COIN
    p}{}=p\Base{(\ell,*)}+(1-p)\Base{(r,*)}$. Let $\epsilon\in\{\NEG,\POS\}$ and
  let $(T,t)\in
  \Srel{(\Limpl{\Trelv{\PLUS\ONE\ONE}{}}{{\Trel{\ONE}{}}})}{\epsilon}$, we must
  prove that $(\LAPP T{\COIN
    p},t(p\Base{(\ell,*)}+(1-p)\Base{(r,*)}))\in\Trel{\ONE}{}$. We have
  \begin{align*}
    \Redmats^\infty_{\LAPP T{\COIN p},\ONELEM}
    = p\Redmats^\infty_{\LAPP T{\IN \ell \ONELEM},
      \ONELEM}+(1-p)\Redmats^\infty_{\LAPP T{\IN r \ONELEM},\ONELEM}
  \end{align*}
  since the first reduction step must be $\COIN p\Redone p\IN \ell \ONELEM$ or
  $\COIN p\Redone{1-p}\IN r \ONELEM$. By our assumption on $(T,t)$ we have
  $\Redmats^\infty_{\LAPP T{\IN i\ONELEM}}\geq t(\Base{(i,*)})$ and hence
  $\Redmats^\infty_{\LAPP T{\COIN
      p},\ONELEM}\geq t(p\Base{(\ell,*)}+(1-p)\Base{(r,*)})$ as contended, by
  linearity of $t$.
\end{proof}

\begin{theorem}[Adequacy]\label{th:rel-ad-lemma}
  Let $M$ be a closed term such that $\TSEQ{}{M}{\ONE}$. Then $\Psem
  M{}=\Redmats^\infty_{M,\ONELEM}$. 
\end{theorem}
By Corollary~\ref{th:soundness-ineq} and Theorem~\ref{lemma:adequacy}.

\section{Full Abstraction}\label{sec:FA}

We prove now the Full Abstraction Theorem~\ref{thm:fa}, that is the converse of the Adequacy Theorem.

\subsection{Outline of the proof.} 
We reason by contrapositive. Assume that two closed terms $M_1$ and $M_2$
have different semantics. Remember from Section~\ref{subsec:model-pcoh}, that a
closed term of type $\sigma$ is interpreted as a vector with indices in the web
$\Web{\Psem\sigma}$, so that there is $a\in\Web{\Psem\sigma}$ such that
$\Psem{M_1}_a\neq\Psem{M_2}_a$. We want to design a term that will separate
$M_1$ and $M_2$ observationally.

We define a \emph{testing term}
$\TSEQ{}{\Testt a}{\LIMPL{\EXCL\Tnat}{(\LIMPL{\EXCL\sigma}\ONE)}}$
that will depend only on the structure of the element $a$ of the
web. We then use properties of the semantics (namely that terms of
type $\LIMPL{\EXCL\Tnat}\tau$ can be seen as power series) to find reals $\vec p$ such
that $\LAPP{\Testt a}{\STOP{\Ran{\vec p}}}$ separates $M_1$ and $M_2$:
$$
\Redmats^\infty_{\LAPP{\LAPP{\Testt a}{\STOP{\Ran{\vec p}}}}{\STOP
    M_1},\ONELEM}\neq \Redmats^\infty_{\LAPP{\LAPP{\Testt a}{\STOP{\Ran{\vec
          p}}}}{\STOP M_2},\ONELEM}
$$

\medskip
\subsubsection*{Let us detail the key points of the proof.}

Remember from Section~\ref{subsec:Kleisli_fun} that, because
$\TSEQ{}{\Testt a}{\LIMPL{\EXCL\Tnat}{(\LIMPL{\EXCL\sigma}\ONE)}}$, its
interpretation $\Psem{\Testt a}$ can be seen as a power series
\begin{equation*}
 \Fun{\Psem{\Testt a}}(\vec\zeta)
=\Psem{\Testt a}\Compl {\Prom{\vec\zeta\ }}=\left(\sum_{[\List k1n]\in\Web{\Excl
         \Tnat}}\Psem{\Testt a}_{[\List k1n],b}\prod_{i=1}^n\zeta_{k_i}\right)_{b\in\Web{\LIMPL{\EXCL\sigma}\ONE}}
\end{equation*}
with infinitely many parameters
$\vec\zeta=(\zeta_0,\dots,\zeta_n,\dots)$. Moreover, if
$\sum_{i=0}^\infty \zeta_i\le 1$, then $\vec\zeta\in\Pcoh{\Psem\Tnat}$ and
$\Fun{\Psem{\Testt a}}(\vec\zeta)\in\Pcoh{\Psem{\LIMPL{\EXCL\sigma}\ONE}}$ (see
Theorem~\ref{prop:kleisli-morph-charact}).

The first key point is to remark that the testing term $\Testt a$ is defined in
such a way that ${\Psem{\Testt a}}$ has actually only finitely many parameters
$\List\zeta 0{\Lent a}$ (meaning that if the support of the multiset $c$ is not
included in $\{0,\dots,\Lent a\}$, then $\Psem{\Testt a}_{(c,b)}=0$).  Now, for any $u\in\Pcoh{\Psem\sigma}$,  $\Psem{\ABST y{\EXCL\sigma}{\ABST x{\EXCL\Tnat}{\LAPP{\LAPP{\Testt a}{ x}}y}}} \Prom u\in\Pcoh{\Psem {\LIMPL{\EXCL\Tnat}\ONE}}$. It is also a power series that depends on the same finitely many parameters $\List\zeta 0{\Lent a}$.

The second key point is a separation property of $\Psem{\Testt a }$:
we prove that, in the power series
$\Psem{\ABST y{\EXCL\sigma}{\ABST x{\EXCL\Tnat}{\LAPP{\LAPP{\Testt a}{
          x}}y}}} \Prom u$,
the coefficient of the unitary monomial\footnote{That is, the monomial
  where each exponent is equal to one.} $\prod_{k=0}^{\Lent a}\zeta_k$
is equal to $\coeff - a\,u_a$ with a coefficient $\coeff - a\neq 0$
which depends only on $a$. Now, by assumption, $\Psem{M_1}_a$ and
$\Psem{M_2}_a$ have different coefficient. 
For $i=1,2$, we have
  ${\Psem{\ABST x{\EXCL\Tnat}{\LAPP{\LAPP{\Testt a}{ x}}{\STOP
          M_i}}}}= \Psem{\ABST y{\EXCL\sigma}{\ABST
      x{\EXCL\Tnat}{\LAPP{\LAPP{\Testt a}{
            x}}y}}}{\Prom{\Psem{M_i}}}$. 
Thus, the power series
${\Psem{\ABST x{\EXCL\Tnat}{\LAPP{\LAPP{\Testt a}{ x}}{\STOP M_i}}}}$
(for $i=1,2$) have \emph{different coefficients}.

The last key point uses classical analysis: if two power series with
non-negative real coefficients and finitely many parameters have different
coefficients, then they differ on non-negative arguments $\vec p$ close enough
to zero: $\sum p_i\le 1$, so that $\vec p\in\Pcoh{\Psem\Tnat}$ and
${\Psem{\ABST x{\EXCL\Tnat}{\LAPP{\LAPP{\Testt a}{ x}}{\STOP M_1}}}}\Prom{\vec
p\,}\neq{\Psem{\ABST x{\EXCL\Tnat}{\LAPP{\LAPP{\Testt a}{ x}}{\STOP
        M_2}}}}\Prom{\vec p\,}$.

Finally, in order to substitute in $\Testt a$ the parameters $\vec \zeta$ with
the reals $\vec p$, we use $\Ran{\vec p}$ as introduced in
Paragraph~\ref{subsec:exsyn}. Indeed,
$\Psem{\LAPP{\LAPP{\Testt a}{\STOP{\Ran{\vec p}}}}{\STOP M_i}}=
\Fun{\Psem{\ABST x{\EXCL\Tnat}{\LAPP{\LAPP{\Testt a}{ x}}{\STOP M_i}}}}\Prom{\vec
p\,}$.
We conclude thanks to the Adequacy Theorem~\ref{th:rel-ad-lemma} that ensures
that
$$\Redmats^\infty_{\LAPP{\LAPP{\Testt a}{\STOP{\Ran{\vec p}}}}{\STOP M_1},\ONELEM}\neq \Redmats^\infty_{\LAPP{\LAPP{\Testt a}{\STOP{\Ran{\vec p}}}}{\STOP M_2},\ONELEM}$$

\subsection{Notations.}
In order to define the testing term $\Testt a$, we will reason by
induction and we will need to associate three kinds of \emph{testing
  terms} with the points of the webs. More precisely:
\begin{itemize}
\item Given a positive type $\phi$ and $a\in\Web{\Tsem\phi}$, we define a term
  $\Testv a$ such that 
  \begin{align*}
    &\TSEQ{}{\Testv a}{\LIMPL{\EXCL\Tnat}{\LIMPL\phi\ONE}}.
  \end{align*}
\item Given a general type $\sigma$ and $a\in\Web{\Tsem\sigma}$, we define
  terms $\Testa a$  and
  $\Testt a$ such that 
  \begin{align*}
    &\TSEQ{}{\Testa a}{\LIMPL{\EXCL\Tnat}\sigma}
    &\TSEQ{}{\Testt a}{\LIMPL{\EXCL\Tnat}{\LIMPL{\EXCL\sigma}\ONE}}.
  \end{align*}
  \end{itemize}

We also introduce natural numbers $\Lenv a$, $\Lent a$ and $\Lena a$
depending only on $a$. They represent the finite numbers of parameters
on which the power series $\Fun{\Psem{\Testv a}}$,
$\Fun{\Psem{\Testt a}}$ and $\Fun{\Psem{\Testa a}}$ depend
respectively.  

We denote as $\coeffv a$, $\coefft a$ and $\coeffa a$ natural numbers
depending only on $a$ and that will appear as the coefficient of the
unitary monomial $\prod_{k=0}^{\Lenv a}\zeta_k$,
$\prod_{k=0}^{\Lent a}\zeta_k$ and $\prod_{k=0}^{\Lena a}\zeta_k$
respectively of the corresponding power series. These numbers are all
$>0$.

\medskip

We use the terms introduced in the probabilistic tests paragraph of Subsection~\ref{subsec:exsyn}  and whose semantics are given in Subsection~\ref{subsec:exden}:
\begin{itemize}
\item $\Ran{\Vect p}$ which reduces to $\Num i$ with probability $p_i$ for
  $\sum_{i=0}^{k-1}p_i\le 1$,
\item $M_0\cdot N$ which reduces to $V$ with probability $ p\,q$ if $M_0$
  reduces to $\ONELEM$ with probability $p$ and $N$ reduces to $V$ with
  probability $q$,
\item $M_0\AND\dots\AND M_{k-1}$ which reduces to $\ONELEM$ with probability
  $\prod_{i=0}^{k-1}p_i$ if $M_i$ reduces to $\ONELEM$ with probability $p_i$,
\item $\LAPP{\Pchoose_k^\sigma(\List M0{k-1})}P$ which reduces to $M_i$
  with probability $p_i$ if  $p_i$ is the probability
  of $P$ to reduce to $\NUM i$,
\item $\LAPP{\Pext lr}{\GO Z}$ and $\LAPP{\Pwin k{\vec n}}{\GO Z}$ to
  partition the parameters $\GO Z$. Indeed $Z$ will denote a variable
  of type $\EXCL\Tnat$ and $\GO Z$ has to be considered as the
  sequence of parameters $\Vect\zeta$ of the power series
  interpreting testing terms. $\Lenv a$, $\Lent a$ and $\Lena a$
  represent the number of parameters on which the respective testing
  terms depend. We use $\LAPP{\Pwin i{\vec n}}{\GO Z}$ to extract
  subsequences of $\Vect \zeta$ that will be given as arguments to
  subterms in the inductive definition of the testing terms.  Remember
  that $\LAPP{\Pwin i{n_0,\dots,n_k}}{\GO Z} \Redone {p_l} \Num l$ if
  $l$ is in the $i$th window of size $n_i$, that is
  $n_0+\cdots+n_{i-1}\le l\le n_0+\cdots+n_{i}-1$ and $p_l$ is the
  probability that $\GO Z$ reduces to $\Num l$, that is the
  non-negative real   parameter $\zeta_l$.
  This is a key ingredient in the computation of the coefficient of the
  unitary monomial of the interpretation of testing terms by induction
  on type and on the structure of $a$. 
\end{itemize}

\subsection{Testing terms.}
We define the terms $\Testv a$, $\Testa a$ and $\Testt a$ and
the associated natural numbers $\Lenv a$, $\Lena a$ and $\Lent a$, by induction
on the structure of the point $a$.

\medskip
{Let $\phi$ be a positive type and $a\in\Web{\Tsem\phi}$.  We define
$\Testv a$ and $\Testa a$ by induction on the size of $a$ using the structure
of $\phi$}

\paragraph{Let $\phi=\EXCL\tau$ and $a=\Mset{\List b0{k-1}}$ with
$b_i\in\Web{\Tsem\tau}$.}

 By inductive hypothesis, we have built terms
$\TSEQ{}{\Testt{b_i}}{\LIMPL{\EXCL\Tnat}{\LIMPL{\EXCL\tau}\ONE}}$ and $\TSEQ{}{\Testa{b_i}}{\LIMPL{{\EXCL\Tnat}}\tau}$. Then
we set

\begin{multline*}
  \Testv a= \ABST{Z}{\EXCL\Tnat}{\ABST{x}{\EXCL\tau}{\LAPP{\EAPP{\Testt{b_0}}{{\left(\LAPP{\Pwin 0{\Lent{b_0},\dots,\Lent{b_{k-1}}}}\GO Z\right)}}}x\AND\cdots\\
\AND\LAPP{\EAPP{\Testt{b_{k-1}}}{\left(\LAPP{\Pwin {k-1}{\Lent{b_0},\dots,\Lent{b_{k-1}}}}\GO Z\right)}}x}},\\
 \coeffv a =\prod_{i=0}^{k-1} \coefft{b_i},\ 
\text{ and }  \Lenv a = \Lent{b_0} + \cdots + \Lent{b_{k-1}}\,.
\end{multline*}
\begin{multline*}
  \Testa a = \ABST{Z}{\EXCL\Tnat}{\STOP{\left(\langle\Pchoose^{\tau}_k\left(\EAPP{\Testa{b_0}}{\left(\LAPP{\Pwin{1}{k,\Lena{b_0},\cdots,\Lena{b_{k-1}}}}\GO Z\right)}\right.\right.,\dots,\\\left.\left.\EAPP{\Testa{b_{k-1}}}{\left(\LAPP{\Pwin{k}{k,\Lena{b_0},\cdots,\Lena{b_{k-1}}}}\GO Z\right)}\right)\rangle{\GO Z}\right)}},
\\ \coeffa a =\Factor a \,\prod_{i=0}^{k-1} \coeffa{b_i},\
  \text{ and }
  \Lena a =k+ \Lena{b_0} + \cdots + \Lena{b_{k-1}}\,.
\end{multline*}
Remember that the factorial $\Factor a$ of a multiset $a$ has been defined in
Paragraph~\ref{subsec:Kleisli_fun} as the number of permutations that fix $a$.

\paragraph{If $\phi=\TENS{\phi_\ell}{\phi_r}$ and $a=(b_\ell,b_r)$ with
$b_i\in\Web{\Tsem{\phi_i}}$ for $i\in\{\ell,r\}$,} then we set
\begin{multline*}
  \Testv a = \ABST{Z}{\EXCL\Tnat}{\ABST{x}{\phi
}
    {\LAPP{\EAPP{\Testv{b_\ell}}{\left(\LAPP{\Pwin 0{\Lenv{b_\ell},\Lenv{b_r}}}\GO Z\right)}}{\PR \ell x}\AND\LAPP{\EAPP{\Testv{b_r}}{\left(\LAPP{\Pwin 1{\Lenv{b_\ell},\Lenv{b_r}}}\GO Z\right)}}{\PR r x}}},\\
   \coeffv a = \coeffv{b_\ell}\, \coeffv{b_r},\ 
  \text{ and } \Lenv a=\Lenv{b_\ell}+\Lenv{b_r}\,.
\end{multline*}
\begin{multline*}
  \Testa a = \ABST{Z}{\EXCL\Tnat}\PAIR{\EAPP{\Testa{b_\ell}}{\left(\LAPP{\Pwin 0{\Lena{b_\ell},\Lena{b_r}}}\GO Z\right)}}{\EAPP{\Testa{b_r}}{\left(\LAPP{\Pwin 1{\Lena{b_\ell},\Lena{b_r}}}\GO Z\right)}},\\
  \coeffa a = \coeffa{b_\ell}\,\coeffa{b_r}
  \text{ and } 
  \Lena a=\Lena{b_\ell}+\Lena{b_r}\,.
\end{multline*}

\paragraph{If $\phi=\PLUS{\phi_\ell}{\phi_r}$ and $a=(\ell,a_\ell)$ with $b\in\Web{\Tsem{\phi_\ell}}$
(the case $a=(r,a_r)$ is similar),} then we set
\begin{align*}
  \Testv a = 
    \ABST{Z}{\EXCL\Tnat}{\ABST{x}
  {\PLUS{\phi_\ell}{\phi_r}}{\CASE{x}{y_\ell}
    {\LAPP{\LAPP{\Testv {a_\ell}}Z}{y_\ell}}{y_r}{\LOOP\ONE}}}, \quad \coeffv a = \coeffv{a_\ell}
\quad\text{and }  \Lenv a = \Lenv{a_\ell}\,.
\end{align*}
\begin{align*}
  \Testa a =\ABST{Z}{\EXCL\Tnat}{ \IN \ell {\LAPP{\Testa a_\ell}{Z}}}, \quad \coeffa a = \coeffa{a_\ell}\quad\text{and }
  \Lena a = \Lena{a_\ell}\,. 
\end{align*}

\medskip
Finally, for a general type $\sigma$ and $a\in\Web{\Tsem\sigma}$, we define  $\Testa a$ and $\Testt a$.

\paragraph{If $\sigma=\phi$ is positive,} then we have already defined  $\Testa a$.
\\
Let us now define $\Testt a$. This term does not
depend on the structure of $\phi$:
\begin{align*}
  \Testt a = \ABST Z{\EXCL\Tnat}{\ABST x{\EXCL\phi}{\LAPP{\LAPP{\Testv a}Z}{\GO x}}}
  , \quad \coefft a = \coeffv{a}\quad\text{and }   \Lent a = \Lenv{a} \,.
\end{align*}

\paragraph{If $\sigma=\LIMPL\phi\tau$ and $a=(b,c)$
with $b\in\Web{\Tsem\phi}$ and $c\in\Web{\Tsem\tau}$,} then we set
\begin{multline*}
  \Testa a =\ABST{Z}{\EXCL\Tnat}{ \ABST{x}{\phi}
      {{{\LAPP{\EAPP{\Testv b}{\left(\LAPP{\Pwin 0{\Lenv b,\Lena c}}\GO Z\right)}}x}\,\cdot\,{{\EAPP{\Testa c}{\left({\LAPP{\Pwin 1{\Lenv b,\Lena c}}\GO Z}\right)}}}}}}\,,\\
    \coeffa a = \coeffv{b}\,\coeffa{c}\,,
    \text{ and }\Lena a=\Lenv b+\Lena c\,.
\end{multline*}
\begin{multline*}
  \Testt a = \ABST Z{\EXCL\Tnat}{\ABST{f}{\EXCL{(\LIMPL\phi\tau)}}{
      {\LAPP{\EAPP{\Testt c}{\left(\LAPP{\Pwin 1{\Lena b,\Lent c}}{\GO Z}\right)}}{\STOP{\left(\LAPP{\GO f}{\EAPP{\Testa b}{\left(\LAPP{\Pwin 0{\Lena b,\Lent c}}\GO Z\right)}}\right)}}}}}\\
  \coefft a = \coeffa{b}\,\coefft{c},\,
  \text{ and }\Lent a=\Lena b+\Lent c\,.
\end{multline*}

It is easy to check that these terms satisfy the announced typing judgments. It
is also clear that $\coeffv a$, $\coeffa a$ and $\coefft a$ are non zero natural numbers.

\medskip We will now tackle the proof of the main observation: that is
that 
the semantics of
$\Testt a$ is a power series with finitely many parameters
and whose coefficient of the unitary monomial can be seen as a
morphism in $\Pcoh{\Tsem{\LIMPL{\EXCL\sigma}\ONE}}$.

  Lemma~\ref{lem:coeff-morph}
  introduces notations for the unitary monomials and provides
  useful properties for proving the key
  Lemma~\ref{lem:point-test} which gives the coefficients of these
  monomials.
\begin{lemma}\label{lem:coeff-morph}
  Let $\sigma$ be a general type and 
$t\in\Pcoh{\Tsem{\LIMPL{\EXCL\Tnat}{\sigma}}}$. 
\begin{enumerate}
\item 
Assume that there is $k\in\Nat$ such that for any $c\in\Web{\Tsem{\sigma}}$, the power series $\Fun t_c$
over $\Pcoh{\Tsem\Tnat}$ depends on  the $k$ first parameters.
For any $c\in\Web{\Tsem\sigma}$, let us denote as $\mon{\vec\zeta}{\Fun t}_c$  the coefficient of the monomial $\Listbis\zeta 0{k-1}$ of $\Fun t_c$. Then,  $k^{-k}\,\mon{\vec\zeta}{\Fun t}\in\Pcoh{\Tsem\sigma}$.  
\label{item1:coeff-morph}
\item Assume moreover that $\sigma=\LIMPL\phi\tau$ where $\phi$ is
a positive type and $\tau$
a general type. 
Let $m\in\Pcoh{\Tsem\tau}$
and $a\in\Web{\Tsem{\phi}}$.
If $\forall
u\in \PcohEM{\Tsemca\phi}\ \mon{\vec\zeta}{\Fun t}u
=mu_a$ then $\forall u\in \Pcoh{\Tsem\phi}\ \mon{\vec\zeta}{\Fun
  t}u =mu_a$.
\label{item2:coeff-morph}
\end{enumerate}
\end{lemma}
\begin{proof}
  We prove~\Eqref{item1:coeff-morph}.

  First notice that $\forall c\in\Web{\Tsem\sigma}$,
  the coefficient of the monomial $\prod_{i=0}^{k-1}\zeta_i$ is $\mon{\vec\zeta}{\Fun
    t}_c=t_{([0,\dots,k-1],c)}$. Now, let $\Vect{\tfrac1k}$ be the sequence of $k$ coefficients all equal to $\frac1k$:
  \begin{equation*}
    \Fun{t}_c(\vec{\tfrac
    1k})=\displaystyle\sum_{\substack{\mu\in\Web{\Tsem{\EXCL\Tnat}}\\\Supp\mu\subseteq\{0,\dots,k-1\}}}t_{(\mu,c)}{\tfrac 1k}^{\Card\mu}
\end{equation*}
so that $\mon{\vec\zeta}{\Fun t}_c\, k^{-k}\le
  \Fun{t}(\vec{\tfrac 1k})_c$. Since $\Vect{\tfrac
    1k}\in\Pcoh{\Tsem\Tnat}$, $\Fun{t}(\vec{\tfrac
    1k})\in\Pcoh{\Tsem\sigma}$ which is downward closed, we have that
  $k^{-k}\,\mon{\vec\zeta}{\Fun t}\in\Pcoh{\Tsem\sigma}$.

  Now we prove~\Eqref{item2:coeff-morph}.

  For any $a\in\Web{\Tsem{\phi}}$, there is $A$ such that $u_a\le A$
  for any $u\in\Pcoh{\Tsem\phi}$ (see
  Subsection~\ref{subsec:model-pcoh}). Hence, for any
  $u\in\Pcoh{\Tsem\phi}$, $\tfrac{u_a}A m\in\Pcoh{\Tsem\tau}$ and we
  deduce thanks to Lemma~\ref{lemma:PCS-moprh-charact} that
  $u\mapsto u_a\,m\,A^{-1}$ is in
  $\Pcoh{\Tsem{\LIMPL\phi\tau}}$. Without loss of generality, we can
  choose $A\ge k^{k}$ so that $u\mapsto u_a\,m\,A^{-1}$ and
  $u\mapsto A^{-1}\mon{\vec\zeta}{\Fun t}u$ are both in
  $\Pcoh{\Tsem{\LIMPL\phi\tau}}$. Now, since $\Tsemca\phi$
  is dense (see
  Theorem~\ref{th:pos-types-dense}), if
  $u\mapsto A^{-1}\,\mon{\vec\zeta}{\Fun t}u$ and
  $u\mapsto A^{-1}\,u_a\,m$ are equal on all coalgebraic points
  $u\in\PcohEM{\Tsemca\phi}$, they are equal (see
  Definition~\ref{def:dense}). Thus, for all $u\in\Pcoh{\Tsem\phi}$,
  $u_a\, m=\mon{\vec\zeta}{\Fun t}u$.
\end{proof}

We are now ready to prove that the coefficient of the
unitary monomial of a testing term associated with a point $a$ of the
web allows to extract the $a$-coefficient of an argument, up to a
non-zero coefficient depending only on $a$.

This is   central in the proof of Full
Abstraction.
Let us first introduce some notations that will be used along this proof.

Intuitively, for $s\in \Pcoh{\Tsem{\LIMPL{\EXCL\Tnat}\sigma}}$, we
reason on power series $\Fun s$ with values in
$\Pcoh{\Tsem{\sigma}}$. But formally, we reason on the non-negative
real power series $\Fun s_{a'}$ defined for each
$a'\in\Web{\Tsem\sigma}$ and for all
parameters\footnote{\label{footnote:parameters}We follow the common
  mathematical practice of using the same notation
  $\Vect\zeta=(\List\zeta0n)$ to refer to the formal parameters of a
  power series and to real arguments of the corresponding
  function. 
}
$\vec\zeta\in\Pcoh{\Tsem\iota}$ as $\Fun s_{a'}(\vec\zeta)=(\Fun
s(\vec\zeta))_{a'}=(s\Prom{\left.\vec \zeta\,\right.})_{a'}$ (see
Paragraph~\ref{subsec:Kleisli_fun}).

We want to compute the unitary monomial of $\Fun s$ which will be in $\Pcoh{\Tsem\sigma}$. We define it for each $a'\in\Web{\Tsem\sigma}$ as $\mon{\vec\zeta}{s}_{a'}=\mon{\vec\zeta}{\Fun s_{a'}}$.

We will also use the fact that a morphism
$t\in\Pcoh{\Tsem{\LIMPL\phi\ONE}}$ is defined by the collection of
$t_{(a',\ast)}$
for $a'\in\Web{\Tsem\phi}$ and is extensionally characterized
by its values $t\,u$ on every $u\in \Pcoh{\Tsem\phi}$.  Indeed, for
any $a'\in\Web{\Tsem\phi}$, there is $\epsilon>0$ such that
$\epsilon \Base{a'}\in\Pcoh{\Tsem\phi}$ and
$t_{(a',\ast)}=\tfrac1\epsilon\,(t\,\epsilon\Base{a'})_{\ast}$ by
linearity of matrix multiplication.

\begin{lemma}\label{lem:point-test}
  Let $\sigma$ be a type and $a\in\Web{\Tsem\sigma}$.

\begin{enumerate}
\item\label{point-test0}
Assume that $\sigma=\phi$ is positive. If $a'\in\Web{\Tsem\phi}$, then
$\Fun{\Psem{\Testv a}}_{(a',\ast)}$ is a power series over
$\Pcoh{\Tsem\Tnat}$ depending on $\Lenv a$ parameters, so we define $\mon{\vec\zeta}{\Psem{\Testv a}}_{(a',\ast)}=\mon{\vec\zeta}{\Fun{\Psem{\Testv a}}_{(a',\ast)}}$ and check that
$\mon{\vec\zeta}{\Psem{\Testv a}}\in\Pcoh{\Tsem{\LIMPL\phi\ONE}}$  and that
for any $u\in\Pcoh{(\Tsem\phi)}$,
$\mon{\vec\zeta}{\Psem{\Testv a}}u=\coeffv
a\,u_a$.
\label{point-testa0}
\item
Assume that $\sigma$ is a general type. For any $a'\in\Web{\Tsem\sigma}$, $\Fun{\Psem{\Testa a}}_{a'}$ 
is a power series over $\Pcoh{\Tsem\Tnat}$ depending on $\Lena a$ parameters, so we define $\mon{\vec\zeta}{\Psem{\Testa a}}_{a'}=\mon{\vec\zeta}{\Fun{\Psem{\Testa a}}_{a'}}$ and check that $\mon{\vec\zeta}{\Psem{\Testa a}}\in\Pcoh{\Tsem\sigma}$ and that 
$\mon{\vec\zeta}{\Psem{\Testa a}}=\coeffa a\,\Base a$ where $\Base a$ is the base vector such that $(\Base a)_{a'}=\Kronecker{a'}a$ for $a'\in\Web{\Tsem\sigma}$.
\label{point-testa+}
\item
  Let $\sigma$ be a general type.
  For any $a'\in\Web{\Tsem{\EXCL\sigma}}$, $\Fun{\Psem{\Testt a}}_{(a',\ast)}$ is a power series over $\Pcoh{\Tsem\Tnat}$ depending on $\Lent a$ parameters, so we define $\mon{\vec\zeta}{\Psem{\Testt a}}_{(a',\ast)}=\mon{\vec\zeta}{\Fun{\Psem{\Testt a}}_{(a',\ast)}}$ and check that
  $\mon{\vec\zeta}{\Psem{\Testt a}}\in\Pcoh{\Tsem{\LIMPL{\EXCL\sigma}\ONE}}$
  and that for any $u\in \Pcoh{\Tsem\sigma}$,
  $\mon{\vec\zeta}{\Psem{\Testt a}}\Prom u=\coefft
  a\,u_a$. 
\label{point-testa-}
\end{enumerate}
\end{lemma}
\begin{proof}
  Let us argue by mutual induction on the size of $a$ and the structure of  $\phi$.
 
\medskip
Let $\phi$ be a positive type and $a\in\Web{\Tsem\phi}$.  We prove~\Eqref{point-testa+} and~\Eqref{point-testa0}  by induction on  the structure
of $\phi$

\paragraph{Assume that $\phi=\EXCL\tau$ and that  $a=\Mset{\List b0{k-1}}$} with
$b_i\in\Web{\Tsem\tau}$.

\smallskip We prove~\Eqref{point-testa0}. Let
$a'=\Mset{\List{b'}0{k'-1}}\in\Tsem{\Excl\tau}$ with
$b'_j\in\Web{\Tsem\tau}$ and $\vec\zeta\in\Pcoh{\Tsem\Tnat}$ be the
concatenation of the finite sequences\footnote{We assume
  that the support of indices of the sequences are disjoint even if
  this requires some renaming.} $\vec\zeta^i\in\Pcoh{\Tsem\Tnat}$ such
that the length of $\vec\zeta^i$ is $\Lent{b_i}$.

By Theorem~\ref{th:pos-types-dense},
$\Fun{\Psem{\Testv a}}_{(a',\ast)}(\vec\zeta)=(\Psem{\Testv
  a}\,\Prom{\left.\vec\zeta\right.})_{(a',\ast)}\in\Pcoh{\Tsem{\LIMPL{\EXCL\tau}\ONE}}$
is completely determined by the function $u\mapsto \Psem{\Testv a}\,\Prom{\left.\vec\zeta\right.}\,\Prom u$ defined on $\Pcoh{\Psem{\tau}}$.
 By inductive hypothesis,
$\Fun{\Psem{\Testv a}}_{(a',\ast)}$ depends on finitely many
parameters $\Lenv a=\Lent{b_0}+\dots+\Lent{b_{k-1}}$, since
\begin{align*}
  \Psem{\Testv a}\Prom{\left.\vec\zeta\right.}\Prom u
  = \prod_{i=0}^{k-1}\Psem{\Testt{b_i}}\Prom{\left.\vec\zeta^i\right.}\Prom u
  \qquad\text{and therefore}\qquad 
  \mon{\vec\zeta}{  \Psem{\Testv a}}\Prom u
  =
  \prod_{i=0}^{k-1}\mon{\vec\zeta^i}{\Psem{\Testt{b_i}}}\Prom{u}.
\end{align*}
Again, by inductive hypothesis it follows that
\[\mon{\vec\zeta}{\Psem{\Testv{a}}}\Prom u=\prod_{i=0}^{k-1}\coefft{b_i}\,u_{b_i}=\coeffv a\,(\Prom u)_a. \]
We can apply~\Eqref{item2:coeff-morph} of Lemma~\ref{lem:coeff-morph}, so that we have $\mon{\vec\zeta}{\Psem{\Testv{a}}}u=\coeffv a\,u_a$  for all $u\in\Pcoh{\Tsem{\EXCL\tau}}$.

\smallskip
We prove~\Eqref{point-testa+}.
Let
$a'=\Mset{\List{b'}0{k'-1}}\in\Tsem{\Excl\tau}$ and
$\vec\zeta\in\Pcoh{\Tsem\Tnat}$ be the concatenation of the finite sequences
$\vec\zeta^\ast,\vec\zeta^0,\dots,\vec\zeta^{k-1}\in\Pcoh{\Tsem\Tnat}$ such that the length of $\Vect\zeta^{\ast}$ is $k$
and the length of $\vec\zeta^{i}$ is $\Lena{b_i}$ for $i\ge 0$.
By inductive hypothesis, $\Fun{\Psem{\Testa a}}_{a'}$ depends on finitely many parameters $\Lenv a=k+\Lena{b_0}+\dots+\Lena{b_{k-1}}$, since

\begin{align}
  \label{eq:aplus-expression}
  \Fun{\Psem{\Testa a}}_{a'}(\vec \zeta)=(\Psem{\Testa a}\Prom{\left.\vec\zeta\right.})_{a'} &= \Promp{\sum_{i=0}^{k-1} {\zeta^\ast_i}\ \Psem{\Testa{b_i}}\,\Prom{\left.\vec\zeta^{i}\right.}}_{a'}=\prod_{j=0}^{k'-1}\left(\sum_{i=0}^{k-1} {\zeta^\ast_i}\ \Psem{\Testa{b_i}}\,\Prom{\left.\vec\zeta^{i}\right.}\right)_{b'_j}\,.
\end{align}
We want to compute the coefficient of the unitary monomial, which
contains exactly one copy of each parameter of each $\Vect\zeta^i$. 
If $k'\neq k$ then expression~\Eqref{eq:aplus-expression} contains no monomial where each parameter of $\Vect\zeta^\ast$ appears exactly once,
 so that
$\mon{\vec\zeta}{\Psem{\Testa a}}_{a'}=0$ in that case. If $k'=k$ and
$\mathfrak{S}_k$ is the set of permutations over $k$, then by using
the fact that factorial
$\Factor a=\Card\{\rho\in\mathfrak S_k \St \forall i\,
b_i=b_{\rho(i)}\}$,
by denoting the Kronecker symbol as $\Kronecker{a}{a'}$ and by the
inductive hypothesis, we get:
\begin{align*}
  \mon{\vec\zeta}{\Psem{\Testa a}}_{a'}=\sum_{\rho\in\mathfrak{S}_k}\prod_{i=0}^{k-1}\mon{\vec\zeta^i}{\Psem{\Testa{b_i}}}_{b'_{\rho(i)}}=\Factor a\  \prod_{i=0}^{k-1} \coeffa{b_i}\Kronecker{b_i}{b'_i}=\coeffa a\, \Kronecker{a}{a'}=\coeffa a\,   (\Base a)_{a'}\,.
\end{align*}

\paragraph{
Assume that $\phi=\TENS{\phi_\ell}{\phi_r}$ and  that $a=(b_\ell,b_r)$ with
$a_i\in\Web{\Tsem{\phi_i}}$.} Let $\vec\zeta\in\Pcoh{\Tsem\Tnat}$ be the concatenation of the finite sequences $\vec\zeta^\ell,\vec\zeta^r\in\Pcoh{\Tsem\Tnat}$ such that the length of $\vec\zeta^i$ is $\Lena{b_i}$. 

\smallskip
We prove~\Eqref{point-testa0}. 
Let $a'=(b'_\ell,b'_r)\in\Web{\Tsem{\TENS{\phi_\ell}{\phi_r}}}$.
By Theorem~\ref{th:pos-types-dense},
$\Fun{\Psem{\Testv a}}_{(a',\ast)}(\vec\zeta)=(\Psem{\Testv a}\,\Prom{\left.\vec\zeta\right.})_{(a',\ast)}\in\Pcoh{\Tsem{\LIMPL{\TENS{\phi_\ell}{\phi_r}}\ONE}}$ 
is completely determined by the function $u\mapsto \Psem{\Testv a}\,\Prom{\left.\vec\zeta\right.}\,\Prom u$ defined on $\Pcoh{\Psem{\Tens{\phi_\ell}{\phi_r}}}$.
Besides, 
if $u\in\PcohEM{(\Tsemca{\TENS{\phi_\ell}{\phi_r}})}$, then $u=\Tens{u_\ell}{u_r}$ where $u_i=\Projt
i(u)\in\PcohEM{(\Tsemca{\phi_i})}$ for $i\in\{\ell,r\}$ (see Lemma~\ref{lemma:sem-values}). Therefore,
 by inductive hypothesis, $\Fun{\Psem{\Testv a}}_{(a',\ast)}$ depends on finitely many parameters $\Lenv a=\Lenv{b_\ell}+\Lenv{b_r}$, since
\begin{align*}
  \Psem{\Testv a}\Prom{\left.\vec\zeta\right.}u
  &= 
  \Psem{\Testv{b_\ell}}\Prom{\left.\vec\zeta^\ell\right.}u_\ell\,\Psem{\Testv{b_r}}\Prom{\left.\vec\zeta^r\right.}u_r
\quad  \text{ and  therefore }\quad
  \mon{\vec\zeta}{  \Psem{\Testv a}}u=
  \mon{\vec\zeta^\ell}{\Psem{\Testv{b_\ell}}}u_\ell\,\mon{\vec\zeta^r}{\Psem{\Testv{b_r}}}u_r.
\end{align*}
Hence, by inductive hypothesis $\mon{\vec\zeta}{  \Psem{\Testv a}}u=\coeffv{b_\ell}(u_\ell)_{b_\ell}\,\coeffv{b_r}(u_r)_{b_r}=\coeffv{a}u_{a}$ for $u\in\PcohEM{\Tsemca\phi}$. We conclude by Lemma~\ref{lem:coeff-morph}, that this holds also for $u\in\Pcoh{\Tsem\phi}$.

\smallskip
We prove~\Eqref{point-testa+}. 
Let $a'=(b'_\ell,b'_r)\in\Web{\Tsem{\TENS{\phi_\ell}{\phi_r}}}$. By inductive hypothesis, $\Fun{\Psem{\Testa a}}_{a'}$ depends on finitely many parameters $\Lena a=\Lena{b_\ell}+\Lena{b_r}$, since
\begin{align*}
 \Fun{\Psem{\Testa a}}_{a'}(\vec\zeta)= (\Psem{\Testa a}\Prom{\vec\zeta\,})_{a'} = \Tens{(\Psem{\Testa{b_\ell}}\Prom{\left.\vec\zeta^\ell\right.})_{b'_\ell}}{(\Psem{\Testa{b_r}}\Prom{\left.\vec\zeta^r\right.})_{b'_r}}
\end{align*}
We deduce using inductive hypothesis that 
\begin{align*}
  \mon{\vec\zeta}{\Psem{\Testa a}}_{a'}
  = \mon{\vec\zeta^\ell}{\Psem{\Testa{b_\ell}}}_{b'_\ell}\,\mon{\vec\zeta^r}{\Psem{\Testa{b_r}}}_{b'_r}=\coeffa{b_\ell}\,\Kronecker{b_\ell}{b'_\ell}\ \coeffa{b_r}\,\Kronecker{b_r}{b'_r}=\coeffa{a}\,\Kronecker{a}{a'}.
\end{align*}

\paragraph{
Assume that $\phi=\PLUS{\phi_\ell}{\phi_r}$ and  that $a=(\ell,a_\ell)$ with
$a_\ell\in\Web{\Tsem{\phi_\ell}}$}(the case $a=(r,a_r)$ is similar).

\smallskip
We prove~\Eqref{point-testa0}.  Let $a'=(i,a'_i)\in \Web{\Tsem{\PLUS{\phi_\ell}{\phi_r}}}$.
By Theorem~\ref{th:pos-types-dense},
$\Fun{\Psem{\Testv a}}_{(a',\ast)}(\vec\zeta)=(\Psem{\Testv a}\,\Prom{\left.\vec\zeta\right.})_{(a',\ast)}\in\Pcoh{\Tsem{\LIMPL{\PLUS{\phi_\ell}{\phi_r}}\ONE}}$ is completely determined by the function $u\mapsto \Psem{\Testv a}\,\Prom{\left.\vec\zeta\right.}\,u$ defined on $\PcohEM{\Tsemca{\PLUS{\phi_\ell}{\phi_r}}}$. 
Besides, 
if $u\in\PcohEM{(\Tsemca{\PLUS{\phi_\ell}{\phi_r}})}$, then there is $i\in\{\ell,r\}$ such that
$u=\Inj i {u_i}$ with $u_i\in\PcohEM{(\Tsemca{\phi_i})}$ (see Lemma~\ref{lemma:sem-values}).  Therefore,
 by inductive hypothesis, $\Fun{\Psem{\Testv a}}_{(a',\ast)}$ depends on finitely many parameters $\Lenv a=\Lenv{a_\ell}$, since
if $i=\ell$, then
\begin{align*}
  \Psem{\Testv a}\Prom{\left.\vec\zeta\right.}u=\Psem{\ABST{x}
  {\PLUS{\phi_\ell}{\phi_r}}{\CASE{x}{y_\ell}
    {\LAPP{\Testv{a_\ell}}{y_\ell}}{y_r}{\LOOP\ONE}}}\Prom{\left.\vec\zeta\right.}u=\Psem{\Testv{a_\ell}}\Prom{\left.\vec\zeta\right.}u_\ell
\end{align*}
and if $i=r$ then $\Psem{\Testv a}\Prom{\left.\vec\zeta\right.}u=\Psem{\Loopt\One}=0$. 
So we can compute that  $\mon{\vec\zeta}{\Psem{\Testv
  a}}u=\coeffv{a_\ell}(u_\ell)_{a_\ell}=\coeffv a\, u_a$ for $u\in\PcohEM{\Tsemca\phi}$ and this still holds for $u\in\Pcoh{\Tsem\phi}$ by Lemma~\ref{lem:coeff-morph}.

\smallskip
We prove~\Eqref{point-testa+}.  Let $a'=(i,a'_i)\in \Web{\Tsem{\PLUS{\phi_\ell}{\phi_r}}}$. By inductive hypothesis, $\Fun{\Psem{\Testa a}}_{a'}$ depends on finitely many parameters $\Lena a=\Lena{a_\ell}$, since
\begin{align*}
  \Fun{\Psem{\Testa a}}_{a'}(\vec\zeta)=
  (\Psem{\Testa a}\Prom{\left.\vec\zeta\right.})_{a'}=(\Psem{\IN{\ell}{\Testa{a_\ell}}}\Prom{\left.\vec\zeta\right.})_{(i,a'_i)}=\Injtr{\ell}(\Psem{\Testa{a_\ell}}\Prom{\left.\vec\zeta\right.})_{i,a'_i}=\Kronecker\ell i\, (\Psem{\Testa{a_\ell}}\Prom{\left.\vec\zeta\right.})_{a'_i}. 
\end{align*}
We can therefore compute $
  \mon{\vec\zeta}{\Psem{\Testa a}}_{(i,a'_i)}
  =\Kronecker\ell i\,\mon{\vec\zeta}{\Psem{{\Testa{a_\ell}}}}_{a'_i}
  =\coeffa a\,\Kronecker{a}{(i,a'_i)}
$ by inductive hypothesis.

\medskip
Finally, for a general type $\sigma$ and $a\in\Web{\Tsem\sigma}$, we prove~\Eqref{point-testa+} and~\Eqref{point-testa-}.

\paragraph{If $\sigma=\phi$ is positive,} we have already proved~\Eqref{point-testa+}. Let us prove~\Eqref{point-testa-}. Let $a\in\Web{\Tsem{\phi}}$. 
Let $a'\in\Web{\Tsem{\EXCL\phi}}$. 
By Theorem~\ref{th:pos-types-dense} and~Lemma~\ref{lemma:sem-values},
$\Fun{\Psem{\Testt a}}_{(a',\ast)}(\vec\zeta)=(\Psem{\Testt a}\,\Prom{\left.\vec\zeta\right.})_{(a',\ast)}\in\Pcoh{\Tsem{\LIMPL{\EXCL\tau}\ONE}}$ is completely determined by $u\mapsto\Psem{\Testt a}\,\Prom{\left.\vec\zeta\right.}\,\Prom u$ defined on $\Pcoh{\Tsem{\tau}}$. Therefore,
 by inductive hypothesis, $\Fun{\Psem{\Testt a}}_{(a',\ast)}$ depends on finitely many parameters $\Lent a=\Lenv{a}$, since
\begin{align*}
  \Psem{\Testt a}\Prom{\left.\vec\zeta\right.}\Prom u
  =\Psem{\Testv a}\Prom{\left.\vec\zeta\right.}u
  \qquad\text{and therefore}\qquad \mon{\vec\zeta}{\Psem{\Testt a}}\Prom u=\mon{\vec\zeta}{\Psem{\Testv a}}u
\end{align*}
By inductive hypothesis,
$\mon{\vec\zeta}{\Psem{\Testv a}}u= \coeffv a\, u_a=\coefft a\,u_a$.

\paragraph{
Last, let $\sigma=\LIMPL\phi\tau$} Let $a=(b,c)\in\Web{\Tsem\sigma}$.

\smallskip
We prove~\Eqref{point-testa+}. Let $a'=(b',c')\in\Web{\Tsem\sigma}$.
Let $\vec\zeta\in\Pcoh{\Tsem\Tnat}$ be the concatenation of the finite sequences $\vec\zeta^1,\vec\zeta^2\in\Pcoh{\Tsem\Tnat}$ such that the length of $\vec\zeta^1$ is $\Lenv{b}$ and the length of $\vec\zeta^2$ is $\Lena{c}$. 
By inductive hypothesis, for any $u\in\Pcoh{\Tsem\phi}$,
\begin{align*}
  \Fun{\Psem{\Testa a}}(\vec\zeta)u=  \Psem{\Testa a}\Prom{\left.\vec\zeta\right.}u=(\Psem{\Testv b}\Prom{\left.\vec\zeta^1\right.}u)_{\ast}\,(\Psem{\Testa c}\Prom{\left.\vec\zeta^2\right.}).
\end{align*}
Now, let $\epsilon>0$ such that $\epsilon\Bcanon{b'}\in\Pcoh{\Tsem{\phi}}$, we compute
\begin{align*}
\Fun{\Psem{\Testa a}}_{(b',c')}(\Prom{\left.\vec\zeta\right.})=\tfrac 1\epsilon  (\Psem{\Testa a}\Prom{\left.\vec\zeta\right.} \,\epsilon\Bcanon{b'})_{c'}=\tfrac 1\epsilon(\Psem{\Testv b}\Prom{\left.\vec\zeta^1\right.}\,\epsilon\Bcanon{b'})_{\ast}\,(\Psem{\Testa c}\Prom{\left.\vec\zeta^2\right.})_{c'}=\Fun{\Psem{\Testv b}}_{(b',\ast)}(\vec\zeta^1)\,\Fun{\Psem{\Testa c}}_{c'}(\vec\zeta^2).
\end{align*}
Therefore, by inductive hypothesis, $\Fun{\Psem{\Testa a}}_{(b',c')}$ depends on finitely many coefficients $\Lena a=\Lenv b+\Lena c$.
We compute using inductive hypothesis that 
\begin{align*}
\mon{\vec\zeta}{\Psem{\Testa a}}_{(b',c')}=\mon{\vec\zeta^1}{\Psem{\Testv b}}_{(b',\ast)}\,\mon{\vec\zeta^2}{\Psem{\Testa c}}_{c'}=\coeffv b\,\Kronecker{b}{b'}\,\coeffa c\,\Kronecker{c}{c'}.
\end{align*}
We conclude that $\mon{\vec\zeta}{\Psem{\Testa a}}=\coeffv b\,\coeffa c\,\Bcanon{b,c}=\coeffa a\,\Bcanon a$.

\smallskip
We prove~\Eqref{point-testa-}.
Let $a'=[(b'_0,c'_0),\dots,(b'_{k-1},c'_{k-1})])\in\Web{\Tsem{\EXCL\sigma}}$.
 Let  $\vec\zeta\in\Pcoh{\Tsem\Tnat}$ be the concatenation of the finite sequences $\vec\zeta^1,\vec\zeta^2\in\Pcoh{\Tsem\Tnat}$ such that the length of $\vec\zeta^1$ is $\Lena{b}$ and the length of $\vec\zeta^2$ is $\Lent{c}$. 
For any $w\in\Pcoh{\Tsem{\EXCL{(\LIMPL{\phi}\tau}})}$, we have:
\begin{align*}
  \Psem{\Testt a}\Prom{\left.\vec\zeta\right.} w
  = \Psem{\Testt c}\Prom{\left.\vec\zeta^2\right.}\Prom{\left(\Derel{\Tsem{\LIMPL\phi\tau}}(w)\Psem{\Testa b}\Prom{\left.\vec\zeta^1\right.}\right)}.
\end{align*}
Let $\epsilon>0$ such that $\epsilon\Bcanon{a'}\in\Pcoh{\Tsem{\EXCL{(\LIMPL\phi\tau})}}$, then
\begin{align*}
    \Fun{\Psem{\Testt a}}_{(a',\ast)}(\vec\zeta)=\tfrac 1\epsilon(\Psem{\Testt a}\Prom{\left.\vec\zeta\right.}\epsilon\Bcanon{a'})_\ast=\tfrac 1\epsilon(\Psem{\Testt c}\Prom{\left.\vec\zeta^2\right.}\Prom{\left(\Derel{\Tsem{\LIMPL\phi\tau}}(\epsilon\Bcanon{a'})\Psem{\Testa b}\Prom{\left.\vec\zeta^1\right.}\right)})_\ast.
  \end{align*}
By inductive hypothesis, we get that $\Fun{\Psem{\Testt a}}_{(a',\ast)}$ depends on  $\Lent a=\Lenv b+\Lent c$ coefficients.

Let now $u\in\Pcoh{\Psem{\LIMPL\phi\tau}}$, then by~Lemma~\ref{lemma:sem-values} $\Derel{\Tsem{(\LIMPL\phi\tau)}}(\Prom u)=u$ and we compute:
\begin{align*}
  \mon{\vec\zeta}{\Psem{\Testt a}}{\Prom u}=\mon{\vec\zeta^1}{\mon{\vec\zeta^2}{\Psem{\Testt c}}\Prom{\left({u\Psem{\Testa b}\Prom{\left.\vec\zeta^1\right.}}\right)}}\,.
\end{align*}
By inductive hypothesis, we have
$\mon{\vec\zeta^2}{\Psem{\Testt c}}\Prom{({u\Psem{\Testa b}\Prom{\left.\vec\zeta^1\right.}})}=\coefft
c\, (u\Psem{\Testa b}\Prom{\left.\vec\zeta^1\right.})_c$.
Moreover, notice that  $u\in\Pcoh{\Tsem{\LIMPL\phi\tau}}$, seen as a morphism in
$\PCOH(\Tsem\phi,\Tsem\tau)$ is linear, and there is $\epsilon>0$ such that
$\epsilon \Bcanon b\in\Pcoh{\Tsem\phi}$), so that we can apply $u$ to $\Bcanon b$. Now, by using inductive hypothesis, we get
that
$\mon{\vec\zeta^1}{u\Psem{\Testa b}}_c=(u\,\mon{\vec\zeta^1}{\Psem{\Testa
    b}})_c=
(u\,\coeffa b\,e_{b})_c=\coeffa b\,u_{(b,c)}$.
Therefore, we have
$\mon{\vec\zeta}{\Psem{\Testt a}}\Prom u=\coeffa b\,\coefft c\,u_{(b,c)}=\coefft
a\, u_{a}$.
\end{proof}

\begin{theorem}[Full Abstraction]\label{thm:fa}
\hspace{1cm}

  If $\TSEQ{}{M_1}{\sigma}$ and  $\TSEQ{}{M_2}{\sigma}$ satisfy $M_1\Rel\Obseq M_2$ then
  $\Psem{M_1}=\Psem{M_2}$.
\end{theorem}
\begin{proof}
  By contrapositive. Assume that
  $\Psem{M_1}\not=\Psem{M_2}$. There is $a\in\Web{\Tsem\sigma}$ such
  that $\Psem{M_1}_a\neq \Psem{M_2}_a$. Then by
  Lemma~\ref{lem:point-test},
  $\Psem{\ABST x{\EXCL\Tnat}{\LAPP{\LAPP{\Testt a}x}{\STOP{M_i}}}}$,
  for $i\in \{1,2\}$, are power series with different coefficients,
  namely the coefficients of the monomial
  $\Listbis \zeta 0{\Lent a-1}$ are $\coefft a\,\Psem{M_i}_a$ for
  $i\in\set{1,2}$ as
  $\Psem{\ABST x{\EXCL\Tnat}{\LAPP{\LAPP{\Testt
          a}x}{\STOP{M_i}}}}{\Prom{\left.\Vect\zeta\right.}}=\Psem{\Testt
    a}{\Prom{\left.\Vect\zeta\right.}}{\Prom{\Psem{M_i}}}$.
  There is $\vec\zeta=(\List\zeta0{\Lent a-1})\in\Pcoh{\Tsem\Tnat}$
  with $\zeta_i\in\mathbb Q\cap[0,1]$ such that
  $\Psem{\Testt a}\Prom{\left.\vec\zeta\right.}\Prom{\Psem{M_1}}\neq
  \Psem{\Testt a}\Prom{\left.\vec\zeta\right.}\Prom{\Psem{M_2}}$.
  Yet,
  $\Psem{\LAPP{\LAPP{\Testt a}{\STOP{\Ran{\Vect\zeta}}}}{\STOP{M_i}}}=
  \Psem{\Testt a}\Prom{\left.\vec\zeta\right.}({\Psem{M_i}})$.
  By Theorem~\ref{th:rel-ad-lemma}, we get that
  $\LAPP{\LAPP{\Testt a}{\STOP{\Ran{\Vect\zeta}}}}{\STOP{M_1}}$ and
  $\LAPP{\LAPP{\Testt a}{\STOP{\Ran{\Vect\zeta}}}}{\STOP{M_2}}$
  converge to $\ONELEM$ with different probabilities. It follows that
  $M_1\Rel{\not\Obseq}M_2$.
\end{proof}

\subsection*{Acknowledgements}

We would like to thank the referees for their many
useful and constructive comments and suggestions. 
We are also grateful to Michele Pagani and Rapha\"elle Crubill\'e for numerous and deep discussions.

This work has been partly funded by the  ANR Project RAPIDO ANR-14-CE25-0007, by the French-Chinese project
ANR-11-IS02-0002 and NSFC 61161130530 \emph{Locali}.

\bibliographystyle{plain}


\bibliography{biblio}

\newcommand{\SortNoop}[1]{}
\begin{thebibliography}{10}

\bibitem{Abramsky1998}
Samson Abramsky and Guy McCusker.
\newblock {\em Call-by-value games}, pages 1--17.
\newblock Springer Berlin Heidelberg, Berlin, Heidelberg, 1998.

\bibitem{AminiEhrhard15}
Shahin Amini and Thomas Ehrhard.
\newblock {On Classical PCF, Linear Logic and the MIX Rule}.
\newblock In Stephan Kreutzer, editor, {\em 24th {EACSL} Annual Conference on
  Computer Science Logic, {CSL} 2015, September 7-10, 2015, Berlin, Germany},
  volume~41 of {\em LIPIcs}, pages 582--596. Schloss Dagstuhl - Leibniz-Zentrum
  fuer Informatik, 2015.

\bibitem{phdBarker}
Tyler Barker.
\newblock {\em A monad for randomized algorithms}.
\newblock PhD thesis, Tulane University, 2016.

\bibitem{DanosEhrhard08}
Vincent Danos and Thomas Ehrhard.
\newblock {Probabilistic coherence spaces as a model of higher-order
  probabilistic computation}.
\newblock {\em {Information and Computation}}, 152(1):111--137, 2011.

\bibitem{Ehrhard16a}
Thomas Ehrhard.
\newblock {Call-By-Push-Value from a Linear Logic Point of View}.
\newblock In Peter Thiemann, editor, {\em Programming Languages and Systems -
  25th European Symposium on Programming, {ESOP} 2016, Held as Part of the
  European Joint Conferences on Theory and Practice of Software, {ETAPS} 2016,
  Eindhoven, The Netherlands, April 2-8, 2016, Proceedings}, volume 9632 of
  {\em {Lecture Notes in Computer Science}}, pages 202--228. Springer, 2016.

\bibitem{EhrhardGuerrieri16}
Thomas Ehrhard and Giulio Guerrieri.
\newblock The bang calculus: an untyped lambda-calculus generalizing
  call-by-name and call-by-value.
\newblock To appear, accepted at PPDP 2016, available on
  \url{https://www.irif.univ-paris-diderot.fr/~giuliog/}.

\bibitem{EhrhardPaganiTasson11}
Thomas Ehrhard, Michele Pagani, and Christine Tasson.
\newblock {The computational meaning of probabilistic coherent spaces}.
\newblock In {\em {Proceedings of the 26th Annual IEEE Symposium on Logic in
  Computer Science, LICS 2011, June 21-24, 2011, Toronto, Ontario, Canada}},
  pages 87--96. {IEEE Computer Society}, 2011.

\bibitem{EhrhardRegnier02}
Thomas Ehrhard and Laurent Regnier.
\newblock The differential lambda-calculus.
\newblock {\em {Theoretical Computer Science}}, 309(1-3):1--41, 2003.

\bibitem{EhrhardPaganiTasson14}
Thomas Ehrhard, Christine Tasson, and Michele Pagani.
\newblock {Probabilistic coherence spaces are fully abstract for probabilistic
  PCF}.
\newblock In Suresh Jagannathan and Peter Sewell, editors, {\em POPL}, pages
  309--320. {Association for Computing Machinery}, 2014.

\bibitem{EhrhardPaganiTasson15}
Thomas Ehrhard, Christine Tasson, and Michele Pagani.
\newblock {Full Abstraction for Probabilistic PCF}.
\newblock Technical report, {Preuves, Programmes et Syst\`emes}, 2015.
\newblock Submitted for publication to a journal.

\bibitem{EhrhardPaganiTasson18}
Thomas Ehrhard, Christine Tasson, and Michele Pagani.
\newblock {The cartesian closed category of measurable cones and stable,
  measurable functions}.
\newblock In {\em ACM SIGPLAN annual Symposium on Principles of Programming
  Languages}, 2018.
\newblock To appear.

\bibitem{Girard87}
Jean-Yves Girard.
\newblock Linear logic.
\newblock {\em {Theoretical Computer Science}}, 50:1--102, 1987.

\bibitem{Girard91a}
Jean-Yves Girard.
\newblock A new constructive logic: classical logic.
\newblock {\em {Mathematical Structures in Computer Science}}, 1(3):225--296,
  1991.

\bibitem{goubaultvarraca11}
Jean Goubault-Larrecq and Daniele Varacca.
\newblock Continuous random variables.
\newblock In {\em LICS}, pages 97--106. IEEE Computer Society, 2011.

\bibitem{Griffin90}
Timothy~G. Griffin.
\newblock The formulae-as-types notion of control.
\newblock In {\em Proceedings of the 17h ACM Symposium on Principles of
  Programming Languages (POPL)}, pages 47--57. {Association for Computing
  Machinery}, January 1990.

\bibitem{HeunenKSY17}
Chris Heunen, Ohad Kammar, Sam Staton, and Hongseok Yang.
\newblock A convenient category for higher-order probability theory.
\newblock In {\em 32nd Annual {ACM/IEEE} Symposium on Logic in Computer
  Science, {LICS} 2017, Reykjavik, Iceland, June 20-23, 2017}, pages 1--12.
  {IEEE} Computer Society, 2017.

\bibitem{JonesPlotkin89}
Claire Jones and Gordon Plotkin.
\newblock A probabilistic powerdomains of evaluation.
\newblock In {\em Proceedings of the 4th Annual IEEE Symposium on Logic in
  Computer Science}. IEEE Computer Society, 1989.

\bibitem{JungTix98}
Achim Jung and Regina Tix.
\newblock The troublesome probabilistic powerdomain.
\newblock {\em {Electronic Notes in Theoretical Computer Science}}, 13:70--91,
  1998.

\bibitem{KeimelP16}
Klaus Keimel and Gordon~D. Plotkin.
\newblock Mixed powerdomains for probability and nondeterminism.
\newblock {\em Logical Methods in Computer Science}, 13(1), 2017.

\bibitem{LaurentRegnier03}
Olivier Laurent and Laurent Regnier.
\newblock {About Translations of Classical Logic into Polarized Linear Logic}.
\newblock In {\em 18th {IEEE} Symposium on Logic in Computer Science {(LICS}
  2003), 22-25 June 2003, Ottawa, Canada, Proceedings}, pages 11--20. {IEEE}
  Computer Society, 2003.

\bibitem{LevyP02}
Paul~Blain Levy.
\newblock {Adjunction Models For Call-By-Push-Value With Stacks}.
\newblock {\em {Electronic Notes in Theoretical Computer Science}},
  69:248--271, 2002.

\bibitem{LevyP04}
Paul~Blain Levy.
\newblock {\em {Call-By-Push-Value: {A} Functional/Imperative Synthesis}},
  volume~2 of {\em Semantics Structures in Computation}.
\newblock Springer-Verlag, 2004.

\bibitem{MarzRohrStreicher99}
Michael Marz, Alexander Rohr, and Thomas Streicher.
\newblock {Full Abstraction and Universality via Realisability}.
\newblock In {\em 14th Annual {IEEE} Symposium on Logic in Computer Science,
  Trento, Italy, July 2-5, 1999}, pages 174--182. IEEE Computer Society, 1999.

\bibitem{Mellies09}
Paul-Andr\'e Melli\`es.
\newblock {Categorical semantics of linear logic}.
\newblock {\em Panoramas et Synth\`eses}, 27, 2009.

\bibitem{Mislove16}
Michael~W. Mislove.
\newblock Domains and random variables.
\newblock {\em CoRR}, abs/1607.07698, 2016.

\bibitem{Pitts93}
Andrew~M. Pitts.
\newblock {Computational Adequacy via "Mixed" Inductive Definitions}.
\newblock In Stephen~D. Brookes, Michael~G. Main, Austin Melton, Michael~W.
  Mislove, and David~A. Schmidt, editors, {\em Mathematical Foundations of
  Programming Semantics, 9th International Conference, New Orleans, LA, USA,
  April 7-10, 1993, Proceedings}, volume 802 of {\em {Lecture Notes in Computer
  Science}}, pages 72--82. Springer, 1993.

\bibitem{Plotkin77}
Gordon Plotkin.
\newblock {LCF} considered as a programming language.
\newblock {\em {Theoretical Computer Science}}, 5:223--256, 1977.

\bibitem{Rennela16}
Mathys Rennela.
\newblock Convexity and order in probabilistic call-by-name {FPC}, 2016.

\bibitem{Saheb-Djahromi80}
Nasser Saheb-Djahromi.
\newblock {CPO}s of measures for nondeterminism.
\newblock {\em {Theoretical Computer Science}}, 12(1):19--37, 1980.

\bibitem{PlotkinSimpson00}
Alex~K. Simpson and Gordon~D. Plotkin.
\newblock {Complete Axioms for Categorical Fixed-Point Operators}.
\newblock In {\em 15th Annual {IEEE} Symposium on Logic in Computer Science,
  Santa Barbara, California, USA, June 26-29, 2000}, pages 30--41. IEEE
  Computer Society, 2000.

\bibitem{StatonYWHK16}
Sam Staton, Hongseok Yang, Frank Wood, Chris Heunen, and Ohad Kammar.
\newblock Semantics for probabilistic programming: higher-order functions,
  continuous distributions, and soft constraints.
\newblock In Martin Grohe, Eric Koskinen, and Natarajan Shankar, editors, {\em
  Proceedings of the 31st Annual {ACM/IEEE} Symposium on Logic in Computer
  Science, {LICS} '16, New York, NY, USA, July 5-8, 2016}, pages 525--534.
  {Association for Computing Machinery}, 2016.

\bibitem{TixKP09a}
Regina Tix, Klaus Keimel, and Gordon~D. Plotkin.
\newblock Semantic domains for combining probability and non-determinism.
\newblock {\em Electr. Notes Theor. Comput. Sci.}, 222:3--99, 2009.

\bibitem{Vaux08}
Lionel Vaux.
\newblock The algebraic lambda-calculus.
\newblock {\em {Mathematical Structures in Computer Science}},
  19(5):1029--1059, 2009.

\end{thebibliography}

\end{document}